\documentclass{arxiv}
\usepackage{amsmath}
\usepackage{amssymb}
\usepackage{amsfonts}
\usepackage{amstext}
\usepackage{amsopn}
\usepackage{amsxtra}
\usepackage{graphicx}
\usepackage[latin1]{inputenc}
\usepackage{dsfont}
\newcommand\R{{\ensuremath {\mathbb R} }}
\newcommand\C{{\ensuremath {\mathbb C} }}

\newcommand\1{{\ensuremath {\mathds 1} }}
\renewcommand\phi{\varphi}

\newcommand{\gH}{\mathfrak{H}}
\newcommand{\gX}{\mathfrak{X}}
\newcommand{\bH}{\mathbb{H}}

\newcommand{\gS}{\mathfrak{S}}
\newcommand{\wto}{\rightharpoonup}
\newcommand{\gto}{\underset{\rm g}{\rightharpoonup}}
\renewcommand{\to}{\rightarrow}

\newcommand{\cN}{\mathcal{N}}
\newcommand{\cS}{\mathcal{S}}
\newcommand{\cD}{\mathcal{D}}
\newcommand{\cA}{\mathcal{A}}

\newcommand{\cB}{\mathcal{B}}

\newcommand{\sgn}{{\rm sgn}}
\newcommand{\cF}{\mathcal F}

\newcommand{\cT}{\mathcal T}

\newcommand{\tr}{{\rm tr}\,}
\newcommand{\cE}{\mathcal{E}}

\newcommand{\cK}{\mathcal{K}}
\newcommand{\bT}{\mathbb{T}}

\newcommand\ii{{\ensuremath {\infty}}}
\newcommand\pscal[1]{{\ensuremath{\left\langle #1 \right\rangle}}}
\newcommand{\norm}[1]{ \left| \! \left| #1 \right| \! \right| }

\renewcommand{\tr}{{\rm Tr} }

\renewcommand{\leq}{\leqslant}
\renewcommand{\geq}{\geqslant}
\newcommand{\cFN}{\cF^{\leq N}}
\newcommand{\dqed}{\hfill$\diamond$}
\newcommand{\rank}{{\rm rank}}

\begin{document}
\date{December 2, 2010}
\title{\Large Geometric methods for nonlinear\\ \medskip many-body quantum systems}
\authormark{Mathieu LEWIN}
\runningtitle{Geometric methods for nonlinear many-body quantum systems}

\author{Mathieu LEWIN}
\address{CNRS and Laboratoire de Mathématiques (CNRS UMR 8088)\\ Universit{\'e} de Cergy-Pontoise, 95 000 Cergy-Pontoise - France.\\ Email: \email{Mathieu.Lewin@math.cnrs.fr}}

\maketitle

\bigskip

\begin{abstract}
Geometric techniques have played an important role in the seventies, for the study of the spectrum of many-body Schrödinger operators. In this paper we provide a formalism which also allows to study nonlinear systems.

We start by defining a weak topology on many-body states, which appropriately describes the physical behavior of the system in the case of lack of compactness, that is when some particles are lost at infinity. We provide several important properties of this topology and use them to write a simple proof of the famous HVZ theorem in the repulsive case.
In a second step we recall the method of geometric localization in Fock space as proposed by Derezi\'nski and Gérard,  and we relate this tool to our weak topology.

We then provide several applications. We start by studying the so-called finite-rank approximation which consists in imposing that the many-body wavefunction can be expanded using finitely many one-body functions. We thereby emphasize geometric properties of Hartree-Fock states and prove nonlinear versions of the HVZ theorem, in the spirit of works of Friesecke.

In the last section we study translation-invariant many-body systems comprising a nonlinear term, which effectively describes the interactions with a second system. As an example, we prove the existence of the multi-polaron in the Pekar-Tomasevich approximation, for certain values of the coupling constant.

\medskip 

%

\noindent{\scriptsize\copyright~2010 by the author. This paper may be reproduced, in its entirety, for non-commercial~purposes.\\ Final version to appear in \textit{J. Func. Anal.}}
\end{abstract}


\bigskip

\tableofcontents

\section*{Introduction}
\addcontentsline{toc}{section}{Introduction}

A system of $N$ (spinless) quantum particles is usually described by an energy functional $\Psi\mapsto \cE(\Psi)\in\R$ where $\Psi$ is a normalized function of the $N$-body space
\begin{equation}
\gH^N:=\bigotimes_{n=1}^NL^2(\R^d)\simeq L^2((\R^d)^N).
\end{equation}
Here $d$ is the dimension of the space in which the $N$ particles evolve, that is $d=3$ in the physical case. If the particles are indistinguishable bosons (resp. fermions) it is additionally assumed that $\Psi$ is symmetric (resp. antisymmetric) with respect to exchanges of variables $(x_1,...,x_N)\in(\R^d)^N$. 

In the simplest case the energy $\cE$ is the quadratic form associated with a self-adjoint operator on $\gH^N$. For nonrelativistic particles interacting with a two-body potential $W$ and submitted to an external potential $V$, the corresponding \emph{$N$-body Hamiltonian} reads
\begin{equation}
H^V(N)=\sum_{j=1}^N\left(-\frac{\Delta_{x_j}}{2}+V(x_j)\right)+\sum_{1\leq k<\ell\leq N}W(x_k-x_\ell).
 \label{eq:H_V_N_intro}
\end{equation}
The study of the properties of self-adjoint operators of this form has a long history \cite{HunSig-00} and it is certainly one of the most significant successes of mathematical physics in the past decades. Of particular interest is the spectrum of $H^V(N)$.

The advent of \emph{geometric methods} in the late seventies has been particularly important. By `geometric' it is usually meant the use of clever partitions of unity in configuration space in order to relate local properties of $H^V(N)$ (seen as a partial differential operator) and spectral properties. Initiated in the sixties by Zhislin \cite{Zhislin-60} and Jörgens and Weidmann \cite{JorWei-73}, the systematic use of geometric ideas in Schrödinger operators theory really started in 1977 with the works of Enss \cite{Enss-77}, Deift and Simon \cite{DeiSim-77}, and Simon \cite{Simon-77}. It was then further developed by Morgan \cite{Morgan-79}, Morgan and Simon \cite{MorSim-80}, and Sigal \cite{Sigal-82,Sigal-83,Sigal-84}. For a review of these techniques we refer for instance to \cite{ReeSim4,CycFroKirSim-87,HunSig-00}.

A famous example of the use of geometric methods is the so-called HVZ Theorem of Zhislin \cite{Zhislin-60}, Van Winter \cite{VanWinter-64} and Hunziker \cite{Hun-66}. Under suitable decay assumptions on $V$ and $W$, it relates the bottom of the essential spectrum of $H^V(N)$ to the ground state energy of systems with less particles:
\begin{equation}
\inf\sigma_{\rm ess}(H^V(N))=\inf\left\{E^V(N-k)+E^0(k),\ k=1,...,N\right\},
\label{eq:HVZ_intro}
\end{equation}
where $E^V(N):=\inf\sigma(H^V(N))$ is the ground state energy for $N$ particles. Physically this result says that in order to reach the bottom of the essential spectrum one has to remove $k$ particles from the system and place them at infinity. The total energy is then the sum of the ground state energy $E^V(N-k)$ of the $N-k$ remaining particles plus the energy $E^0(k)$ of the $k$ particles at infinity. The number $k$ of particles to extract is chosen such as to minimize the total energy obtained by this procedure. A consequence of \eqref{eq:HVZ_intro} is that $E^V(N)$ is an isolated eigenvalue if and only if
\begin{equation}
E^V(N)<E^V(N-k)+E^0(k),\quad \forall k=1,...,N.
\label{eq:HVZ_intro2}
\end{equation}

Although physically quite natural, the HVZ formula \eqref{eq:HVZ_intro} is mathematically not obvious, in particular because the three problems corresponding to having $N$, $k$ and $N-k$ particles are posed on the different Hilbert spaces $\gH^N$, $\gH^{k}$ and $\gH^{N-k}$. When proving \eqref{eq:HVZ_intro}, geometric methods indeed make a crucial use of the fact that the many-body space has the structure of a \emph{tensor product}, that is $\gH^N\simeq \gH^{N-k}\otimes\gH^k$.

\medskip

Linear problems are not the only possible ones occurring in the study of many-body quantum systems. Indeed, most numerical methods used by physicists and chemists resort to nonlinear models. Sometimes the energy is kept linear but the set of states is reduced by assuming that the wavefunctions $\Psi$ belong to a well-chosen manifold. In some other cases it is convenient to modify the many-body energy $\cE$ by adding nonlinear empirical terms in order to account for involved physical effects which are too complicated to describe in a precise manner.

Nonlinear methods also have a long history, in particular within the field of partial differential equations. Loosely speaking, a typical question is to understand the behavior of sequences of functions $\{\phi_n\}$ (say in $L^2(\R^d)$), in particular in the case of lack of compactness, that is when $\phi_n\wto\phi$ weakly in $L^2$ but $\phi_n\nrightarrow\phi$ strongly. The sequence $\{\phi_n\}$ can be a minimizing sequence of some variational problem or a Palais-Smale sequence \cite{Struwe} (in these cases the goal is often to prove by contradiction that it must converge strongly). Or it can be the solution of a time-dependent equation, which experiments a dispersive or a blow-up behavior in finite or infinite time (in this case lack of compactness has some physical reality).

The first to tackle such issues on a specific example were Sacks and Uhlenbeck \cite{SacUhl-81} in 1981 who dealt with a concentration phenomenon for harmonic maps. Brezis and Nirenberg \cite{BreNir-83} then faced similar difficulties for some elliptic partial differential equations with a critical Sobolev exponent. In 1983, Lieb proved in \cite{Lieb-83} a useful lemma dealing with lack of compactness due to translations in the locally compact case. A general method for dealing with locally compact problems was published by Lions \cite{Lions-84,Lions-84b} in 1984 under the name ``concentration-compactness''. Later in 1984-85, Struwe \cite{Struwe-84} and, independently, Brezis and Coron \cite{BreCor-85} have provided the first  ``bubble decompositions', whereas Lions adapted his concentration-compactness method to the nonlocal case \cite{Lions-85a,Lions-85b}. For a review of all these techniques, we refer for instance to \cite{Struwe}.

When studying the compactness of minimizing sequences for a variational problem of the general form
$$I(N)=\inf_{\int_{\R^d}|\phi|^2=N}\cE(\phi),$$
a useful argument is to rely on so-called \emph{binding inequalities}
\begin{equation}
I(N)<I(N-\lambda)+I^0(\lambda),\quad \forall 0<\lambda\leq N,
\label{eq:binding_intro} 
\end{equation}
where $I^0(N)$ is the ground state energy when the system is sent to infinity (that is when all the local terms have been dropped in the energy $\cE$). Imagine that one can prove that a non-compact minimizing sequence $\{\phi_n\}$ would necessarily split into pieces in such a way that the total energy becomes the sum of the energies of these pieces. Then an energetic inequality like \eqref{eq:binding_intro} yields a contradiction and implies that all minimizing sequences must be compact. Arguments of this type are ubiquitous in studies of nonlinear minimization problems.

The formal link between the HVZ formula \eqref{eq:HVZ_intro} and binding inequalities of the form of \eqref{eq:binding_intro} has been known for a long time. There are important differences, however. In the HVZ case one has a \emph{quantized} inequality \eqref{eq:HVZ_intro2} in which only an integer number of particles can escape to infinity. On the contrary the binding inequality \eqref{eq:binding_intro} is \emph{not quantized} since in $L^2(\R^d)$ the sequence $\{\phi_n\}$ can split in pieces having an arbitrary mass. Vaguely speaking, this comes from the fact that in the case of lack of compactness, $\phi_n$ usually behaves as a sum of functions whereas an $N$-body wavefunction is rather a tensor product.

\medskip

The goal of this paper is to present a theory which combines nonlinear and geometric techniques, with the purpose to study some many-body systems involving nonlinear effects. A first attempt in this direction was already made by Friesecke in his paper  \cite{Friesecke-03} on multiconfiguration methods, a work which partly inspired the present paper. However, instead of concentrating only on some specific examples, a large part of this article (Sections~\ref{sec:geom_CV} and~\ref{sec:geom_loc}) is devoted to the presentation of a simple but general theory which, we hope, will be reusable in many other situations. We apply it to some nonlinear models in Sections~\ref{sec:HF_MCSCF} and~\ref{sec:nonlinear}.

In this work, we are particularly interested in finding an appropriate description of the possible lack of compactness of many-body wavefunctions. As we now explain, usual methods of nonlinear analysis are rather inefficient in this respect. Consider for instance a sequence of  two-body wavefunctions of the form:
\begin{equation}
\Psi_n=\phi\otimes\phi_n,
\label{eq:example_Psi_n} 
\end{equation}
that is $\Psi_n(x_1,x_2)=\phi(x_1)\phi_n(x_2)$, with $\phi,\phi_n\in L^2(\R^d)$. We assume that $\phi_n\wto0$ weakly in $L^2(\R^d)$, hence we may think of $\Psi_n$ as describing a system of two particles, one in the fixed state $\phi$ and the other one `escaping to infinity'. It is then easily verified that
$$\Psi_n\wto0\quad\text{weakly in $L^2(\R^d)\otimes L^2(\R^d)\simeq L^2((\R^d)^2)$,}$$
which suggests that looking at weak limits of two-body wavefunctions does not say much on the real behavior of the system.
We would rather like to have, for obvious physical reasons, that
\begin{equation}
\text{``}\Psi_n\wto\phi\text{''}
\label{eq:want_CV_intro} 
\end{equation}
since one particle is lost and the other one stays in the one-particle state $\phi$. However this does not make much sense as such, since $\Psi_n\in L^2(\R^d)\otimes L^2(\R^d)$ and $\phi\in L^2(\R^d)$ live in different Hilbert spaces. 

In Section~\ref{sec:geom_CV} we introduce a very natural topology on many-body states, which we call \emph{geometric topology}, and for which \eqref{eq:want_CV_intro} is actually correct. The geometric topology is very different from the usual weak topology (as can already be seen from the fact that $\Psi_n\wto0$ weakly). It is however the one which is physically relevant for many-body systems.

Let us vaguely explain how the geometric topology is defined. As is suggested by \eqref{eq:want_CV_intro}, even if we start with a sequence of states containing $N$ particles (in the $N$-body space $\gH^N$), we have to allow limits in spaces with less particles. All the particles could even be lost in the studied process, in which case we would end up with the vacuum. For this reason, the behavior of $N$-body states must be studied in the so-called \emph{truncated Fock space}
\begin{equation}
\cFN:=\C\oplus\gH^1\oplus\cdots \oplus\gH^N
\label{eq:truncated_Fock_space_intro}
\end{equation}
which gathers all the spaces of $k$ particles, with $0\leq k\leq N$. As we shall see on specific examples, it is also natural to allow a geometric limit which is a \emph{mixed state}, even when the sequence is only made of pure states. 
Let us recall that a mixed state $\Gamma$ on $\cFN$ is a trace-class self-adjoint operator such that $\Gamma\geq0$ and $\tr_{\cFN}(\Gamma)=1$. A pure state is a rank-one projector, $\Gamma=|\Phi\rangle\langle\Phi|$ with $\Phi\in\cFN$ (for instance $\Phi=0\oplus\cdots\oplus0\oplus\Psi$ in the case of a pure $N$-body state $\Psi\in\gH^N$).

The geometric topology on mixed states on $\cFN$ is defined by means of the weak topologies of all the corresponding density matrices, which are specific marginals (partial traces) reflecting the tensor product structure of the ambient Hilbert space (hence the name `geometric'). The definition of the density matrices is recalled in Section~\ref{sec:notation} below. In particular we say that $\Gamma_n\wto_g\Gamma$ geometrically when all the density matrices of $\Gamma_n$  converge to that of $\Gamma$, weakly--$\ast$ in the trace class. 
Let us emphasize that the geometric limit $\Gamma$ is always a \emph{state}, that is it satisfies $\tr(\Gamma)=1$. There is never any loss in the trace norm when passing to geometric limits.

For instance the one-body density matrix of our two-body sequence $\{\Psi_n\}$ in \eqref{eq:example_Psi_n} is the operator acting on $L^2(\R^d)$
$$\Gamma^{(1)}_n=|\phi\rangle\langle\phi|+|\phi_n\rangle\langle\phi_n|$$
(we assume for simplicity that $\phi\perp\phi_n$ for all $n$). By the weak convergence of $\phi_n\wto0$, it holds
$$\Gamma^{(1)}_n\wto |\phi\rangle\langle\phi| \quad \text{weakly--$\ast$}.$$
The operator $|\phi\rangle\langle\phi|$ is precisely the one-body density matrix of the one-body state $\phi\in\gH^1$.
We indeed have that 
$$0\,\oplus\,0\,\oplus\,|\phi\otimes\phi_n\rangle\langle\phi\otimes\phi_n|\;\gto \;0\,\oplus\,|\phi\rangle\langle\phi|\,\oplus\,0\qquad \text{geometrically in $\cF^{\leq 2}$,}$$
which is the precise mathematical meaning that we can give to \eqref{eq:want_CV_intro}. 

Our weak topology is the restriction to states on $\cFN$ of a well-known weak--$\ast$ topology associated with the CAR/CCR algebra (Remark~\ref{rmk:Clifford}). But, to our knowledge, the usefulness of this notion of convergence for many-body problems has never been pointed out in the literature. As will be seen on several examples in this work, it is however the most natural weak topology for many-body states. It is a crucial notion when strong convergence does not hold \emph{a priori}, that is in the case of possible  lack of compactness.

In Section~\ref{sec:prop_geom_CV} we proceed to give important properties of geometric convergence. We start by showing that the set of states is compact for the geometric topology in Lemma~\ref{lem:compact}. This means that any sequence of states $\{\Gamma_n\}$ on the truncated Fock space $\cFN$ has a subsequence such that $\Gamma_{n_k}\wto_g\Gamma$ geometrically. This result is very important in applications. We then show in Lemma~\ref{lem:N_wlsc} that strong convergence is equivalent to the conservation of the total average particle number.

We illustrate the use of our theory in Section~\ref{sec:HVZ_wlsc}: We consider an $N$-body Hamiltonian of the form of \eqref{eq:H_V_N_intro} with $W\geq0$ and we show that, in contrast with the usual weak topology of $\gH^N$, the associated quantum energy is lower semi-continuous for the geometric topology. This enables us to provide a very simple proof of the HVZ Theorem, in this particular case.

\medskip

Equipped with a new weak topology, we then need a second important notion: \emph{geometric localization}. As we have already mentionned, localization has always played an important role in the study of Schrödinger operators. As we want to find out where are the particles which stay and where do go those which escape to infinity, we need to be able to describe the state of our system in a given domain $D\subset\R^d$.

If we think of a one-body state $\phi\in L^2(\R^d)$, then the corresponding localized state in a domain $D$ should clearly be described by the function $\1_D\phi$. However, $\1_D\phi$ is in general not a state since $\int_{D}|\phi|^2<1$, except when $\phi$ has its support in $D$. Having removed what is outside $D$ corresponds, in our language, to the vacuum state. Thus the localized state should rather be
\begin{equation}
\left(1-\int_{D}|\phi|^2\right)\,\oplus\,|\1_D\phi\rangle\langle\1_D\phi|
\label{eq:localized_one_body_intro} 
\end{equation}
in the truncated Fock space $\cF^{\leq 1}$.

The correct notion of localization of any mixed state of $\cFN$ which generalizes \eqref{eq:localized_one_body_intro} was introduced by Derezi\'nski and Gérard  \cite{DerGer-99} in the context of Quantum Field Theory. It is even possible to define a localization with respect to any operator $B$ on $L^2(\R^d)$ such that $BB^*\leq1$, not only for $B$ the multiplication operator by the characteristic function $\1_D$ (this is in particular useful when dealing with smooth cut-off functions). In Section~\ref{sec:geom_loc} we recall the definition of geometric localization in our context and we provide several of its properties.
Of particular interest is the fact that if $\Gamma_n\wto_g\Gamma$ geometrically and $\{\Gamma_n\}$ has a bounded kinetic energy, then one gets a strong convergence of the localized states in any bounded domain $D$. This generalizes the well-known Rellich compactness embedding theorem in Sobolev spaces, to the setting of many-body states and geometric topology.

In Section~\ref{sec:HVZ_general}, using both geometric convergence and localization we are able to provide a simple proof of the HVZ theorem in the general setting (when $W$ has no particular sign), which particularly enlightens a crucial but simple geometric property of $N$-body functions, see Equation \eqref{eq:relation_trace_localized_N_body_fn} below. It is not our intention to pretend that our proof of the HVZ theorem is better than any of the other existing proofs. We rather aim at accustoming the reader to the techniques that we will use for nonlinear models in Sections~\ref{sec:HF_MCSCF} and~\ref{sec:nonlinear}, and for which usual linear methods are inappropriate.

\medskip

We turn to the study of nonlinear models in Sections~\ref{sec:HF_MCSCF} and~\ref{sec:nonlinear}.

In Section~\ref{sec:HF_MCSCF} we study the so-called \emph{finite rank approximation} in which one restricts to $N$-body states which can be expanded using a finite number of (unknown) one-body orbitals $\phi_1,...,\phi_r$. In the bosonic case we obtain the Hartree model for $r=1$. In the fermionic case the Hartree-Fock theory \cite{LieSim-77} is obtained when $r=N$ (the number of particles) whereas $r>N$ leads to \emph{multiconfiguration methods} \cite{Friesecke-03,Lewin-04a}. Despite the fact that these methods are essentials tools of quantum physics and chemistry, their geometric properties have deserved little interest in the literature so far. In \cite{Friesecke-03}, Friesecke was, to our knowledge, the first to consider both the Hartree-Fock and the multiconfiguration theories as real $N$-body models and to use geometric techniques in order to derive nonlinear HVZ-type results. 

Our goal is to emphasize geometric properties of finite-rank states, that is to find what can be said on geometric limit points of sequences or on geometric localization of such special states. For instance we show in Section~\ref{sec:geom_prop} that the geometric limit of a sequence of pure Hartree-Fock states is always a convex combination of pure Hartree-Fock states, see Example~\ref{ex:geom_limit_HF} below. Using such properties, we are able to provide a simple proof of Friesecke's results, as well as to derive other theorems. For instance in Theorem~\ref{thm:HF_translation_invariant} below, we prove a nonlinear HVZ-type result for a \emph{translation-invariant} Hartree-Fock theory, combining ideas of Lions \cite{Lions-84,Lions-84b} and geometric techniques. This result is in the same spirit as what was done for neutron stars in a recent collaboration with Lenzmann \cite{LenLew-10}.

In Section~\ref{sec:nonlinear} we study another kind of nonlinear models where all possible many-body states are considered but nonlinear effective terms are added to the quantum energy $\cE$ in order to describe some specific physical effects. To be more precise we concentrate on translation-invariant models of the form
\begin{equation}
\cE(\Psi)=\pscal{\Psi,H^0(N)\Psi}+F(\rho_\Psi) 
\label{eq:form_nonlinear_intro}
\end{equation}
where $F$ is a concave nonlinear function of the charge density $\rho_\Psi$, and $H^0(N)$ is the $N$-body Hamiltonian \eqref{eq:H_V_N_intro} with $V=0$. In practice the purpose of the nonlinear term $F(\rho_\Psi)$ is to model the interaction of the $N$ particles with a second complicated system. For instance we consider in Section~\ref{sec:polaron} the multi-polaron in the Pekar-Tomasevich approximation. This is a system of $N$ nonrelativistic electrons with an effective nonlinear term
$$F(\rho_\Psi)=-\frac\alpha2\int_{\R^3}\int_{\R^3}\frac{\rho_\Psi(x)\,\rho_\Psi(y)}{|x-y|}\,dx\,dy$$
modeling interactions with the phonons of a polar crystal in the regime of strong coupling. We show the existence of bound states for all $\alpha> \tau_c(N)$ where $\tau_c(N)<1$, which covers the physical case. This complements recent results of \cite{FraLieSeiTho-10a,FraLieSeiTho-10b}.

\medskip

The paper is organized as follows. In Section~\ref{sec:notation} we provide necessary notations and some preliminary results. The reader at ease with the concepts of Fock space, creation and annihilation operators and density matrices may want to skip most of the material of Section \ref{sec:notation}. Of importance is Lemma~\ref{lem:density_matrices} which provides crucial properties of density matrices for states on a truncated Fock space $\cFN$. Section~\ref{sec:geom_CV} is devoted to the definition and the derivation of important properties of the geometric topology and convergence. This is followed by a proof of the HVZ theorem in the repulsive case. In Section~\ref{sec:geom_loc} geometric localization is defined and the general HVZ theorem is proved. Sections~\ref{sec:HF_MCSCF} and~\ref{sec:nonlinear} are respectively devoted to the study of the finite-rank approximation, and of nonlinear systems of the form \eqref{eq:form_nonlinear_intro}.

For the sake of clarity we usually do not state the most general results and rather favor some chosen applications. Many of our theorems can be generalized in several directions.

\bigskip

{\small \noindent\textbf{Acknowledgment.} I started this work after several interesting discussions with Enno Lenzmann. He was the first to draw my attention to the multi-polaron model, which was the starting point of this article. I am also indebted to Vladimir Georgescu for many stimulating discussions.

The research leading to these results has received funding from the European Research Council under the European Community's Seventh Framework Programme (FP7/2007--2013 Grant Agreement MNIQS no. 258023).}

\section{Notation and preliminaries}\label{sec:notation}
We start by fixing some important notation and vocabulary, as well as by providing some preliminary results that will be useful throughout the paper. The reader acquainted with Fock spaces can jump to Section~\ref{sec:def_DM} where density matrices are defined and some of their important properties derived.

\subsection{Spaces and algebras}
For a (separable) Hilbert space $\gH$, we denote by $\cB(\gH)$ and $\cK(\gH)$ the algebras of, respectively, bounded and compact operators on $\gH$. The Schatten space $\gS_p(\gH)\subset\cK(\gH)$ is defined \cite{Simon-79} by requiring that
$$\norm{A}_{\gS_p(\gH)}:=\tr\left(|A|^p\right)^{1/p}<\ii,$$
with $|A|=\sqrt{A^*A}$.
Operators in $\gS_1(\gH)$ have a well-defined trace $\tr(A)=\sum_{i}\pscal{f_i,A f_i}$ (for any orthonormal basis $\{f_i\}$ of $\gH$). Operators in $\gS_2(\gH)$ are called Hilbert-Schmidt. We recall that \cite{Simon-79}
\begin{equation}
\big(\cK(\gH)\big)'=\gS_1(\gH) \quad\text{and}\quad \big(\gS_1(\gH)\big)'=\cB(\gH).
\label{eq:dual_compact} 
\end{equation}
Since $\cK(\gH)$ is separable when $\gH$ is separable, \eqref{eq:dual_compact} means that any bounded sequence $\{\Gamma_n\}$ in $\gS_1(\gH)$ has a subsequence which converges weakly--$\ast$ in the sense that $\lim_{n\to\ii}\tr_\gH(\Gamma_nK)=\tr_\gH(\Gamma K)$ for all $K\in\cK(\gH)$. The same holds for bounded sequences in $\cB(\gH)$, with $\cK(\gH)$ replaced by $\gS_1(\gH)$.

In the whole paper we fix as space for one quantum particle $\gH=L^2(\R^d)$. We could as well work in a domain $\Omega$ with appropriate boundary conditions, use a discrete model, or even, for most of our results, take an abstract Hilbert space. These obvious generalizations are left to the reader for shortness.

Similarly, for simplicity we almost always restrict ourselves to the case of quantum systems made of \emph{one kind of indistinguishable particles (fermions or bosons) without spin}. Most results can be easily generalized to the case of several kinds of particles having internal degrees of freedom. The space for \emph{$N$ indistinguishable fermions} is the antisymmetric tensor product
$$\gH_a^N:=\bigwedge_1^N\gH=L^2_a((\R^d)^N)$$
consisting of wavefunctions $\Psi$ which are antisymmetric with respect to exchanges of variables:
$\Psi(x_1,...,x_i,...,x_j,...,x_N)=-\Psi(x_1,...,x_j,...,x_i,...,x_N)$, with $x_k\in\R^d$ for $k=1,...,N$.
The space for \emph{$N$ indistinguishable bosons} is the symmetric tensor product
$$\gH_s^N:=\bigvee_1^N\gH=L^2_s((\R^d)^N)$$
consisting of wavefunctions $\Psi$ which are symmetric with respect to exchanges of variables:
$\Psi(x_1,...,x_i,...,x_j,...,x_N)=\Psi(x_1,...,x_j,...,x_i,...,x_N)$, with $x_k\in\R^d$ for $k=1,...,N$.

The corresponding fermionic or bosonic Fock space is denoted as 
$$\cF_{a/s}=\C\oplus\bigoplus_{N\geq1}\gH_{a/s}^N.$$
Saying differently, it is the space composed of sequences of the form $\Psi=(\psi^0,\psi^1,\psi^2,...)\in\C\times\gH\times\gH^2_{a/s}\times\cdots$ satisfying the constraint that
$$\norm{\Psi}_{\cF_{a/s}}^2:=\sum_{n\geq0}\norm{\psi^n}_{\gH^n_{a/s}}^2<\ii.$$
It is a Hilbert space when endowed with the scalar product $$\pscal{\Psi_1,\Psi_2}_{\cF_{a/s}}=\sum_{n\geq0}\pscal{\psi_1^n,\psi_2^n}_{\gH^n_{a/s}}.$$ 
The \emph{vacuum state} is by convention defined as $\Omega:=(1,0,0,...)\in\cF_{a/s}$.

As we consider $N$-body systems, we most always work in the `truncated' Fock space 
\begin{equation}
\boxed{\cFN_{a/s}:=\C\oplus\bigoplus_{n=1}^N\gH^n_{a/s}}
\label{eq:cut-Fock-space}
\end{equation}
which we identify to a closed subspace of $\cF_{a/s}$. Similarly, any $N$-body vector of $\gH^N_{a/s}$ can be viewed as a vector of $\cFN_{a/s}$ or of $\cF_{a/s}$. As we explain later in Section~\ref{sec:geom_CV}, the `geometric' limit of a sequence $(\psi_n)\subset\gH^N_{a/s}$ will always live in the truncated Fock space $\cFN_{a/s}$.

For $\psi_1\in\gH_a^{N_1}$ and $\psi_2\in \gH_a^{N_2}$, we define the antisymmetric tensor product $\psi_1\wedge\psi_2\in \gH_a^{N_1+N_2}$ as follows:
\begin{multline}
\psi_1\wedge\psi_2(x_1,...,x_{N_1+N_2}):=\frac{1}{\sqrt{N_1!\, N_2!\,(N_1+N_2)!}}\;\times\\
\times\sum_{\sigma\in\cS_{N_1+N_2}}\sgn(\sigma)\,\psi_1\big(x_{\sigma(1)},...,x_{\sigma(N_1)}\big)\;\psi_2\big(x_{\sigma(N_1+1)},...,x_{\sigma(N_1+N_2)}\big). 
\end{multline}
Here $\cS_N$ is the group of permutations of $\{1,...,N\}$. 
When $\{f_i\}$ is an orthonormal basis of $\gH$, then $\{f_{i_1}\wedge\cdots\wedge f_{i_N}\}_{i_1<\cdots <i_N}$ forms an orthonormal basis of $\gH^N_a$. 

For bosons, we define similarly, for $\psi_1\in\gH_s^{N_1}$ and $\psi_2\in \gH_s^{N_2}$,
\begin{multline}
\psi_1\vee\psi_2(x_1,...,x_{N_1+N_2}):=\frac{1}{\sqrt{N_1!\, N_2!\,(N_1+N_2)!}}\;\times\\
\times\sum_{\sigma\in\cS_{N_1+N_2}}\psi_1\big(x_{\sigma(1)},...,x_{\sigma(N_1)}\big)\;\psi_2\big(x_{\sigma(N_1+1)},...,x_{\sigma(N_1+N_2)}\big). 
\end{multline}
When $\{f_i\}$ is an orthonormal basis of $\gH$, then $\{f_{i_1}\vee\cdots\vee f_{i_N}\}_{i_1\leq\cdots \leq i_N}$ is an ortho\emph{gonal} basis of $\gH^N_s$. Note that by definition
$$\underbrace{f\vee\cdots \vee f}_{\text{$N$ times}}=\sqrt{N!}\;f\otimes\cdots\otimes f.$$

\subsection{Creation and annihilation operators}
For every $f\in\gH$, we define the \emph{creation operator} $a^\dagger(f)$ on $\cF^{\rm fin}_{a/s}:=\cup_{N\geq1}\cFN_{a/s}\subset\cF_{a/s}$ by requiring $a^\dagger(f)\gH^N_{a/s}\subset\gH^{N+1}_{a/s}$ for all $N\geq0$, with
$$\forall\psi\in\gH^N_{a/s},\qquad a^\dagger(f)\psi:=\left\{\begin{array}{ll}
f\wedge\psi& \text{for fermions,}\\
f\vee\psi& \text{for bosons.}\end{array}\right.
$$
By linearity, $a^\dagger(f)$ can be defined as an operator on $\cF^{\rm fin}_{a/s}$. 
Note that if $\{f_i\}_{i\geq1}$ is an orthonormal basis of $\gH$, then $\{\prod_{k=1}^K a^\dagger(f_{i_k})\Omega\}_{i_1<\cdots <i_K,\ K\geq0}$ is an orthonormal basis of $\cF_a$ and $\{\prod_{k=1}^K a^\dagger(f_{i_k})\Omega\}_{i_1\leq\cdots \leq i_K,\ K\geq0}$ is an orthogonal basis of $\cF_s$.

Similarly, we define the \emph{annihilation operator} $a(f)$ by requiring $a(f)\gH^N_{a/s}\subset\gH^{N-1}_{a/s}$ for all $N\geq1$, $a(f)\Omega=0$ and 
$$\forall\psi\in\gH^N_{a/s},\qquad \left(a(f)\psi\right)\big(x_1,...,x_{N-1}\big):=\sqrt{N}\int_{\R^d}\overline{f(x)}\psi(x,x_1,...,x_{N-1})\,dx.$$
It can be verified that $a(f)$ is the adjoint of $a^\dagger(f)$ on $\cF^{\rm fin}_{a/s}$: 
$$\forall \Psi,\Psi'\in\cF^{\rm fin}_{a/s},\qquad \pscal{\Psi,a^\dagger(f)\Psi'}_{\cF_{a/s}}=\pscal{a(f)\Psi,\Psi'}_{\cF_{a/s}}.$$

In the fermionic case the creation and annihilation operators satisfy the so-called \emph{Canonical Anticommutation Relations} (CAR):
\begin{equation}
\left\{\begin{array}{rl}
a(g)a^\dagger(f)+a^\dagger(f)a(g)&=\pscal{g,f}\1_{\cF_a},\\[0,1cm]
a^\dagger(f)a^\dagger(g)+a^\dagger(g)a^\dagger(f)&=0,\\[0,1cm]
a(f)a(g)+a(g)a(f)&=0.
\end{array}\right.
\label{eq:CAR}
\end{equation}
These relations are satisfied on $\cF^{\rm fin}_a$ but it is deduced from the CAR that $\norm{a^\dagger(f)}=\norm{a(f)}= \norm{f}_\gH$, hence that $a^\dagger(f)$ and $a(f)$ can be extended to bounded operators on the whole fermionic Fock space $\cF_a$. 
In the bosonic case, the creation and annihilation operators satisfy the so-called \emph{Canonical Commutation Relations} (CCR):
\begin{equation}
\left\{\begin{array}{rl}
a(g)a^\dagger(f)-a^\dagger(f)a(g)&=\pscal{g,f}\1_{\cF_s},\\[0,1cm]
a^\dagger(f)a^\dagger(g)-a^\dagger(g)a^\dagger(f)&=0,\\[0,1cm]
a(f)a(g)-a(g)a(f)&=0.
\end{array}\right.
\label{eq:CCR}
\end{equation}
These relations are satisfied on $\cF^{\rm fin}_s$. Now $a(f)$ and $a^\dagger(f)$ are unbounded operators. However, they are bounded on $\cFN_s$ (with values in $\cF_s^{N\pm1}$) for every fixed $N$.

\subsection{Observables}
We now define operators and quadratic forms on $\cF_{a/s}$. The most important one is the so-called \emph{number operator} which equals $N$ on any $\gH^N_{a/s}$: 
$$\cN:=\bigoplus_{N\geq0}N.$$ 
This operator is unbounded on $\cF_{a/s}$ and its maximal domain is
$$\cD(\cN):=\bigg\{\Psi=(\psi^0,\psi^1,...)\in\cF\ :\ \sum_{N\geq0}N^2\norm{\psi^N}^2_{\gH^N_{a/s}}<\ii \bigg\}.$$

More generally, for every (densely defined) self-adjoint operator $A$ on $\gH$, we may define by 
$$\mathbb{A}:=0\oplus\bigoplus_{N\geq1}\left(\sum_{i=1}^NA_{x_i}\right)$$
the operator on $\cF_{a/s}$. When $A$ is bounded from below, the domain of $\sum_{i=1}^NA_{x_i}$ is simply $\bigwedge_1^N\cD(A)\subset\gH^N_{a/s}$; in the general case, $\sum_{i=1}^NA_{x_i}$ is essentially self-adjoint on $\bigwedge_1^N\cD(A)\subset\gH^N_{a/s}$, see \cite{ReeSim1}. The operator $\mathbb{A}$ is self-adjoint on the domain
$$\cD(\mathbb{A}):=\left\{\Psi=(\psi^0,\psi^1,...)\in\bigoplus_{N\geq0}\cD\left(\sum_{j=1}^NA_{x_j}\right)\ :\ \sum_{N\geq0}\norm{\left(\sum_{j=1}^NA_{x_j}\right)\psi^N}^2_{\gH^N_{a/s}}<\ii \right\}.$$
In the literature, the second quantization $\mathbb{A}$ of $A$ is often denoted by $\sum_iA_i$ or by ${\rm d}\Gamma(A)$. Note that $\cN$ is the second quantization of the identity on $\gH$.

The operator $\mathbb{A}$ can be expressed in terms of creation and annihilation operators. Let $\{f_i\}_{i\geq1}$ be an orthonormal basis of $A$, with $f_i\in\cD(A)$ for every $i\geq1$. Then we have (both in the fermionic and bosonic cases)
\begin{equation}
\mathbb{A}=\sum_{j\geq1}a^\dagger(Af_j)\,a(f_j)=\sum_{i,j\geq1}A_{ij}\,a^\dagger(f_i)\,a(f_j),\qquad \text{with }\quad A_{ij}=\pscal{f_i,Af_j}_\gH.
\label{eq:2nd_qtz_1_body}
\end{equation}
The above series are well-defined when restricted to any $\bigwedge_1^N\cD(A)\subset\gH^N_{a/s}\subset\cF_{a/s}$ and they coincide with $\sum_{i=1}^NA_{x_i}$, which is the correct interpretation of the (formal) equality \eqref{eq:2nd_qtz_1_body}.
Applying this to the number operator, we obtain:
\begin{equation}
\cN=\sum_{i\geq1}a^\dagger(f_i)\,a(f_i).
\label{eq:2nd_qtz_N} 
\end{equation}

Similarly, we can associate to any two-body operator $W:\gH^2_{a/s}\to\gH^2_{a/s}$ an operator $\mathbb{W}$ on Fock space, defined by
$$\mathbb{W}:=0\oplus0\bigoplus_{N\geq2}\left(\sum_{1\leq i<j\leq N}W_{ij}\right)$$
where $W_{ij}$ denotes the operator $W$ acting on the variables $x_i$ and $x_j$ but not on the other variables. We do not discuss problems of domains for shortness. As for one-body operators, the second quantization $\mathbb{W}$ in Fock space of a two-body operator $W$ can be expressed in terms of creation and annihilation operators as follows:
\begin{equation}
\mathbb{W}=\sum_{\substack{1\leq k\leq\ell\\ 1\leq i\leq j}} W_{ij,k\ell}\;a^\dagger(f_i)\,a^\dagger(f_j)\,a(f_\ell)\,a(f_k),
\label{eq:2nd_qtz_2_body}
\end{equation}
with
\begin{equation}
W_{ij,k\ell}:=\left\{\begin{array}{ll}
\pscal{f_i\wedge f_j,Wf_k\wedge f_\ell}_{\gH^2_a}&\text{ (fermions)}\\[0,4cm]
\displaystyle\frac{\pscal{f_i\vee f_j,Wf_k\vee f_\ell}_{\gH^2_s}}{(1+\delta_{ij})(1+\delta_{k\ell})}&\text{ (bosons).}\\
\end{array}\right.
\label{eq:def_W_ij} 
\end{equation}
Note the normalization factor $(1+\delta_{ij})(1+\delta_{k\ell})=\norm{f_i\vee f_j}^2\norm{f_k\vee f_\ell}^2$ for bosons.

In particular, for an $N$-body Hamiltonian of the form
\begin{equation}
H^V(N):=\sum_{j=1}^N\left(\frac{-\Delta_{x_j}}{2}+V(x_j)\right)+\sum_{1\leq k<\ell\leq N}W(x_k-x_\ell),
\end{equation}
with the convention that $H^V(1)=-\Delta/2+V$ and $H^V(0)=0$, the corresponding Hamiltonian in Fock space defined by
$\mathbb{H}^V:=\bigoplus_{N\geq0}H^V(N)$
can be expressed as
\begin{equation}
\mathbb{H}^V=\sum_{i,j\geq1}h_{ij}\,a^\dagger(f_i)\,a(f_j)+\sum_{\substack{1\leq k\leq\ell\\ 1\leq i\leq j}} W_{ij,k\ell}\;a^\dagger(f_i)\,a^\dagger(f_j)\,a(f_\ell)\,a(f_k)
\label{eq:2nd_qtz_Schrodinger} 
\end{equation}
with $W_{ij,k\ell}$ as in \eqref{eq:def_W_ij} and 
$$h_{ij}=\int_{\R^d}\left(\frac{\overline{\nabla f_i(x)}\cdot\nabla f_j(x)}{2}+V(x)\overline{f_i(x)}f_j(x)\right)\, dx.$$

\begin{remark}
Physicists rather prefer to use the creation operator $\phi^\dagger(x)$ of a particle at $x\in\R^d$, formally related to the `smeared' operator $a^\dagger(f)$ by
$$a^\dagger(f)=\int_{\R^d}f(x)\phi^\dagger(x)\,dx,\qquad \phi^\dagger(x)=\sum_{i\geq1}\overline{f_i(x)}\,a^\dagger(f_i)$$
where $\{f_i\}$ is any orthonormal basis of $L^2(\R^d)$.
The formula \eqref{eq:2nd_qtz_Schrodinger} can then be rewritten as follows:
\begin{multline}
\mathbb{H}=\int_{\R^d}\left(\frac12\nabla \phi^\dagger(x)\cdot\nabla \phi(x)+V(x) \phi^\dagger(x)\phi(x)\right)dx\\
+\frac12\int_{\R^d}\int_{\R^d}W(x-y)\phi^\dagger(x)\phi^\dagger(y)\phi(y)\phi(x)\,dx\,dy.
\end{multline}
\end{remark}

\subsection{States, density matrices}\label{sec:def_DM}
A \emph{(mixed or normal) state} on a (separable) Hilbert space $\gX$ is a non-negative trace-class self-adjoint operator $\Gamma\in\gS_1(\gX)$ such that $\tr(\Gamma)=1$. A \emph{pure state} is an orthogonal projector: $\Gamma=|\Psi\rangle\langle\Psi|$. By the spectral theorem, any state is a convex combination of pure states:
$$\Gamma=\sum_{i\geq1}n_i\,|\Psi_i\rangle\langle\Psi_i|,\qquad \text{where $n_i\geq 0$ and $\sum_{i\geq1}n_i=1$}.$$
Even when the system is expected to be in a pure state, mixed states are very important tools that we use all the time.

We always use the word `state' for mixed state and only make comments related to a more general notion of states (a positive and normalized linear form on a $C^*$-algebra \cite{BraRob1,BraRob2}). We denote by
$$\cS(\gX):=\left\{\Gamma=\Gamma^*\geq0\ :\ \tr_\gX(\Gamma)=1\right\}$$
the convex set of all states on the Hilbert space $\gX$. The natural topology on $\cS(\gX)$ is that induced by the strong topology of $\gS_1(\gX)$. The set $\cS(\gX)$ is convex but it is \emph{not} closed for the weak--$\ast$ topology of $\gS_1(\gX)$. Indeed we have in general that if $\Gamma_n\wto\Gamma$ weakly--$\ast$ with $\{\Gamma_n\}\subset\cS(\gX)$, then $\Gamma=\Gamma^*\geq0$ but 
$$\tr_\gX(\Gamma)\leq \liminf_{n\to\ii}\tr_\gX(\Gamma_n)=1$$
which is the operator version of Fatou's Lemma \cite{Simon-79}. However it is known \cite{dellAntonio-67,Simon-79} that if $\tr_\gX(\Gamma)=1$ then the convergence is strong: $\norm{\Gamma_n-\Gamma}_{\gS_1(\gX)}\to0$. The fact that a weak--$\ast$ limit of a sequence of states is not always a state is a disease that will be repaired in Section \ref{sec:geom_CV}, when we introduce the geometric topology.

\medskip

For a state $\Gamma$ on the fermionic or bosonic Fock space $\cF_{a/s}$, we define the \emph{density matrix} $[\Gamma]^{(p,q)}:\gH^q_{a/s}\to\gH^p_{a/s}$ by the relation
\begin{equation}
{\pscal{g_1 \circ\cdots \circ g_p\,,\,[\Gamma]^{(p,q)}f_1\circ\cdots \circ f_q}_{\gH^p_{a/s}}=\tr_{\cF_{a/s}}\bigg(\Gamma\; a^\dagger(f_1)\cdots a^\dagger(f_q)\,a(g_p)\cdots a(g_1)\bigg)}
\label{eq:def_density_matrix}
\end{equation}
where $\circ=\wedge$ for fermions and $\circ=\vee$ for bosons.
When $p=q$ we use the notation $[\Gamma]^{(p)}$ for the usual $p$-body density matrix of $\Gamma$. Note that $[\Gamma]^{(0)}=\tr_{\cF_{a/s}}(\Gamma)=1$ by definition.

\begin{remark}
If $\Gamma$ commutes with the number operator $\cN$, that is $\Gamma=\bigoplus_{n\geq0}G_n$ with $G_n:\gH^n_{a/s}\to\gH^n_{a/s}$, then it holds $[\Gamma]^{(p,q)}\equiv0$ for $p\neq q$.\dqed
\end{remark}

\begin{remark}
We may define by the same formula the density matrices $[\Gamma]^{(p,q)}$ of any trace-class operator $\Gamma$ (not necessarily self-adjoint and non-negative).\dqed
\end{remark}

For fermions the creation and annihilation operators are bounded and \eqref{eq:def_density_matrix} always properly define the operators $[\Gamma]^{(p,q)}$. For bosons, however, assumptions on $\Gamma$ are needed to make \eqref{eq:def_density_matrix} meaningful. In the following we almost always consider states on the truncated Fock space $\cFN_{a/s}$ for which \eqref{eq:def_density_matrix} makes sense, as we explain below.

Any state $G$ on the $N$-body space $\gH^N_{a/s}$ can also be seen as a state on the Fock spaces $\cFN_{a/s}$ and $\cF_{a/s}$,  by extending it to zero on sectors of $k$ particles with $k\neq N$. A calculation shows that the kernel of $[G]^{(p)}$ is given for $p=0,...,N$ by the well-known formula:
\begin{multline}
[G]^{(p)}(x_1,...,x_p;x_1',...,x_p')\\={N\choose p}\int_{\R^d}dy_{p+1}\cdots\int_{\R^d}dy_{N}\;G\big(x_1,...,x_p,y_{p+1},...,y_N\,;\, x'_1,...,x'_p,y_{p+1},...,y_N\big).\label{eq:DM_N_body}
\end{multline}
Saying differently, it is obtained (up to a constant) by taking a partial trace of $G$ with respect to $N-p$ variables. In particular it holds $\tr_{\gH^p}[G]^{(p)}={N\choose p}$. If $p\geq N+1$, then $[G]^{(p)}\equiv0$.

If $\Gamma$ is any state on the truncated Fock space $\cFN_{a/s}$, then $[\Gamma]^{(p,q)}\equiv0$ if $p\geq N+1$ or $q\geq N+1$. Furthermore all the $[\Gamma]^{(p,q)}$ are trace-class operators, as stated in the following fundamental result.

\begin{lemma}[Density matrices of states on $\cFN_{a/s}$]\label{lem:density_matrices}
$\rm (i)$ For all $0\leq p,q\leq N$ and all state $\Gamma\in\cS(\cFN_{a/s})$, the density matrix $[\Gamma]^{(p,q)}$ is trace-class:
\begin{equation}
\norm{[\Gamma]^{(p,q)}}_{\gS_1(\gH^p_{a/s},\gH^q_{a/s})}\leq \sum_{j=0}^{\min(N-p,N-q)}\sqrt{{{p+j}\choose p}{{q+j}\choose q}},
\label{eq:estim_DM}
\end{equation}
Furthermore the map
\begin{equation}
\Gamma\in\cS\left(\cFN_{a/s}\right)\longmapsto [\Gamma]^{(p,q)}\in \gS_1\left(\gH^q_{a/s},\gH^p_{a/s}\right)
\label{eq:map_densiy_matrices}
\end{equation}
is continuous.

\medskip

\noindent $\rm{(ii)}$ States on $\cFN_{a/s}$ are fully determined by their density matrices: if $\Gamma_1,\Gamma_2\in\cS\big(\cFN_{a/s}\big)$ are such that  $[\Gamma_1]^{(p,q)}=[\Gamma_2]^{(p,q)}$ for all $0\leq p,q\leq N$, then $\Gamma_1=\Gamma_2$.
\end{lemma}

The bound \eqref{eq:estim_DM} is certainly not optimal and it is only provided as an illustration. It is a well known general fact that (regular) states are fully determined by their density matrices~\cite{BraRob2}. Our proof below is based on the explicit relation~\eqref{eq:relation_DM_G} between the density matrices $[\Gamma]^{(p,q)}$ and the state $\Gamma$, when the latter is in $\cS(\cFN_{a/s})$. These relations are useful in practice.

Note that the linear map in \eqref{eq:map_densiy_matrices} is \emph{not weakly--$\ast$} continuous. If for instance $\{\phi_n\}$ is an orthonormal system of $\gH$ and $\Gamma_n=|\Psi_n\rangle\langle\Psi_n|$ with $\Psi_n=\phi_1\wedge\phi_n\in\gH^2_a$, then $\Gamma_n\wto0$ weakly--$\ast$ in $\gS_1(\cFN_{a/s})$ but $[\Gamma_n]^{(1)}\wto|\phi_1\rangle\langle\phi_1|$ weakly--$\ast$ in $\gS_1(\gH)$. Indeed, the purpose of the next section is precisely to introduce and study a weak topology that renders all the maps in \eqref{eq:map_densiy_matrices} weakly continuous.

We now provide the proof of Lemma~\ref{lem:density_matrices}.

\begin{proof}
We start by proving that for any state $\Gamma$, it holds $[\Gamma]^{(p,q)}\in\gS_1(\gH^q_{a/s},\gH^p_{a/s})$ for all $0\leq p,q\leq N$. We introduce the matrix elements $G_{mn}=\Pi_m\Gamma\Pi_n:\gH^n_{a/s}\to \gH^m_{a/s}$ where $\Pi_n:=\1_{\{n\}}(\cN)$ is the orthogonal projector onto $\gH^n_{a/s}$. Since $\Gamma\in\gS_1(\cFN_{a/s})$, we have $G_{mn}\in \gS_1(\gH^n_{a/s}, \gH^m_{a/s})$ for all $0\leq m,n\leq N$.

It is easy to see from the definition of the density matrices that $[G_{mn}]^{(p,q)}=0$ except when $m-p=n-q\geq0$. A calculation shows that, in terms of kernels,
\begin{multline}
[G_{mn}]^{(p,q)}(x_1,...,x_p;x'_1,...,x'_q)\\
=\sqrt{{m\choose p}{n\choose q}}\int_{\R^d}dy_{q+1}\cdots \int_{\R^d}dy_{n}\; G_{mn}(x_1,...,x_p,y_{q+1},...,y_n;x'_1,...,x'_q,y_{q+1},...,y_n).
\label{eq:expression_DM_Gmn}
\end{multline}
Since the partial trace of a trace-class operator is itself trace-class, we conclude that $[G_{mn}]^{(p,q)}$ is trace-class for all $0\leq p,q\leq N$, and that
$$\norm{[G_{mn}]^{(p,q)}}_{\gS_1(\gH^p_{a/s},\gH^q_{a/s})}\leq \sqrt{{m\choose p}{n\choose q}}\norm{G_{mn}}_{\gS_1(\gH^m_{a/s},\gH^n_{a/s})}\leq \sqrt{{m\choose p}{n\choose q}}$$
where we have used that $\norm{\Pi_m\Gamma\Pi_n}_{\gS_1}\leq \norm{\Gamma}_{\gS_1}=1$.
The continuity in the trace norm is an obvious consequence of the continuity of partial traces.

Let $0\leq p,q\leq N$ and recall that only the matrix elements $G_{mn}$ such that $m-p=n-q\geq0$ contribute to $[\Gamma]^{(p,q)}$.
For instance $[\Gamma]^{(N,k)}=[G_{Nk}]^{(N,k)}=G_{Nk}$ and $[\Gamma]^{(k,N)}=[G_{kN}]^{(k,N)}=G_{kN}$
for all $0\leq k\leq N$.
Indeed the following holds for all $0\leq p,q\leq N$:
\begin{equation}
[\Gamma]^{(p,q)}=\sum_{j=0}^{\min(N-p,N-q)} [G_{p+j\,q+j}]^{(p,q)}
\label{eq:triangular_system}
\end{equation}
which implies \eqref{eq:estim_DM}.
If we think of the density matrices $[\Gamma]^{(p,q)}$ as being given, the previous equation \eqref{eq:triangular_system} is a triangular system which allows to find all the $G_{mn}$ by induction. 
Inverting this system leads to the following formula:
\begin{equation}
G_{mn}=[\Gamma]^{(m,n)}+\sum_{j=1}^{\min(N-m,N-n)}(-1)^{j} \left[ [\Gamma]^{(m+j,n+j)}\right]^{(m,n)}.
\label{eq:relation_DM_G}
\end{equation}
This shows that states on $\cFN_{a/s}$ are uniquely determined by their density matrices.
\end{proof}

For $m=n$, \eqref{eq:relation_DM_G} may be written
\begin{equation}
G_{mm}=[\Gamma]^{(m,m)}+\sum_{j=1}^{N-m}(-1)^{j} {m+j\choose m}\;\tr_{m+1\to m+j}\; [\Gamma]^{(m+j,m+j)}
\label{eq:relation_DM_G2}
\end{equation}
where $\tr_{m+1\to m+j}$ denotes the partial trace with respect to the $j$ last variables.
For instance, we have $[\Gamma]^{(N)}=G_{NN}$, and
$$G_{N-1\,N-1}=[\Gamma]^{(N-1)}-N\,\tr_{N}[\Gamma]^{(N)}$$
which follows from the fact that 
$$[\Gamma]^{(N-1)}=[G_{N-1\,N-1}]^{(N-1)}+[G_{N\,N}]^{(N-1)}.$$

\medskip

\begin{remark}
Lemma~\ref{lem:density_matrices} is not true as such on the set $\cS(\cF_{a/s})$ of states on the whole Fock space. In general, we have $\Gamma^{(p)}\geq0$ and
$$\tr_{\gH^p_{a/s}}\Gamma^{(p)}=\tr_{\cF_{a/s}}\left({\cN\choose p}\Gamma\right)$$
which is finite only under appropriate assumptions on $\Gamma$. The off-diagonal density matrices $[\Gamma]^{(p,q)}$ are in general only Hilbert-Schmidt when all the $[\Gamma]^{(p)}$ are trace-class.\dqed
\end{remark}

\begin{remark}
We say that a family of operators $\{\Upsilon^m\}_{m=0}^N$ with $\Upsilon^m\in\gS_1(\gH^m)$ is \emph{$\cFN_{a/s}$--representable} when there exists $\Gamma\in \cS(\cFN_{a/s})$ with $[\Gamma,\cN]=0$ such that $\Gamma^{(m,m)}=\Upsilon^m$ for all $m=0,...,N$. Using Formula~\eqref{eq:relation_DM_G2}, we see that $\cFN_{a/s}$--representability is equivalent to having $\Upsilon^{0}=1$ and 
$$\forall m=0,...,N,\qquad \Upsilon^{m}+\sum_{j=m+1}^N(-1)^{j+m} {j\choose m}\;\tr_{m+1\to j}\; \Upsilon^{j}\geq0.$$
The case of states which do not commute with $\cN$ is more involved.\dqed
\end{remark}

\medskip

In this section we have introduced Fock spaces and creation/annihilation operators for indistinguishable fermions or bosons. When working in the truncated space $\cFN_{a/s}$ defined in \eqref{eq:cut-Fock-space}, the statistics of the particles does not make a big difference. To simplify notation, we now write $\gH^p$, $\cF$, $\cFN$, etc, without specifying the considered statistics, except for results which are specific to bosons or fermions.

\section{Geometric convergence}\label{sec:geom_CV}

\subsection{Definition and properties}
\subsubsection{Definition}

We define a weak topology on states in $\cS(\cFN)$, induced by the weak--$\ast$ topologies of all the density matrices $[\Gamma]^{(p,q)}$:

\begin{definition}[Geometric topology \& convergence]\label{def:geom_CV}\it
We define the \emph{geometric topology} $\cT$ on $\cS(\cFN)$ as the coarsest topology such that the maps
\begin{equation}
\Gamma\in\cS(\cFN) \longmapsto \pscal{\psi\,,\,[\Gamma]^{(p,q)}\psi'}_{\gH^p}
\label{eq:map_gto} 
\end{equation}
remain continuous for all $(\psi,\psi')\in\gH^p\times \gH^q$ and all $0\leq p,q\leq N$. 

Let $\{\Gamma_n\}$ be a sequence of states on $\cFN$, and $\Gamma$ be a state on $\cFN$. The sequence $\{\Gamma_n\}$ is said to \emph{converge geometrically} to $\Gamma$ if 
\begin{equation}
\lim_{n\to\ii}\pscal{\psi\,,\,[\Gamma_n]^{(p,q)}\psi'}_{\gH^p}=\pscal{\psi\,,\,[\Gamma]^{(p,q)}\psi'}_{\gH^p}
\label{eq:def_geom_CV}
\end{equation}
for all $(\psi,\psi')\in\gH^p\times \gH^q$ and all $0\leq p,q\leq N$.
We use the notation $\Gamma_n\gto\Gamma$. 
\end{definition}

Note that, when it exists, the geometric limit $\Gamma$ is uniquely defined since $\Gamma\in\cS(\cFN)$ is characterized by its density matrices $[\Gamma]^{(p,q)}$, by Lemma~\ref{lem:density_matrices}. 

We give several examples right after the following result which is an immediate consequence of Lemma~\ref{lem:density_matrices}.

\begin{lemma}[Elementary properties of geometric convergence]\label{lem:prop_geom_CV}

\vspace{-0,2cm}
\begin{enumerate}
\item The geometric topology $\cT$ is coarser than the usual norm topology. If $\Gamma_n\to\Gamma$ strongly in $\gS_1(\cFN)$, then $\Gamma_n\wto_g\Gamma$ geometrically.
\item We have $\Gamma_n\wto_g\Gamma$ in $\cFN$, if and only if $[\Gamma_n]^{(p,q)}\wto [\Gamma]^{(p,q)}$ weakly--$\ast$ in $\gS_1$, for all $0\leq p,q\leq N$.
\end{enumerate}
\end{lemma}
\begin{proof}
The first assertion follows from the (strong) continuity of the maps $\Gamma\mapsto[\Gamma]^{(p,q)}$ for all $0\leq p,q\leq N$, as stated in Lemma~\ref{lem:density_matrices}. The second assertion is a consequence of the uniform trace-class bound \eqref{eq:estim_DM} on all the density matrices and of the Banach-Alaoglu Theorem in $\gS_1(\gH^q,\gH^p)$.
\end{proof}

Let us emphasize that the geometric limit $\Gamma$ of a sequence of states is, by definition, always a \emph{state}, that is it must satisfy $\tr_\cF(\Gamma)=1$. Contrarily to the usual weak--$\ast$ convergence on $\gS_1(\cFN)$, there is never any loss in the trace-norm when $\Gamma_n\wto_g\Gamma$. If in the geometric limit some particles are lost, then $\Gamma$ lives on spaces with less particles in $\cFN$. If all the particles are lost, then we have $\Gamma=|\Omega\rangle\langle\Omega|$, the vacuum state in $\cFN$. 

We now provide examples of sequences $\{\Gamma_n\}$ which geometrically converge but do not strongly converge to a limit $\Gamma$. Our claims can be verified by computing the density matrices $[\Gamma_n]^{(p,q)}$ and checking their weak--$\ast$ convergence towards $[\Gamma]^{(p,q)}$.

\begin{example}
Let $\{\phi_n\}$ be an orthonormal basis of $\gH$. Define a sequence of two-body fermionic wavefunctions by $\Psi_n:=\phi_1\wedge\phi_n$, with associated state in $\cF^{\leq2}_a$ denoted by $\Gamma_n=0\oplus0\oplus|\Psi_n\rangle\langle\Psi_n|$. It holds $\Gamma_n\wto_\ast 0$ weakly--$\ast$ and 
$\Gamma_n\wto_g 0\oplus|\phi_1\rangle\langle\phi_1|\oplus0$ geometrically in $\cF^{\leq2}_a$. The geometric limit $\Gamma$ describes a system composed of only one particle, in the state $\phi_1$. The other particle in the state $\phi_n$ has vanished in the limit.\dqed
\end{example}

\begin{example}
Even when $\Gamma_n$ is a pure state for all $n$, the geometric limit $\Gamma$ is not always a pure state. For instance if $\Psi_n:=\phi_n^1\wedge\phi_n^2$ with $\phi_n^1=\cos\alpha\,\phi_1+\sin\alpha\,\phi_n$ and $\phi_n^2=\cos\beta\,\phi_2+\sin\beta\,\phi_{n+1}$, then the corresponding state $\Gamma_n=0\oplus0\oplus|\Psi_n\rangle\langle\Psi_n|$ on $\cF^{\leq2}_a$ converges geometrically to 
\begin{multline*}
\Gamma_n\gto\Gamma=\Big(\sin^2\!\alpha\,\sin^2\!\beta\Big)\oplus\Big(\cos^2\!\alpha\,\sin^2\!\beta\;|\phi_1\rangle\langle\phi_1|+\cos^2\!\beta\,\sin^2\!\alpha\;|\phi_2\rangle\langle\phi_2|\Big)\\
\oplus  \Big(\cos^2\!\alpha\,\cos^2\!\beta\;|\phi_1\wedge\phi_2\rangle\langle\phi_1\wedge\phi_2|\Big). 
\end{multline*}
On the other hand, we have $$\Gamma_n\underset{\ast}{\wto}0\oplus0\oplus\cos^2\!\alpha\,\cos^2\!\beta\;|\phi_1\wedge\phi_2\rangle\langle\phi_1\wedge\phi_2|$$ weakly--$\ast$ in $\gS_1(\cF^{\leq2}_a)$.\dqed
\end{example}

\begin{example}[Hartree states]\label{ex:Hartree}
For bosons, a \emph{Hartree state} takes the form $\Psi=\phi\otimes\cdots\otimes\phi\in\gH^N_s$ where $\norm{\phi}_\gH=1$. Assume that $\{\phi_n\}$ is a sequence of normalized functions in $\gH$, with $\phi_n\wto\phi$ weakly. Let $\Gamma_n=0\oplus\cdots0\oplus|(\phi_n)^{\otimes N}\rangle\langle(\phi_n)^{\otimes N}|$ be the associated $N$-body state in $\cS(\cFN_s)$. Then it holds
$$\Gamma_n\gto \bigoplus_{k=0}^N{N\choose k}\left(1-\norm{\phi}_\gH^2\right)^{N-k}|\phi^{\otimes k}\rangle\langle\phi^{\otimes k}|.$$
It is clear that the convergence is strong if and only if $\norm{\phi}_\gH=1$.\dqed
\end{example}

\begin{example}[Coherent states]\label{ex:coherent}
For bosons, \emph{coherent states} are defined by the formula $\Gamma_f:=W(f)|\Omega\rangle\in\cF_s$ where $W(f)=\exp(a^\dagger(f)-a(f))$ is the Weyl unitary operator ($f$ is any vector of the one-body space $\gH^1$). The latter satisfies the following interwinning relations
\begin{equation}
W(f)^*\Big(a(g)-\pscal{g,f}\Big)W(f)=a(g),\qquad W(f)^*\Big(a^\dagger(g)-\pscal{f,g}\Big)W(f)=a^\dagger(g).
\label{eq:Weyl_relations}
\end{equation}
The density matrix $[\Gamma_f]^{(p,q)}$ of a coherent state $\Gamma_f=W(f)\,|\Omega\rangle\langle\Omega|\,W(f)^*$ is 
$[\Gamma_f]^{(p,q)}=|f^{\otimes p}\rangle\,\langle f^{\otimes q}|$.
Consider a sequence $\{\Gamma_{f_n}\}$ of coherent states with $\{f_n\}$ bounded in $\gH^1$, such that $f_n\wto f$ weakly in $\gH^1$. Then $\Gamma_{f_n}\wto_g \Gamma_f$ geometrically in the sense that $[\Gamma_{f_n}]^{(p,q)}\wto_\ast [\Gamma_{f}]^{(p,q)}$ weakly--$\ast$, for all $p,q\geq0$. Note that coherent states do not live on any truncated Fock space $\cFN$, hence Definition~\ref{def:geom_CV} has to be generalized in an obvious fashion on the whole Fock space $\cF$.\dqed
\end{example}

\begin{example}[Hartree-Fock(-Bogoliubov) states]\label{ex:HFB}
For fermions, there is a subclass of states which are fully characterized by their one-body density matrix $[\Gamma]^{(1)}$ and their \emph{pairing density matrix} $[\Gamma]^{(2,0)}$ (if they commute with $\cN$, they are only characterized by $[\Gamma]^{(1)}$). These states are called \emph{generalized Hartree-Fock states} \cite{BacLieSol-94} or \emph{Hartree-Fock-Bogoliubov states} (when $[\Gamma]^{(2,0)}\neq0$).
Here `fully characterized' means that any density matrix $[\Gamma]^{(p,q)}$ is an explicit function of $[\Gamma]^{(1,1)}$ and $\Gamma^{(2,0)}$, given by Wick's formula, see Eq. (2a.11) in \cite{BacLieSol-94}. 
When $\tr_\gH([\Gamma_n]^{(1)})$ is uniformly bounded, it is easily seen that geometric convergence of generalized Hartree-Fock states is equivalent to the weak--$\ast$ convergence of $[\Gamma_n]^{(1)}$ and of $[\Gamma_n]^{(2,0)}$. The geometric limit is always a generalized Hartree-Fock state.

Note that if $[\Gamma]^{(1)}$ has an infinite rank (but a finite trace), then the corresponding Hartree-Fock state $\Gamma$ does not live on any truncated Fock space $\cF^{\leq k}$. However, geometric convergence can be understood in the same fashion as in the previous example.\dqed
\end{example}

\begin{example}
Let $\Gamma_0$ be any state on $\cFN$ and let $U(t)=e^{-itT}$ (with $T=-\Delta/2$) be the unitary free evolution on the one-body space $\gH^1$, of a non-relativistic particle. Let
$$\mathbb{U}(t)=1\oplus U(t)\oplus \Big(U(t)\otimes U(t)\Big)\oplus\cdots\oplus \Big(U(t)^{\otimes N}\Big)=e^{it\mathbb{T}}$$ 
be the unitary evolution of the second quantization of $T$ on the truncated Fock space $\cF^{\leq N}$: $$\mathbb{T}=0\oplus\bigoplus_{n=1}^N\left(\sum_{j=1}^n\frac{(-\Delta)_j}2\right).$$
The state $\Gamma(t):=\mathbb{U}(t)\Gamma_0\mathbb{U}(t)^*$ is the unique weak solution to the Schrödinger-von Neumann equation
$$\left\{\begin{array}{l}
\displaystyle i\frac{\rm d}{{\rm d}t}\Gamma(t)=\left[\mathbb{T}\,,\,\Gamma(t)\right]\\[0,4cm]
\Gamma(t=0)=\Gamma_0.
\end{array}\right.$$
Then 
$$\Gamma(t)\gto |\Omega\rangle\langle\Omega|\quad \text{as $t\to\pm\ii$}.$$
 Indeed, we have
$$\forall 0\leq p,q\leq N,\qquad [\Gamma(t)]^{(p,q)}=\underbrace{U(t)\otimes\cdots\otimes U(t)}_{p}\;[\Gamma_0]^{(p,q)}\;\underbrace{U(t)^*\otimes\cdots\otimes U(t)^*}_{q}$$
which tends to 0 weakly--$\ast$ in $\gS_1$, when $(p,q)\neq(0,0)$. The same holds if $U(t)$ is any unitary family satisfying  $U(t)\wto0$ weakly as $t\to\pm\ii$.\dqed
\end{example}

After these examples, we now make some fundamental remarks about the notion of geometric convergence.

\begin{remark}[Geometric convergence is a $C^*$-algebra concept]\label{rmk:Clifford}
The geometric topology is the restriction to $\cFN$ of a well-known weak topology arising in $C^*$-algebra theory, a fact that we will need in the proof of Lemma~\ref{lem:compact} below.

For fermions, an equivalent way of formulating \eqref{eq:def_geom_CV} is, by the definition \eqref{eq:def_density_matrix} of density matrices,
\begin{equation}
\forall A\in{\cA},\qquad \lim_{n\to\ii}\tr_\cF(\Gamma_nA)=\tr_\cF(\Gamma A)
\label{eq:geom_CV_Clifford}
\end{equation}
where $\cA$ is the $C^*$-algebra \cite{BraRob1,BraRob2} generated by all the $a^\dagger(f)$ with $f$ any vector in $\gH$. Therefore for fermions the topology $\cT$ is nothing but the usual weak--$\ast$ topology of states on the CAR algebra $\cA$, restricted to states of the truncated Fock space $\cFN_{a/s}$.
For bosons, the same holds true with $\cA$ being the CCR algebra, generated by the Weyl operators of Example~\ref{ex:coherent}. 

Note that we have $\Gamma_n\wto\Gamma$ for the weak--$\ast$ topology of $\gS_1(\cF)$ if and only if 
$$\forall K\in\cK(\cF),\qquad \lim_{n\to\ii}\tr_\cF(\Gamma_nK)=\tr_\cF(\Gamma K)$$
where we recall that $\cK(\cF)$ is the algebra of compact operators.
In both the fermionic and bosonic cases, the CAR/CCR algebra ${\cA}$ does not contain any nontrivial compact operator: ${\cA}\cap\cK(\cF)=\{0\}$. Geometric convergence is thus \emph{a priori} not related to the usual weak--$\ast$ convergence and it is possible to have $\Gamma_n\wto_g\Gamma$ with $\tr_{\cF}(\Gamma)=1$ whereas $\Gamma_n\wto0$ weakly--$\ast$ in $\gS_1(\cF)$, like in the previous examples.\dqed
\end{remark}

\begin{remark}\label{rmk:commute}
If $\Gamma_n$ commutes with the number operator $\cN$ for all $n$, $[\Gamma_n,\cN]=0$, then $[\Gamma_n]^{(p,q)}\equiv0$ for all $p\neq q$ and it is easy to verify that the geometric limit $\Gamma$ of $\{\Gamma_n\}$ must also commutes with $\cN$.\dqed
\end{remark}

\begin{remark}\label{rmk:several_kinds}
A similar definition of the geometric topology and convergence can be provided if the system contains several species of particles. One introduces the density matrices $[\Gamma]^{(p_1,...,p_k,q_1,...,q_\ell)}$ where $p_i$ and $q_i$ respectively count the number of annihilation and creation operators of the species $i$ (bosons or fermions). One works in the truncated Fock space $\cF^{\leq N_1,...,N_k}$ corresponding to having at most $N_i$ particles of species $i$.\dqed
\end{remark}

\subsubsection{Compactness results}\label{sec:prop_geom_CV}
The following result is very useful in practice. It allows us to work with weak limits of density matrices while being sure, at the same time, that the limits arise from a state $\Gamma$.

\begin{lemma}[Geometric compactness of $\cS(\cFN)$]\label{lem:compact}
The set of states $\cS(\cFN)$ on $\cFN$ is (sequentially) compact for the geometric topology $\cT$: every sequence of states $\{\Gamma_n\}\subset\cS(\cFN)$ has a subsequence which converges geometrically, $\Gamma_{n_k}\gto \Gamma$.
\end{lemma}

\begin{proof}
This result immediately follows from well-known facts in the theory of $C^*$-algebras (recall Remark~\ref{rmk:Clifford}). By the Banach-Alaoglu Theorem, any sequence of states $\{\Gamma_n\}\subset\cS(\cF^{\leq N})$ on the CAR (resp. CCR) algebra $\cA$ generated by the creation operators (resp. Weyl operators), has a weakly-convergent subsequence in the sense that for every $A\in\cA$, one has $\tr(\Gamma_{n_k} A)\to\omega(A)$ where $\omega$ is a positive normalized linear form on $\cA$, \cite{BraRob1}. Since $\Gamma_{n}$ lives on the truncated Fock space $\cF^{\leq N}$ for every $n$, it has a uniformly bounded average particle number, hence its weak limit $\omega$ must be a normal state~\cite{BraRob2}: there is a $\Gamma\in\cS(\cF)$ such that $\omega(A)=\tr_\cF(\Gamma A)$ for all $A$. Since $[\Gamma]^{(N+1,N+1)}=0$, it is easy to verify that $\Gamma$ must also live on $\cF^{\leq N}$ and the result follows.
\end{proof}

\begin{remark}
Up to extraction of subsequences, one can always assume that $[\Gamma_n]^{(p,q)}\wto_\ast \Upsilon^{(p,q)}$ weakly--$\ast$ in $\gS_1(\gH^q,\gH^p)$. The matrix elements $G_{m,n}$ of the limit state $\Gamma$ are then uniquely determined from the operators $\Upsilon^{(p,q)}$ by Formula~\eqref{eq:relation_DM_G}. What is more subtle is the fact that the so-obtained $\Gamma$ is really a state, that is $\Gamma=\Gamma^*\geq0$.\dqed
\end{remark}

\begin{remark}
Lemma \ref{lem:compact} can obviously be extended to sequences of states $\{\Gamma_n\}$ on the whole Fock space $\cF$ which satisfy a uniform bound of the form $\tr_\cF(\cN\Gamma_n)\leq C$. \dqed
\end{remark}

The following result says that the total number of particles in the system cannot increase under geometric convergence, and that there is strong convergence if and only if no particle has been lost.

\begin{lemma}[Average particle number and strong convergence]\label{lem:N_wlsc}
Let $\{\Gamma_n\}$ be a sequence of states in $\cS(\cFN)$ and $\Gamma\in\cS(\cFN)$ be a state such that  $\Gamma_n\wto_g\Gamma$. The average particle number is lower semi-continuous:
\begin{equation}
\tr_\cF(\cN\Gamma)\leq\liminf_{n\to\ii}\,\tr_\cF(\cN\Gamma_n).
\label{eq:N_wlsc} 
\end{equation}
Furthermore, if $\lim_{n\to\ii}\tr_\cF(\cN\Gamma_n)=\tr_\cF(\cN\Gamma)$, then $\Gamma_n\to\Gamma$ strongly in $\gS_1(\cFN)$.
\end{lemma}

\begin{proof}
Let us recall that  $\tr_\cF(\cN\Gamma)=\tr_{\gH}[\Gamma]^{(1)}$, hence, since $[\Gamma_n]^{(1)}\wto[\Gamma]^{(1)}$ weakly--$\ast$ in $\gS_1(\gH)$ by Lemma~\ref{lem:prop_geom_CV}, we have
$$\tr_\cF(\cN\Gamma)=\tr_\gH[\Gamma]^{(1)}\leq\liminf_{n\to\ii}\tr_\gH[\Gamma_n]^{(1)}=\liminf_{n\to\ii}\,\tr_\cF(\cN\Gamma_n).$$
Another proof consists in writing that
$\tr_{\cF}\left(\Gamma_n\sum_{i=1}^Ka^\dagger(f_i)\,a(f_i)\right)\leq \tr_{\cF}\left(\Gamma_n\cN\right)$
by \eqref{eq:2nd_qtz_N}. It then suffices to pass to the limit first as $n\to\ii$ and then as $K\to\ii$.

The proof that conservation of the average particle number implies strong convergence requires a bit more work. We start with a sequence of $N$-body states, that is $\Gamma_n=0\oplus\cdots\oplus G^n$ where $G^n\in\cS(\gH^N)$. We assume that $\Gamma_n\wto_g \Gamma$ in $\cFN$. From Remark~\ref{rmk:commute}, we know that $\Gamma$ commutes with $\cN$:
$$\Gamma=\left(\begin{matrix}
G_{00} & &0\\
 & \ddots & \\
0 & & G_{NN}
\end{matrix}\right).$$
The assumption that 
$$N=\lim_{n\to\ii}\tr_{\cF}(\Gamma_n)=\tr_{\cF}(\cN \Gamma)=\sum_{k=0}^Nk\,\tr_{\gH^k}(G_{kk})$$
together with the fact that $\sum_{k=0}^N\tr_{\gH^k}(G_{kk})=1$ since $G$ is a state, imply that $G_{kk}=0$ for all $k=0,...,N-1$ and $\tr_{\gH^N}(G_{NN})=1$.
However, we know that $G_{NN}$ is the weak--$\ast$ limit of $G^n$ in $\gS_1(\gH^N)$. Therefore $\tr_{\gH^N}(G_{NN})=1$ implies that $G^n\to G_{NN}$ strongly in $\gS_1$, by the reciprocal of Fatou's Lemma for trace-class operators (see \cite{dellAntonio-67,Simon-79}), and the result follows.

We now come back to the general case. Let $\Gamma_n\wto_g \Gamma$ be an arbitrary sequence which converges geometrically in $\cFN$, such that $\tr_\cF(\cN\Gamma)=\lim_{n\to\ii}\tr_\cF(\cN\Gamma)$. We denote by $G_{k\ell}^n$ the matrix elements of $\Gamma_n$ and introduce the auxiliary state
$$\tilde\Gamma_n=\left(\begin{matrix}
G_{00}^n & &0\\
 & \ddots & \\
0 & & G_{NN}^n
\end{matrix}\right)$$
obtained by retaining only the diagonal of $\Gamma_n$. It is easy to check that $\tilde{\Gamma}_n\wto_g\tilde{\Gamma}$, the diagonal of $\Gamma$. We first prove that $\tilde{\Gamma}_n\to\tilde{\Gamma}$ strongly. Indeed we may write
$$\tilde{\Gamma}_n=\sum_{k=0}^Nt_k^n\tilde{G}_{kk}^n$$
where $t_k^n=\tr_{\gH^k}(G_{kk}^n)$ and $\tilde{G}_{kk}^n={G}_{kk}^n/t_k^n$ is a state on $\gH^k$ (with an obvious convention when $t_k^n=0$).
We have $\tilde{G}_{kk}^n\wto_g \tilde{G}_{kk}$ for all $k=0,...,N$ and $\tilde{\Gamma}=\sum_{k=0}^N t_k\tilde{G}_{kk}$ with $t_k=\lim_{n\to\ii}t_k^n$ (up to subsequences). Our assumption means that
$$\sum_{k=0}^Nt_k\tr_{\cF}(\cN\tilde{G}_{kk})= \sum_{k=0}^Nk\,t_k.$$
However by \eqref{eq:N_wlsc}, it holds $\tr_{\cF}(\cN\tilde{G}_{kk})\leq k$ for all $k$, hence the previous equation means that  $\tr_{\cF}(\cN\tilde{G}_{kk})= k$ for all $k=0,...,N$ such that $t_k\neq0$. As we have shown in the previous paragraph, this implies that $\tilde{G}_{kk}^n\to \tilde{G}_{kk}$ strongly in $\gS_1(\gH^k)$. When $t_k=0$, we have simply $G_{kk}^n\to0$ strongly. This eventually shows that $\tilde{\Gamma}_n\to\tilde\Gamma$ strongly.

We now conclude that $\Gamma_n\to\Gamma$ strongly. Indeed, we have $\Gamma_n\wto \Gamma'$ weakly--$\ast$ in $\gS_1(\cFN)$ and we know that the diagonal of $\Gamma_n$ converges strongly, hence in particular $\tr(\Gamma')=1$. By the reciprocal of Fatou's Lemma \cite{dellAntonio-67,Simon-79}, this implies that $\Gamma_n\to\Gamma$ strongly, which ends the proof of Lemma~\ref{lem:N_wlsc}.
\end{proof}

\subsection{Application: HVZ theorem in the lower semi-continuous case}\label{sec:HVZ_wlsc}
In this section, we illustrate the use of geometric convergence on the very simple example of a many-body system with a \emph{non-negative} two-body interaction. Our example covers the celebrated case of atoms and molecules. 

We consider the following many-body Hamiltonian
\begin{equation}
\boxed{H^V(N):=\sum_{j=1}^N\left(-\frac{\Delta_{x_j}}{2}+V(x_j)\right)+\sum_{1\leq k\leq \ell\leq N}W(x_k-x_\ell)}
\label{def:Hamiltonian_W}
\end{equation}
on $L^2_{a/s}((\R^d)^N)$.
Since in practice $W$ is fixed (it is a characteristics of the studied particles) whereas $V$ is an external field that can be varied, we only emphasize $V$ in the notation of $H^V(N)$.
We choose any statistics (bosons or fermions) for our particles. The spectrum of $H^V(N)$ depends on this statistics but our results are stated the same in both cases. 

We assume that $W$ is even and that the two real functions $V$ and $W$ can both be written in the form $\sum_{i=1}^Kf_i$ with $f_i\in L^{p_i}(\R^d)$ where $\max(1,d/2)<p_i<\ii$ or $p_i=\ii$ but $f_i\to0$ at infinity.
These conditions ensure that $(1-\Delta)^{-1/2}V(1-\Delta)^{-1/2}$ and $(1-\Delta)^{-1/2}W(1-\Delta)^{-1/2}$ are compact operators. Then, by the KLMN Theorem \cite{ReeSim2}, $H^V(N)$ has a unique self-adjoint realization in the $N$-body space $L^2_{a/s}((\R^d)^N)$ with quadratic form domain $H^1_{a/s}((\R^d)^N)$. More precisely, for every $0<\epsilon<1$, there exists a constant $C=C(N,\epsilon)\geq0$ such that
\begin{equation}
(1-\epsilon)\left(\sum_{j=1}^N-\Delta_{x_j}\right)-C \leq H^V(N)\leq (1+\epsilon)\left(\sum_{j=1}^N-\Delta_{x_j}\right)+C 
\label{eq:bound_Hamiltonian}
\end{equation}
in the sense of quadratic forms on $L^2_{a/s}((\R^d)^N)$. In this section we will make the assumption that the interaction is repulsive, that is
$$\boxed{W\geq0.}$$
The general case is treated later in Section~\ref{sec:HVZ_general}.

\begin{example}[Atoms and molecules]\label{ex:atoms_molecules}
For atoms and molecules in which the electrons are treated as quantum particles whereas the nuclei are considered as fixed pointwise classical particles (Born-Oppenheimer approximation), we have in atomic units, on $L^2_a((\R^3)^N)$,
$$V(x)=-\sum_{m=1}^M\frac{z_m}{|x-R_m|}\qquad\text{and}\qquad W(x-y)=\frac{1}{|x-y|},$$
where $R_m$ and $z_m$ are the positions and charges of the nuclei. The functions $V$ and $W$ are respectively the Coulomb attraction potential induced by the nuclei, and the Coulomb repulsion between the electrons.\dqed
\end{example}

The second-quantization of $H^V(N)$ is the Fock Hamiltonian 
$$\bH^V:=0\oplus\bigoplus_{k\geq1}H^V(k)$$
which we restrict to the truncated Fock space $\cF^{\leq N}$.
The energy of the system in the state $\Gamma\in\cS\left(\cFN\right)$ reads, using \eqref{eq:2nd_qtz_Schrodinger} and the definition \eqref{eq:def_density_matrix} of the one- and two-body density matrices $[\Gamma]^{(1)}$ and $[\Gamma]^{(2)}$:
\begin{align}
\cE^V(\Gamma)&:=\tr_\cF\big(\bH^V\Gamma\big)\nonumber\\
&= \tr_{L^2(\R^d)}\left(\left(-\frac12\Delta+V\right)[\Gamma]^{(1)}\right)+\tr_{L^2_{a/s}(\R^d\times\R^d)}\left(W[\Gamma]^{(2)}\right).\label{eq:energy_wlsc}
\end{align}
By \eqref{eq:bound_Hamiltonian}, the energy is well-defined for states $\Gamma\in\cS(\cFN)$ such that 
\begin{align*}
\tr_\cF\left(\bT^{1/2}\Gamma\bT^{1/2}\right)&=\tr_\gH\left((-\Delta)^{1/2}\Gamma^{(1)}(-\Delta)^{1/2}\right)\\
&=\frac1{N-1}\tr_{\gH^2_{a/s}}\left((-\Delta_x+-\Delta_y)^{1/2}\Gamma^{(2)}(-\Delta_x-\Delta_y)^{1/2}\right)<\ii. 
\end{align*}
When the previous kinetic energy term is infinite, we can let $\cE^V(\Gamma):=+\ii$. 

One difficulty of many-body systems is the lack of weak lower semi-continuity (wlsc) of the quantum energy $\Psi\in\gH^N\mapsto\pscal{\Psi,H^V(N)\Psi}$. This was for instance pointed out by Friesecke, see Lemma 1.2 (iii) in \cite{Friesecke-03}. Indeed if we denote by 
\begin{equation}
\boxed{E^V(N):=\inf\sigma\left(H^V(N)\right),\qquad  \Sigma^V(N):=\inf\sigma_{\rm ess}\left(H^V(N)\right),}
\label{def:E_V_N}
\end{equation}
respectively the ground state energy and the bottom of the essential spectrum, we usually have that $\Sigma^V(N)<0$. This implies that for a singular Weyl sequence $\Psi_n\wto0$ it holds $\pscal{\Psi_n,H^V(N)\Psi_n}\to\Sigma^V(N)<0$, showing that the energy is not wlsc. 

We now prove that, on the contrary, when $W\geq0$ the energy is lower semi-continuous for the \emph{geometric convergence} which we have introduced in the previous section. 

\begin{lemma}[Lower semi-continuity of the energy under geometric convergence]\label{lem:wlsc}
Assume that $W\geq0$ and let $\{\Gamma_n\}$ be a sequence of states in $\cFN$ which converges geometrically to $\Gamma$. Then
$$\cE^V(\Gamma)\leq \liminf_{n\to\ii}\cE^V(\Gamma_n).$$
\end{lemma}

\begin{proof}
Under our assumptions on $V$ and $W$, it is easily verified that $\cE^V$ is lower semi-continuous for the \emph{strong} topology of $\gS_1(\cFN)$. We have to prove that the same holds for the geometric topology.

When the kinetic energy of $\{\Gamma_n\}$ is not bounded, there is nothing to show by \eqref{eq:bound_Hamiltonian}, hence we may as well assume that 
$$\tr_{\cF}\left(\bT^{1/2}\,\Gamma_n\,\bT^{1/2}\right)\leq C$$
for a constant $C$ independent of $n$ (this is actually equivalent to assuming that each $p$-body density matrix has a bounded kinetic energy). Since we have by assumption $[\Gamma_n]^{(p)}\wto [\Gamma]^{(p)}$ weakly--$\ast$ in $\gS_1$, we deduce that the geometric limit $\Gamma$ has a finite kinetic energy, hence a finite total energy.
We now remark that
$$\cE^V(\Gamma_n)=\frac12\tr_{L^2(\R^d)}\left((-\Delta)[\Gamma_n]^{(1)}\right)+\int_{\R^d}V\rho_{\Gamma_n}+\tr_{L^2_{a/s}(\R^d\times\R^d)}\left(W[\Gamma_n]^{(2)}\right).$$
where $\rho_{\Gamma_n}(x)=[\Gamma_n]^{(1)}(x,x)$ is the density of the system.
It is then a classical fact that
$$\tr_{L^2(\R^d)}\left((-\Delta)[\Gamma]^{(1)}\right)\leq\liminf_{n\to\ii}\tr_{L^2(\R^d)}\left((-\Delta)[\Gamma_n]^{(1)}\right)$$
$$\tr_{L^2_{a/s}(\R^d\times\R^d)}\left(W[\Gamma]^{(2)}\right)\leq\liminf_{n\to\ii}\tr_{L^2_{a/s}(\R^d\times\R^d)}\left(W[\Gamma_n]^{(2)}\right)$$
\begin{equation}
\int_{\R^d}V\rho_{[\Gamma]^{(1)}}=\lim_{n\to\ii}\int_{\R^d}V\rho_{[\Gamma_n]^{(1)}}.
\label{eq:continuity_potential_energy} 
\end{equation}
The first two claims follow from Fatou's Lemma for trace-class operators \cite{Simon-77} (using $W\geq0$).
The last claim \eqref{eq:continuity_potential_energy} is shown as follows. First the Hoffmann-Ostenhof inequality \cite{Hof-77}
\begin{equation}
\int_{\R^d}\left|\nabla\sqrt{\rho_{\Gamma}}\right|^2\leq \tr_{L^2(\R^d)}\left((-\Delta)[\Gamma]^{(1)}\right),
\label{eq:Hoffmann-Ostenhof} 
\end{equation}
implies that $\sqrt{\rho_{\Gamma_n}}$ is bounded in $H^1(\R^d)$, hence we may as well assume that $\sqrt{\rho_{\Gamma_n}}\to\sqrt{\rho_\Gamma}$ weakly in $H^1(\R^d)$ and strongly in $L^2_{\rm loc}(\R^d)$.
Recall that $V=\sum_{j=1}^KV_j$ with $V_j\in L^{p_j}(\R^d)$ where $\max(d/2,1)<p_j<\ii$ or $V_j\in L^\ii(\R^d)$ and $V_j\to0$ at infinity. For $d\ge3$, we write 
\begin{equation}
\left|\int_{|x|\geq R}V_j(x)\,\rho_{\Gamma_n}(x)\,dx\right|\leq \norm{V_j}_{L^{p_j}(\R^d\setminus B(0,R))}\norm{\rho_{\Gamma_n}}_{L^{q_j}(\R^d)}\leq C\norm{V_j}_{L^{p_j}(\R^d\setminus B(0,R))}
\label{eq:estim_potential_energy} 
\end{equation}
where $1/p_j+1/q_j=1$, hence $1\leq q_j<d/(d-2)$. In the last inequality we have used the Sobolev injection theorem as well as the fact that $\sqrt{\rho_{\Gamma_n}}$ is bounded in $H^1(\R^d)$. On the other hand, by Rellich's theorem, we have a compact injection $H^1(B(0,R))\hookrightarrow L^q(B(0,R))$ for all $2\leq q<2d/(d-2)$ which implies that
$$\lim_{n\to\ii}\int_{|x|\leq R}V_j(x)\,\rho_{\Gamma_n}(x)\,dx=\int_{|x|\leq R}V_j(x)\,\rho_{\Gamma}(x)\,dx.$$
Together with \eqref{eq:estim_potential_energy}, this proves \eqref{eq:continuity_potential_energy}. The proof is the same in dimensions 1 and 2.
\end{proof}

The following is a famous result for many-body systems:
\begin{theorem}[HVZ in the lower semi-continuous case]\label{thm:HVZ_wlsc}
Assume $W\geq0$. Then it holds $E^0(N)=0$ for all $N\geq0$ and
\begin{equation}
\boxed{\Sigma^V(N)=E^V(N-1).}
\label{eq:HVZ_wlsc}
\end{equation}
In particular, $E^V(N)$ is an isolated eigenvalue if and only if 
$$E^V(N)<E^V(N-1)=\min\{E^V(N-k)+E^0(k),\ k=1,...,N\}.$$
\end{theorem}

\begin{remark}
A similar result holds true if the system contains several kinds of particles (with possibly different interaction potentials), with or without internal degrees of freedom.\dqed
\end{remark}

Theorem~\ref{thm:HVZ_wlsc} is due to Zhislin \cite{Zhislin-60}, Van Winter \cite{VanWinter-64} and Hunziker \cite{Hun-66}.  Simpler proofs were provided later when the so-called geometric methods were developed  \cite{Enss-77,Simon-77,Sigal-82,CycFroKirSim-87}.
The interpretation of \eqref{eq:HVZ_wlsc} is that in order to reach the bottom of the essential spectrum, we have to provide a sufficiently large amount of energy to the system in order to extract at least one particle. The case of a general interaction $W$ is treated later in Section~\ref{sec:HVZ_general}.

Theorem~\ref{thm:HVZ_wlsc} is essential when proving existence of ground and excited states. The bottom of the spectrum $E^V(N)$ is an isolated eigenvalue if and only if the HVZ inequality $E^V(N)<E^V(N-1)$ holds. Such an inequality can be proved by induction on $N$: admitting that $E^V(N-1)<E^V(N-2)$, there is a ground state for $E^V(N-1)$ and one can use this state to construct an $N$-body test state to prove that $E^V(N)<E^V(N-1)$. 

For atoms and molecules (Example~\ref{ex:atoms_molecules}), Zhislin and Sigalov \cite{Zhislin-60,ZhiSig-65} have shown that there is a ground state as well as infinitely many excited states as soon as $N<Z+1$ where $Z=\sum_{m=1}^Mz_m$ is the total nuclear charge. The idea is that, with $N-1$ electrons bound to the nuclei, any additional electron escaping to infinity sees a Coulomb interaction induced by a total charge $Z-(N-1)>0$. This potential is attractive at large distances and the desired inequality $E^V(N)<E^V(N-1)$ follows.

We now turn to the proof of Theorem~\ref{thm:HVZ_wlsc}.

\begin{proof}
The bound $\Sigma^V(N)\leq E^V(N-1)$ is shown by building a convenient singular Weyl sequence, using a Weyl sequence for $E^V(N-1)$. We do not elaborate more on this classical fact and we only explain the proof of the more complicated inequality $\Sigma^V(N)\geq E^V(N-1)$.

First we note that since $E^V(N)\leq \Sigma^V(N)\leq E^V(N-1)$, the map $N\mapsto E^V(N)$ is non-increasing. When $V=0$, $E^0(N)\geq0$ since $W\geq0$, hence $E^0(N)=0$ for all $N$. 

Let now $\{\Psi_n\}\subset\gH^N$ be a singular Weyl sequence for $\Sigma^V(N)$, that is such that $(H^V(N)-\Sigma^V(N))\Psi_n\to0$, $\norm{\Psi_n}=1$ and $\Psi_n\wto0$ weakly in $L^2((\R^d)^N)$. The corresponding pure state on $\cFN$ is $\Gamma_n:=0\oplus\cdots0\oplus|\Psi_n\rangle\langle\Psi_n|$ and it has a bounded energy, $\lim_{n\to\ii}\cE^V(\Gamma_n)=\Sigma^V(N)$, hence a bounded kinetic energy by \eqref{eq:bound_Hamiltonian}.
Extracting a subsequence if necessary, we may assume by Lemma~\ref{lem:compact} that $\Gamma_n\gto\Gamma$ geometrically. We write as usual
$$\Gamma=\left(\begin{matrix}
G_{00} & &0\\
 & \ddots & \\
0 & & G_{NN}
\end{matrix}\right).$$
Recall that $G_{NN}=[G_{NN}]^{(N)}$ is the weak--$\ast$ limit of $|\Psi_n\rangle\langle\Psi_n|$, hence $G_ {NN}=0$ since $\Psi_n\wto0$. By Lemma~\ref{lem:wlsc} we have
\begin{align}
\Sigma^V(N)=\lim_{n\to\ii}\cE^V(\Gamma_n)\geq \cE^V(\Gamma)&=\sum_{j=0}^{N-1}\tr_{\gH^j}(H^V(j)G_{jj})\nonumber\\
&\geq \sum_{j=0}^{N-1}E^V(j)\,\tr_{\gH^j}(G_{jj})\geq E^V(N-1).\label{eq:estim_HVZ_wlsc_proof} 
\end{align}
In the second line we have used that $G_{jj}\geq0$ and that $\sum_{j=0}^{N-1}\tr_{\gH^j}(G_{jj})=1$ since $\Gamma$ is a state.
\end{proof}

\begin{remark}
Let $\{\Psi_n\}$ be a singular Weyl sequence for the bottom $\Sigma^V(N)$ of the essential spectrum of $H^V(N)$, like in the proof of Theorem~\ref{thm:HVZ_wlsc}. Then, if $E^V(N-1)<\Sigma^V(N-1)=E^V(N-2)$, it can be seen from  \eqref{eq:estim_HVZ_wlsc_proof} that its geometric limit $\Gamma$ is a ground state of $H^V(N-1)$ in $L^2_{a/s}((\R^d)^{N-1})$.\dqed
\end{remark}

\section{Geometric localization}\label{sec:geom_loc}

Localization is a fundamental concept of many-body quantum mechanics. In the seminal works of the end of the seventies \cite{Enss-77,DeiSim-77,Simon-77,MorSim-80,Sigal-82,HunSig-00}, the expression `geometric methods' was used to denote the use of appropriate partitions of unity in configuration space.
In this section we explain how one can lift a localization in the one-body space $\gH^1$ to the truncated Fock space $\cFN$, following Derezi\'nski and Gérard  \cite{DerGer-99}, and we relate this tool to the geometric topology defined in the previous section.

\subsection{Definition and properties}
\subsubsection{Definition}
Here we explain how to localize a state $\Gamma\in\cS(\cFN)$. As already suggested in the introduction, the localization of a pure one-body state $\phi\in L^2(\R^d)$ in a domain $D\subset\R^d$ should be described by the state 
\begin{equation}
\Gamma_\chi=\left(1-\int_{\R^d}|\chi\phi|^2\right)\,\oplus|\chi\phi\rangle\langle\chi\phi|
\label{eq:localized_one_body} 
\end{equation}
where $\chi=\1_D$.
Note that the previous formula actually defines a state for every normalized $\phi\in L^2(\R^d)$ and every function $\chi$ such that $0\leq|\chi|^2\leq 1$. This discussion suggests the following definition of localized states.

\begin{proposition}[Definition of localized states]\label{def:localization}
Let $B\in\cB(\gH)$ be a bounded operator on $\gH$, such that $0\leq BB^*\leq 1$, and $\Gamma\in\cS(\cFN)$ be any state on $\cFN$. Then there exists a unique state $\Gamma_B\in\cS(\cFN)$ such that
\begin{equation}
[\Gamma_B]^{(p,q)}=\underbrace{B\otimes\cdots\otimes  B}_{\text{$p$}}\;[\Gamma]^{(p,q)}\;\underbrace{B^*\otimes\cdots\otimes B^*}_{\text{$q$}}
\label{eq:def_localization_DM}
\end{equation}
for all $0\leq p,q\leq N$. The state $\Gamma_B$ is called the \emph{$B$-localization of $\Gamma$}.
\end{proposition}

Note that in general the localized state $\Gamma_B$ is not a pure state, even when $\Gamma$ is itself a pure state.
The concept of localization of states in Fock space was first introduced for bosons by Derezi\'nski and Gérard in \cite{DerGer-99} and generalized to fermions by Ammari in \cite{Ammari-04}. It is now a classical tool in Quantum Field Theory. It was recently used by Hainzl, Solovej and the author of the present paper, to prove the existence of the thermodynamic limit for quantum Coulomb systems in the grand canonical picture, see Appendix A.1 in \cite{HaiLewSol_2-09}. In this latter work, the strong subadditivity of the quantum entropy was also formulated using geometric localization. Although expressed in different terms, the definition of $\Gamma_B$ in Proposition~\ref{def:localization} coincides with that of all these previous works.

We now turn to the proof of Proposition~\ref{def:localization}.

\begin{proof}
A state satisfying \eqref{eq:def_localization_DM} was constructed in \cite{DerGer-99,Ammari-04,HaiLewSol_2-09}, using the partial isometry $f\in\gH\mapsto Bf\oplus\sqrt{1-BB^*}f\in\gH\oplus\gH$ and the fact that $\cF(\gH_1\oplus\gH_2)\simeq\cF(\gH_1)\otimes\cF(\gH_2)$. The state $\Gamma_B$ is obtained by means of a partial trace with respect to the second Hilbert space.
Uniqueness then follows from Lemma~\ref{lem:density_matrices}.
\end{proof}

\begin{remark}
The matrix components $\{G_{mn}^B\}_{m,n=0}^N$ of the operator $\Gamma_B$ can be expressed using Equation~\eqref{eq:relation_DM_G} as follows
\begin{equation}
G_{mn}^B=B^{\otimes m}[\Gamma]^{(m,n)}(B^*)^{\otimes n}+\sum_{j=1}^{\min(N-m,N-n)}(-1)^{j} \left[ B^{\otimes (m+j)}[\Gamma]^{(m+j,n+j)}(B^*)^{\otimes (n+j)}\right]^{(m,n)}.
\label{eq:relation_DM_localized}
\end{equation}
The verification that the so-obtained operator is a state ($\Gamma_B=(\Gamma_B)^*\geq0$) uses the CCR/CAR algebra $\cA$ in a similar way as in the proof of Lemma~\ref{lem:compact}.\dqed
\end{remark}

\begin{remark}\label{rmk:localized_AB}
If $B_1$ and $B_2$ are such that $0\leq B_kB_k^*\leq1$, then $(B_2B_1)(B_2B_1)^*=B_2B_1B_1^*B_2^*\leq B_2B_2^*\leq1$. It is then clear from the definition that $(\Gamma_{B_1})_{B_2}=\Gamma_{B_2B_1}$.\dqed
\end{remark}

We now illustrate Propostion~\ref{def:localization} by several examples of localized states. 

\begin{example}
We have for all state $\Gamma_1=\Gamma$ and $\Gamma_0=|\Omega\rangle\langle\Omega|$ (the vacuum state), corresponding to having, respectively, $B=1$ and $B=0$. If $\phi\in\gH^1$ and $\Gamma=0\oplus|\phi\rangle\langle\phi|$, then
$\Gamma_B=(1-\norm{B\phi}^2)\,\oplus|B\phi\rangle\langle B\phi|$,
as in \eqref{eq:localized_one_body}.
\dqed 
\end{example}

\begin{example}
If $U$ is a unitary operator on $\gH^1$, then 
$(\Gamma)_U=(1\oplus U\oplus\cdots \oplus U^{\otimes N})\,\Gamma\,(1\oplus U^*\oplus\cdots \oplus (U^*)^{\otimes N}).$
\dqed 
\end{example}

\begin{example}[Localization of $N$-body states]\label{ex:localization_N_body}
Let $G\in\cS(\gH^N)$ be an $N$-body state and $\Gamma=0\oplus\cdots\oplus G\in\cS(\cFN)$. A simple calculation based on~\eqref{eq:relation_DM_localized} shows that 
$$\Gamma_B=G_{0}^B \oplus\cdots \oplus G_{N}^B$$
where
\begin{equation}
G_{k}^B={N\choose k}\;\tr_{k+1\to N}\left(B^{\otimes k}\otimes\sqrt{1-BB^*}^{\otimes (N-k)}\,G\,(B^*)^{\otimes k}\otimes\sqrt{1-BB^*}^{\otimes (N-k)}\right)
\label{eq:localization_N_body} 
\end{equation}
with $\tr_{k+1\to N}$ denoting the partial trace with respect to the $N-k+1$ last variables. More explicitely, if $G=|\Psi\rangle\langle\Psi|$ and $0\leq \chi(x)\leq1$, then
\begin{multline}
G^\chi_{k}(x_1,...,x_k;x'_1,...,x'_k)={N\choose k}\prod_{j=1}^k\chi(x_j)\chi(x'_j)\,\int\cdots\int\prod_{j=k+1}^N\left(1-\chi^2(z_j)\right)\times\\
\times{\Psi(x_1,...,x_k,z_{k+1},...,z_N)}\overline{\Psi(x'_1,...,x'_k,z_{k+1},...,z_N)} \,dz_{k+1}\cdots dz_N.\label{eq:localization_N_body_fn}
\end{multline}

We see from \eqref{eq:localization_N_body} that it holds
\begin{equation}
\boxed{\tr_{\gH^k}\left(G^B_k\right)=\tr_{\gH^{N-k}}\left(G^{\sqrt{1-BB^*}}_{N-k}\right)}
\label{eq:relation_trace_localized_N_body_fn}
\end{equation} 
The relation \eqref{eq:relation_trace_localized_N_body_fn} will play a very important role later and it may be considered as one of the basic tools of the geometric methods for many-body systems. For $B=\1_D(x)$, it essentially means that the `weight' in the $k$-particle sector of the localized state in a domain $D$ is equal to that in the $(N-k)$-particle sector outside $D$.\dqed
\end{example}

\begin{example}[Hartree states]\label{ex:Hartree_loc}
Let $\Gamma=0\oplus\cdots\oplus|\phi^{\otimes N}\rangle\langle\phi^{\otimes N}|\in\cF^{\leq N}_s$ be a Hartree state as in Example~\ref{ex:Hartree}. Then
$$\Gamma_B= \bigoplus_{k=0}^N{N\choose k}\left(1-\norm{B\phi}_\gH^2\right)^{N-k}\big|(B\phi)^{\otimes k}\big\rangle\big\langle(B\phi)^{\otimes k}\big|.$$
\dqed
\end{example}

\begin{example}[Coherent and Hartree-Fock-Bogoliubov states]
If $\Gamma_f$ is a coherent state like in Example~\ref{ex:coherent}, then $(\Gamma_f)_B=\Gamma_{Bf}$.
If $\Gamma$ is a Hartree-Fock-Bogoliubov state like in Example~\ref{ex:HFB}, with one-body density matrix $[\Gamma]^{(1)}$ and pairing density matrix $[\Gamma]^{(2,0)}$, then $\Gamma_B$ is the unique Hartree-Fock-Bogoliubov state having $B[\Gamma]^{(1)}B^*$ and $(B\otimes B)\,[\Gamma]^{(2,0)}$ as one-body and pairing density matrices. In Example~\ref{ex:localization_HF} below we detail the case of pure Hartree-Fock states.\dqed
\end{example}

\subsubsection{Convergence results}
Let us now turn to some useful applications of geometric localization. We start by showing that the localization map $\Gamma\mapsto \Gamma_B$ is continuous with respect to the geometric topology.

\begin{lemma}[Continuity of geometric localization]\label{lem:continuity}
Let $\{\Gamma_n\}$ be a sequence of states in $\cS(\cFN)$ which converges geometrically to a state $\Gamma\in\cS(\cFN)$, $\Gamma_n\wto_g\Gamma$. Let $B\in\cB(\gH^1)$ be such that $0\leq BB^*\leq 1$. Then the associated sequence of localized states converges geometrically: $(\Gamma_n)_B\wto_g\Gamma_B$.

Similarly, if $B_n$ is a sequence satisfying $0\leq B_n(B_n)^*\leq 1$, $B_n\to B$ and $(B_n)^*\to B^*$ strongly (that is $B_nx\to Bx$ and $B_n^*x\to B^*x$ strongly in $\gH^1$ for any fixed $x\in\gH^1$), then it holds $(\Gamma_n)_{B_n}\wto_g\Gamma_B$.
\end{lemma}
\begin{proof}
When $\Gamma_n\wto_g\Gamma$, that is $[\Gamma_n]^{(p,q)}\wto_\ast[\Gamma]^{(p,q)}$ for all $0\leq p,q\leq N$, we have that $[(\Gamma_n)_B]^{(p,q)}=B^{\otimes p}[\Gamma_n]^{(p,q)}(B^*)^{\otimes q}$ converges weakly--$\ast$ to $B^{\otimes p}[\Gamma]^{(p,q)}(B^*)^{\otimes q}$. This is by definition $[\Gamma_B]^{(p,q)}$, hence it holds  $(\Gamma_n)_B\wto_g\Gamma_B$. The argument is the same when $B_n\to B$ and $(B_n)^*\to B^*$ strongly.
\end{proof}

The next lemma explains how localization can be used to convert geometric convergence into strong convergence.

\begin{lemma}[Local compactness]\label{lem:local_compactness}
Let $T\geq0$ be a non-negative self-adjoint operator on $\gH^1$, and $B$ be a bounded operator such that $0\leq BB^*\leq 1$. We assume that $B$ is $T^{1/2}$-compact, that is $B(1+T^{1/2})^{-1}\in\cK(\gH^1)$. Let $\{\Gamma_n\}$ be a sequence of states in $\cS(\cFN)$ which converges geometrically to a state $\Gamma\in\cS(\cFN)$, $\Gamma_n\gto\Gamma$. If 
$$\tr_\gH\big(T^{1/2}[\Gamma_n]^{(1)}T^{1/2}\big)\leq C$$
for a constant independent of $n$, then $(\Gamma_n)_B\to\Gamma_B$ strongly in $\gS_1(\cFN)$.
\end{lemma}
\begin{proof}
We have $(\Gamma_n)_B\wto_g\Gamma_B$ geometrically by Lemma \ref{lem:continuity} and it remains to prove that the convergence is strong. It holds
$$\big[(\Gamma_n)_B\big]^{(1)}=B[\Gamma_n]^{(1)}B^*=K(1+T^{1/2})[\Gamma_n]^{(1)}(1+T^{1/2})K^*$$
where $K=B(1+T^{1/2})^{-1}$ is compact by assumption. The sequence $(1+T^{1/2})[\Gamma_n]^{(1)}(1+T^{1/2})$ is bounded in $\gS_1(\gH)$, hence we have that $(1+T^{1/2})[\Gamma]^{(1)}(1+T^{1/2})\in\gS_1(\gH^1)$ and
$$(1+T^{1/2})[\Gamma_n]^{(1)}(1+T^{1/2})\wto_\ast (1+T^{1/2})[\Gamma]^{(1)}(1+T^{1/2})$$
weakly--$\ast$ in $\gS_1$. It is well known that if $A_n\wto A$ weakly--$\ast$ in $\gS_1(\gH^1)$ and $K$ is compact, then $KA_nK^*\to KAK^*$ strongly in $\gS_1(\gH)$. We deduce from the above calculation that 
$[(\Gamma_n)_B]^{(1)}\to [\Gamma_B]^{(1)}$
strongly in $\gS_1(\gH)$. By Lemma~\ref{lem:N_wlsc}, this shows that $(\Gamma_n)_B\to\Gamma_B$ strongly.
\end{proof}

\begin{example}\label{ex:local_compactness}
If $\Gamma_n\wto_g\Gamma$ geometrically in $\cFN$ and the kinetic energy $\tr\left((-\Delta)[\Gamma_n]^{(1)}\right)$ is uniformly bounded, then $(\Gamma_n)_\chi\to\Gamma_\chi$ strongly in $\gS_1$, for every localization function $\chi(x)$ of compact support (even tending to zero at infinity), since $\chi(x)(1+|-i\nabla|)^{-1}$ is always a compact operator. This can be viewed as a generalization to states in $\cFN$ of Rellich's local compactness in Sobolev spaces \cite{LieLos-01}.\dqed
\end{example}

The following is simple consequence of the previous result with $T=1$.

\begin{corollary}[Compact localization]\label{cor:compact_localization_operators}
Let $\{\Gamma_n\}$ be a sequence of states in $\cS(\cFN)$ which converges geometrically to a state $\Gamma\in\cS(\cFN)$, $\Gamma_n\wto_g\Gamma$. Then $(\Gamma_n)_K\to(\Gamma)_K$ strongly in $\gS_1(\cFN)$ for every fixed compact operator $K$ such that  $0\leq KK^*\leq1$.
\end{corollary}

Localization may also be used to approximate a given state by simpler states (for instance finite rank states, see Section~\ref{sec:HF_MCSCF}).
\begin{lemma}[Approximation by localized states]\label{lem:approx_localized}
Let $\{B_n\}$ be a sequence of bounded operators in $\gH$, such that $0\leq B_nB_n^*\leq1$, $B_n\to B$ and $B_n^*\to B^*$ strongly as $n\to\ii$. Then for any state $\Gamma\in\cS(\cFN)$, $\Gamma_{B_n}\to \Gamma_B$ strongly in $\gS_1(\gH)$ as $n\to\ii$.
\end{lemma}
\begin{proof}
By Lemma \ref{lem:continuity}, we have at least $\Gamma_{B_n}\wto_g \Gamma_B$ geometrically. However, since
$[\Gamma_{B_n}]^{(1)}=(B_n)[\Gamma]^{(1)}(B_n)^*\to B[\Gamma]^{(1)}B^*=[\Gamma_{B}]^{(1)}$ strongly in $\gS_1(\gH^1)$, the convergence of $\Gamma_n$ must be strong by Lemma~\ref{lem:N_wlsc}.
\end{proof}

\subsection{Application: HVZ theorem in the general case}\label{sec:HVZ_general}

In Section~\ref{sec:HVZ_wlsc} we have proved the celebrated HVZ Theorem for systems with a repulsive interaction, $W\geq0$, using the lower semi-continuity of the energy with respect to geometric convergence. In particular it was essential that in the absence of external field, $V=0$, the ground state energy of the system vanishes: $E^0(N)=0$. 

When $W$ has no sign \emph{a priori}, the energy $\Gamma\mapsto\cE^V(\Gamma)$ is not necessarily lower semi-continuous, which can be seen by the fact that it may hold $E^0(N)<0$. Particles running off to infinity can carry a negative energy and in the HVZ theorem it is then necessary to take into account the energy of these particles. Separating the particles escaping to infinity from those which are bound by the external potential $V$ is then done via localization.

Let us recall the $N$-body Hamiltonian $H^V(N)$ defined in \eqref{def:Hamiltonian_W}. 
The bottom of its spectrum and of its essential spectrum are respectively denoted by $E^V(N)$ and $\Sigma^V(N)$. 
As usual we make the assumption that $W$ is even, and that $V$ and $W$ can both be written in the form $\sum_{i=1}^Kf_i$ with $f_i\in L^{p_i}(\R^d)$ where $\max(1,d/2)<p_i<\ii$ or $p_i=\ii$ but $f_i\to0$ at infinity.

The result in the general case is the following.
\begin{theorem}[HVZ in the general case]\label{thm:HVZ}
Under the previous assumptions on $V$ and $W$, we have
\begin{equation}
\boxed{\Sigma^V(N)=\inf\big\{E^V(N-k)+E^0(k),\ k=1,...,N\big\}.}
\label{eq:HVZ}
\end{equation}
\end{theorem}

We now provide the proof of Theorem~\ref{thm:HVZ}. This serves as an illustration of the concepts of geometric convergence and localization that we have introduced, but also introduces the reader to the techniques that we use later for nonlinear systems.

\begin{proof}
As in the proof of Theorem~\ref{thm:HVZ_wlsc}, we only explain the lower bound $\geq$. We take the same singular Weyl sequence $\{\Psi_n\}$ such that $(H^V(N)-\Sigma^V(N))\Psi_n\to0$ and let $\Gamma_n=0\oplus\cdots\oplus|\Psi_n\rangle\langle\Psi_n|\in\cFN$. We assume (up to extraction of a subsequence and by Lemma~\ref{lem:compact}) that $\Gamma_n\wto_g\Gamma=G_{00}\oplus\cdots\oplus G_{NN}$ geometrically. Recall that $G_{NN}$ is the weak limit of $|\Psi_n\rangle\langle\Psi_n|$, hence $G_{NN}=0$ since $\Psi_n\wto0$ by assumption.

Our goal is to prove the following fundamental estimate
\begin{equation}
{\Sigma^V(N)=\lim_{n\to\ii}\cE^V(\Gamma_n)\geq \sum_{k=1}^N\Big(E^V(N-k)+E^0(k)\Big)\tr_{\gH^{N-k}}\left(G_{N-k\,N-k}\right).}
\label{eq:final_estimate_HVZ}
\end{equation}
Compared to \eqref{eq:estim_HVZ_wlsc_proof}, the bound now includes the energy $E^0(k)$ of particles running off to infinity, which can be nonzero.
Recall that $\Gamma$ is a state, that is $G_{kk}\geq0$ and $\sum_{k=1}^N\tr_{\gH^{N-k}}\left(G_{N-k\,N-k}\right)=1$ since $G_{NN}=0$. Therefore the right hand side of \eqref{eq:final_estimate_HVZ} is a convex combination and we have
$$\sum_{k=1}^N\Big(E^V(N-k)+E^0(k)\Big)\tr_{\gH^{N-k}}\left(G_{N-k\,N-k}\right)\geq \inf\left\{E^V(N-k)+E^0(k),\ k=1,...,N\right\},$$
which proves the lower bound in \eqref{eq:HVZ}.

In order to show the inequality \eqref{eq:final_estimate_HVZ}, we pick a smooth cut-off function $0\leq \chi\leq1$ which equals 1 on the ball $B(0,1)$ and 0 outside the ball $B(0,2)$, and let $\chi_R(x)=\chi(x/R)$ as well as $\eta_R=\sqrt{1-\chi_R^2}$. 
The rest of the proof goes as follows: 
\begin{enumerate}
 \item[\textit{(i)}] We geometrically localize in and outside the ball of radius $R$ by means of the smooth partition of unity $\chi_R^2+\eta_R^2=1$;
\item[(ii)] We use the fundamental equality \eqref{eq:relation_trace_localized_N_body_fn};
\item[\textit{(iii)}] We pass to the limit as $n\to\ii$;
\item[\textit{(iv)}] We take the limit $R\to\ii$.
\end{enumerate}
As we will explain later in the proof of Theorem \ref{thm:nonlinear}, it is possible to use an $n$-dependent radius of localization $R_n\to\ii$, and to perform the steps \textit{(iii)} and \textit{(iv)} simultaneously. As we do not need this technique here, we defer its use to Section \ref{sec:nonlinear}, for pedagogical purposes.

The so-called IMS formula reads:
\begin{equation}
-\Delta=\chi_R(-\Delta)\chi_R+\eta_R(-\Delta)\eta_R-|\nabla\chi_R|^2-|\nabla\eta_R|^2.
\label{eq:IMS} 
\end{equation}
Hence $-\Delta\geq \chi_R(-\Delta)\chi_R+\eta_R(-\Delta)\eta_R-C/R^2$. Using this for the kinetic energy as well as the partition of unity $1=\chi_R^2+\eta_R^2$ in the interaction energy, we deduce that
\begin{multline*}
\cE^V(\Gamma_n)\geq \cE^V\big((\Gamma_n)_{\chi_R}\big) +\cE^0\big((\Gamma_n)_{\eta_R}\big)+\int_{\R^d}\eta_R(x)^2V(x)\rho_{\Gamma_n}(x)\,dx\\
+2\int_{\R^d}\int_{\R^d}W(x-y)\chi_R(x)^2\eta_R(y)^2[\Gamma_n]^{(2)}(x,y;x,y)\,dx\,dy-CN/R^2.
\end{multline*}
Let us start by estimating the error terms. Since $\{\Psi_n\}$ is a Weyl sequence it is bounded in $H^1((\R^d)^N)$, thus  $\sqrt{\rho_{\Gamma_n}}$ is bounded in $H^1(\R^d)$, by \eqref{eq:Hoffmann-Ostenhof}. Since $V=\sum_{j=1}^KV_j$ with $V_j\in L^{p_j}(\R^d)$ where $\max(1,d/2)<p_j<\ii$ or $p_j=\ii$ but $V_j\to0$ at infinity, we have by Hölder's and Sobolev's inequalities
$$\left|\int_{\R^d}\eta_R(x)^2V(x)\rho_{\Gamma_n}(x)\,dx\right|\leq C\sum_{j=1}^k\norm{V_j\eta_R^2}_{L^{p_j}(\R^d)}$$
which tends to zero as $R\to\ii$. For the interaction term, we may write for instance
\begin{align}
&\int_{\R^d}\int_{\R^d}W(x-y)\chi_R(x)^2\eta_R(y)^2[\Gamma_n]^{(2)}(x,y;x,y)\,dx\,dy\nonumber\\ &\qquad\qquad=\int_{\R^d}\int_{\R^d}W(x-y)\chi_R(x)^2\eta_{3R}(y)^2[\Gamma_n]^{(2)}(x,y;x,y)\,dx\,dy\label{eq:loc_inter_2}\\
&\qquad\qquad+\int_{\R^d}\int_{\R^d}W(x-y)\chi_R(x)^2\eta_{R}(y)^2\chi_{3R}(y)^2[\Gamma_n]^{(2)}(x,y;x,y)\,dx\,dy.\label{eq:loc_inter_3}
\end{align}
In the first term of the right hand side, the integrand is zero except when $|x-y|\geq R$, hence it may be estimated similarly as before by
$$|\eqref{eq:loc_inter_2}|\leq C\sum_{j=1}^K\norm{W_j\1_{|x|\geq R}}_{L^{p_j}(\R^d)}$$
which also tends to zero when $R\to\ii$. Summarizing we have shown that
\begin{multline}
 \cE^V(\Gamma_n)\geq \cE^V\big((\Gamma_n)_{\chi_R}\big) +\cE^0\big((\Gamma_n)_{\eta_R}\big)\\
+2\int_{\R^d}\int_{\R^d}W(x-y)\chi_R(x)^2\eta_{R}(y)^2\chi_{3R}(y)^2[\Gamma_n]^{(2)}(x,y;x,y)\,dx\,dy+\epsilon_R
\end{multline}
where $\epsilon_R$ is independent of $n$ and tends to zero as $R\to\ii$. The total energy of the system can be estimated from below by the sum of the energies of the localized states in and outside the ball of radius $R$, plus error terms.

We now deal with the main two terms and write that
\begin{align}
\cE^V\big((\Gamma_n)_{\chi_R}\big) +\cE^0\big((\Gamma_n)_{\eta_R}\big)&=\sum_{k=0}^N\tr_{\gH^k}\left(H^V(k)G_{\chi_R,k}^n\right)+\sum_{k=0}^N\tr_{\gH^k}\left(H^0(k)G_{\eta_R,k}^n\right)\nonumber\\
&\geq \sum_{k=0}^NE^V(k)\tr_{\gH^k}\left(G_{\chi_R,k}^n\right)+\sum_{k=0}^NE^0(k)\tr_{\gH^k}\left(G_{\eta_R,k}^n\right),\label{eq:decomp_energy_HVZ}
\end{align}
where $(\Gamma_n)_{\chi_R}=G_{\chi_R,0}^n\oplus\cdots \oplus G_{\chi_R,N}^n$ and with a similar definition for $G_{\eta_R,k}^n$.
At this point we use the fundamental relation \eqref{eq:relation_trace_localized_N_body_fn} (valid since $\Gamma_n$ is an $N$-body state for all $n$), which tells us that
$$\tr_{\gH^{k}}\left(G_{\chi_R,k}^n\right)=\tr_{\gH^{N-k}}\left(G_{\eta_R,N-k}^n\right)$$
for all $k=0,...,N$. Inserting in \eqref{eq:decomp_energy_HVZ} and changing $k$ into $N-k$ in the first sum we get
$$\cE^V\big((\Gamma_n)_{\chi_R}\big) +\cE^0\big((\Gamma_n)_{\eta_R}\big)\geq 
\sum_{k=0}^N\Big(E^V(N-k)+E^0(k)\Big)\tr_{\gH^{N-k}}\left(G_{\chi_R,N-k}^n\right).$$
By Lemma~\ref{lem:local_compactness} (or more precisely Example~\ref{ex:local_compactness}), we have $(\Gamma_n)_{\chi_R}\to\Gamma_{\chi_R}$ strongly, therefore 
$$\lim_{n\to\ii}\tr_{\gH^{N-k}}\left(G_{\chi_R,N-k}^n\right)=\tr_{\gH^{N-k}}\left(G_{\chi_R,N-k}\right)$$
where $\Gamma_{\chi_R}=G_{\chi_R,0}\oplus\cdots \oplus G_{\chi_R,N}$. Recall $G_{NN}=0$ hence $G_{\chi_R,N}=(\chi_R)^{\otimes N}G_{NN}(\chi_R)^{\otimes N}=0$ also. 
As a consequence,
\begin{multline*}
\lim_{n\to\ii}\sum_{k=0}^N\Big(E^V(N-k)+E^0(k)\Big)\tr_{\gH^{N-k}}\left(G_{\chi_R,N-k}^n\right)\\
=\sum_{k=1}^N\Big(E^V(N-k)+E^0(k)\Big)\tr_{\gH^{N-k}}\left(G_{\chi_R,N-k}\right).
\end{multline*}
Using that the term in \eqref{eq:loc_inter_3} converges as $n\to\ii$ since $\chi_R(x)^2\eta_R(y)^2\chi_{3R}(y)^2$ has a compact support, we arrive at the estimate
\begin{multline}
\Sigma^V(N)=\lim_{n\to\ii}\cE^V(\Gamma_n)\geq \sum_{k=1}^N\Big(E^V(N-k)+E^0(k)\Big)\tr_{\gH^{N-k}}\left(G_{\chi_R,N-k}\right)\\
+2\int_{\R^d}\int_{\R^d}W(x-y)\chi_R(x)^2\eta_{R}(y)^2\chi_{3R}(y)^2[\Gamma]^{(2)}(x,y;x,y)\,dx\,dy+\epsilon_R.
\end{multline}
Passing finally to the limit $R\to\ii$ (using that $\Gamma_{\chi_R}\to\Gamma$ strongly by Lemma~\ref{lem:approx_localized}, hence $G_{\chi_R,k}\to G_{kk}$ as $R\to\ii$) gives the desired estimate \eqref{eq:final_estimate_HVZ} and ends the proof.
\end{proof}

\section{Finite rank approximation of many-body systems}\label{sec:HF_MCSCF}

In the previous two sections we have introduced geometric tools for many-body systems and we have illustrated their use on linear systems (the HVZ theorem). In practice, physicists and chemists resort to \emph{approximate models} which are simpler to handle and to simulate numerically. These approximations are usually classified in two different categories:
\begin{itemize}
\item those in which the set of states is reduced,

\medskip

\item those in which the energy is modified by adding nonlinear empirical terms.
\end{itemize}
These two methods can of course be combined: in the so-called Kohn-Sham method of atoms and molecules \cite{KohSha-65}, all states are assumed to be of Hartree-Fock type but the energy is further modified to take into account exchange-correlation effects. Both techniques usually lead to \emph{nonlinear models}, either because the class of states is replaced by a manifold or because the energy is itself nonlinear.

The purpose of this section is to study methods of the first kind in which the many-body energy is kept linear, but the set of states is reduced. Methods from the second category will be considered in Section~\ref{sec:nonlinear}. We study here the so-called \emph{finite rank approximation} which consists in assuming that the $N$-body wavefunction can be expanded as tensor products of finitely many unknown one-body functions $\{\phi_1,...,\phi_r\}$. For fermions, this leads to the celebrated \emph{Hartree-Fock method} \cite{LieSim-77} when $r=N$, and to the widely used \emph{multiconfiguration methods} \cite{Friesecke-03,Lewin-04a} when $r>N$. For bosons, the \emph{Hartree method} is obtained when $r=1$.

We investigate properties of geometric limits of finite-rank states, and deduce \emph{nonlinear} versions of the HVZ Theorem. As we will see, the situation is however still rather unclear for bosons and our results are only satisfactory for fermions in the Hartree-Fock approximation or for multiconfiguration methods with repulsive interactions. We hope to come back to the other interesting cases in the future.

\subsection{States living on a subspace of $\gH$, finite rank states}

\subsubsection{Definitions}

\begin{definition}[States living on a subspace]\label{def:subspace}
Let $\gH'\subset\gH$ be a closed subspace of the one-body space $\gH^1$ and $P$ be the orthogonal projection onto $\gH'$. A state $\Gamma\in\cS(\cFN)$ is said to \emph{live on $\gH'$} when $\Gamma_P=\Gamma$.
\end{definition}

The smallest subspace $\gH'$ such that $\Gamma$ lives on $\gH'$ can be called the \emph{support of $\Gamma$}. The following is a reformulation of a result of Löwdin \cite{Lowdin-55a} stating that the support can be found by means of the one-body density matrix $[\Gamma]^{(1)}$ only.

\begin{lemma}[Löwdin's criterion]\label{lem:Lowdin}
Let $\Gamma$ be a state on $\cFN$ and $P:\gH^1\to\gH^1$ be an orthogonal projector. The following assertions are equivalent:
\begin{enumerate}
\item $\Gamma$ lives on $P\gH^1$, that is $\Gamma_P=\Gamma$;
\item $P[\Gamma]^{(1)}P=[\Gamma]^{(1)}$;
\item $\mathbb{P}\Gamma\mathbb{P}=\Gamma$ where $\mathbb{P}=1\oplus P\oplus (P\otimes P)\oplus\cdots \oplus P^{\otimes N}$.
\end{enumerate}
\end{lemma}

\begin{proof}
It is clear from the definition of geometric localization that \textit{(1)} implies \textit{(2)}. If we denote by $G_{k\ell}$ the matrix elements of $\Gamma$, \textit{(3)} means that $P^{\otimes k}G_{k\ell}P^{\otimes \ell}=G_{k\ell}$ for every $0\leq k,\ell\leq N$. Using \eqref{eq:triangular_system}, this is easily seen to imply that $P^{\otimes p}[\Gamma]^{(p,q)}P^{\otimes q}=[\Gamma]^{(p,q)}$ for all $0\leq p,q\leq N$, hence $\Gamma_P=\Gamma$ and \textit{(1)} holds true.

It therefore only remains to show that \textit{(2)} implies \textit{(3)}. We denote as usual by $G_{k\ell}$ the matrix elements of $\Gamma$ and note that, by \eqref{eq:triangular_system}, 
$$[\Gamma]^{(1)}=\sum_{k=1}^N[G_{kk}]^{(1)}.$$
Our assumption that $P[\Gamma]^{(1)}P=[\Gamma]^{(1)}$ implies that 
\begin{equation}
P[G_{kk}]^{(1)}P=[G_{kk}]^{(1)}\qquad \text{for all $k=1,...,N$.}
\label{eq:DM1_vanish}
\end{equation}
 Indeed, we have $P^\perp[\Gamma]^{(1)}P^\perp=0=\sum_{k=1}^NP^\perp[G_{kk}]^{(1)}P^\perp$ where $P^\perp=1-P$. Since $[G_{kk}]^{(1)}\geq0$ for all $k=1,...,N$, this implies that $P^\perp[G_{kk}]^{(1)}P^\perp=0$. Now \eqref{eq:DM1_vanish} follows for instance from the fact that
\begin{equation}
\left(P^\perp [G_{kk}]^{(1)}\right)\left(P^\perp [G_{kk}]^{(1)}\right)^*\leq \norm{[G_{kk}]^{(1)}}\, P^\perp[G_{kk}]^{(1)}P^\perp=0
\end{equation}
which shows that $P^\perp [G_{kk}]^{(1)}= [G_{kk}]^{(1)}P^\perp=0$. 

We now prove that \eqref{eq:DM1_vanish} implies that $P^{\otimes k}G_{kk}P^{\otimes k}=G_{kk}$. We have for any $P_2,...,P_k\in\{P,P^\perp\}$,
\begin{multline*}
\tr_{\gH^k}\left(P^\perp\otimes P_2\otimes\cdots P_k\, G_{kk}\, P^\perp\otimes P_2\otimes\cdots P_k\right)\\
\leq  \tr_{\gH^k}\left(P^\perp\otimes 1\otimes\cdots 1\, G_{kk}\, P^\perp\otimes 1\otimes\cdots 1\right)=\frac{1}{k}\tr_{\gH^1}\left(P^\perp [G_{kk}]^{(1)}\right)=0,
\end{multline*}
by \eqref{eq:DM_N_body}. The argument is the same if $P^\perp$ is not in the first place of the tensor product but appears at another position.
Arguing as before, this implies $P^{\otimes k}G_{kk}P^{\otimes k}=G_{kk}$. For the off-diagonal terms, we have 
$G_{k\ell}G_{\ell k}\leq G_{kk}$
since $0\leq \Gamma\leq1$. This can be used to show that $P^\perp\otimes P_2\otimes\cdots P_k\, G_{k\ell}=0$, hence $P^{\otimes k}G_{k\ell}P^{\otimes\ell}=G_{k\ell}$.
\end{proof}

We now use the previous concept to define finite-rank states.

\begin{definition}[Finite rank states]\label{def:finite_rank}
A state $\Gamma\in\cS(\cFN)$ is said to \emph{have a finite rank} when it lives on a subspace of finite dimension, that is when there exists a projector $P$ of finite rank such that $\Gamma_P=\Gamma$. The \emph{rank of $\Gamma$} is then defined as
$$\rank(\Gamma)=\min\{{\rm rank}(P)\ :\ \Gamma_P=\Gamma,\ P^2=P=P^*\}=\rank\left([\Gamma]^{(1)}\right).$$
The last equality follows from Lemma~\ref{lem:Lowdin}.
\end{definition}

\begin{example}[Coherent, Hartree and Hartree-Fock states]
For bosons, both the Hartree state $|\phi^{\otimes N}\rangle$ and the coherent state $W(f)|\Omega\rangle$ have rank $r=1$. For fermions, a pure Hartree-Fock state $\phi_1\wedge\cdots \wedge\phi_N$ has rank $r=N$.\dqed
\end{example}

The following says that finite-rank states are dense in $\cS(\cFN)$.

\begin{lemma}[Approximation by finite-rank states]
Any state $\Gamma\in\cS(\cFN)$ is a strong limit of finite rank states.
\end{lemma}
\begin{proof}
Let $\{\phi_j\}$ be an orthonormal basis of $\gH$ and $P_n:=\sum_{j=1}^n|\phi_j\rangle\langle\phi_j|$. Then $P_n\to1$ strongly in $\gH$. Therefore by Lemma~\ref{lem:approx_localized}, it holds $\Gamma_{P_n}\to\Gamma$ strongly. But $\Gamma_{P_n}$ has finite rank since $\big(\Gamma_{P_n}\big)_{P_n}=\Gamma_{(P_n)^2}=\Gamma_{P_n}$ by Remark~\ref{rmk:localized_AB}.
\end{proof}

We now show that any state of finite rank is a finite linear combination of monomials in the creation and annihilation operators.

\begin{lemma}[Expansion of finite rank states]\label{lem:Lowdin2}
Assume that $\Gamma_P=\Gamma$ for some orthogonal projector $P=\sum_{j=1}^r|\phi_j\rangle\langle\phi_j|$ of finite rank $r$, and let $(G_{k\ell})_{1\leq k,\ell\leq N}$ be the matrix elements of $\Gamma$. Then each $G_{k\ell}$ can be expanded as follows:
\begin{equation}
G_{k\ell}=\sum_{\substack{I=\{i_1\leq\cdots\leq i_k\}\subset\{1,...,r\}\\ J=\{j_1\leq\cdots\leq j_\ell\}\subset \{1,...,r\}}}c_{IJ}\;a^\dagger(\phi_{i_1})\cdots a^\dagger(\phi_{i_k})|\Omega\rangle\langle\Omega|a(\phi_{j_\ell})\cdots a(\phi_{j_1})
\end{equation}
for some $c_{IJ}\in\C$.
\end{lemma}

\begin{proof}
This follows from the fact that $P^{\otimes k}G_{k\ell}P^{\otimes \ell}=G_{k\ell}$, see \textit{(3)} in Lemma~\ref{lem:Lowdin}.
\end{proof}

Consider a finite rank state, that is such that $[\Gamma]^{(1)}$ has finite rank $r$ (Lemma~\ref{lem:Lowdin}). Then we can write $[\Gamma]^{(1)}=\sum_{j=1}^r n_j|\phi_j\rangle\langle\phi_j|$ for an orthonormal system $\{\phi_j\}_{j=1}^r$ of eigenvectors of $[\Gamma]^{(1)}$. The $n_j$ are usually called the \emph{occupation numbers} and the $\phi_j$ the \emph{natural orbitals of $\Gamma$}.
Lemma~\ref{lem:Lowdin2} then shows that any finite rank state can be expanded by means of its \emph{natural orbitals}. This is the original version of Löwdin's Expansion Theorem \cite{Lowdin-55a} (see also Lemma 1.1 (ii) in \cite{Friesecke-03} and Lemma 1 in \cite{Lewin-04a}).

The simplest example is that of a state of the form $\Gamma=0\oplus\cdots\oplus|\Psi\rangle\langle\Psi|$, that is a pure $N$-body state. Then if $\rank([\Gamma]^{(1)})\leq r$ and $\{\phi_1,...,\phi_r\}$ is an associated orthonormal system of natural orbitals, it holds
$$\Psi=\sum_{1\leq i_1\leq\cdots\leq i_N\leq r}c_{i_1,...,i_N}\phi_{i_1}\circ\cdots \circ\phi_{i_N}$$
where $\circ=\wedge$ (fermions) or $\vee$ (bosons).

\subsubsection{Geometric properties of finite rank states}\label{sec:geom_prop}
We now turn to properties of finite-rank states with regard to geometric localization and convergence. The following is a simple consequence of the characterization of the rank in terms of the one-body density matrix (Lemma~\ref{lem:Lowdin}).

\begin{lemma}[Localization and geometric limit of finite rank states]\label{lem:loc_geom_finite_rank}

\textit{(1)} If a state $\Gamma\in\cS(\cFN)$ has rank $\leq r$, then for every localization operator $B$, $0\leq BB^*\leq1$, the corresponding localized state $\Gamma_B$ has rank $\leq r$.

\smallskip

\textit{(2)} If $\{\Gamma_n\}$ is a sequence of states on $\cFN$ of rank $\leq r$ and $\Gamma_n\wto_g\Gamma$ geometrically, then $\Gamma$ has rank $\leq r$.
\end{lemma}

\begin{proof}
The result follows from the fact that when $\rank([\Gamma]^{(1)})\leq r$, then $\rank\, (B[\Gamma]^{(1)}B^*)\leq r$ for every localization operator $B$. Similarly, when $[\Gamma_n]^{(1)}\wto_\ast [\Gamma]^{(1)}$ weakly--$\ast$ in $\gH$, then $\rank( [\Gamma]^{(1)})\leq \liminf_{n\to\ii}\rank([\Gamma_n]^{(1)})$.
\end{proof}

For $N$-body systems we often have to study sequences of states of the form $\Gamma_n=0\oplus\cdots\oplus|\Psi_n\rangle\langle\Psi_n|$. When $\Gamma_n\wto_g\Gamma$ geometrically and when each $\Gamma_n$ has rank $\leq r$, then we have by Lemma~\ref{lem:loc_geom_finite_rank}, $\Gamma=G_{00}\oplus\cdots\oplus G_{NN}$ where each $G_{kk}$ has rank $\leq r$. A similar property holds for a localized state $\Gamma_B$. This information is unfortunately not enough to be really useful in applications. It is fortunate that this can be precised in the fermionic case, as expressed in the following important result.

\begin{lemma}[Localization of a fermionic $N$-body finite-rank state]\label{lem:localization_MCSCF}
Let $G\in\cS(\gH^N_a)$ be a \emph{fermionic} state of the $N$-body space $\gH^N_a$, of rank $\leq r$, and $\Gamma=0\oplus\cdots\oplus 0\oplus G$ be the corresponding state in $\cFN_a$. Let $B$ be a localization operator, $0\leq BB^*\leq1$, and denote by $\Gamma_B=G_{00}^B\oplus\cdots\oplus G_{NN}^B$ the corresponding localized state in $\cFN_a$. Then each $G_{kk}^B$ belongs to the convex hull of $k$-body states of rank at most $r-N+k$: we have
$$G_{kk}^B=\sum_{j}\alpha^k_j S^k_j$$
with $S^k_j\in\cS(\gH^k_a)$, $\alpha^k_j\geq0$ and
$$\rank(S^k_j)\leq r-N+k.$$
\end{lemma}

This result does not hold in general for bosons. In Example~\ref{ex:Hartree_loc} we have seen that the localization $\Gamma=G_{00}^B\oplus\cdots\oplus G_{NN}^B$ of a Hartree state $\Gamma=0\oplus\cdots\oplus|\phi^{\otimes N}\rangle\langle\phi^{\otimes N}|$ with rank $r=1$ satisfies $\rank(G^B_{kk})=1$ for all $k=1,...,N$. We now provide the proof of Lemma~\ref{lem:localization_MCSCF}.

\begin{proof}
Since $G$ has rank at most $r$, there exists a projector $P=\sum_{j=1}^r|\phi_j\rangle\langle\phi_j|$ of rank $r$ such that $\Gamma_{P}=\Gamma$.
By linearity we can assume that $G$ is a pure state, that is $G=|\Psi\rangle\langle\Psi|$ where
\begin{equation}
\Psi=\sum_{1\leq i_1,\cdots , i_N\leq r}c_{i_1\cdots i_N}\phi_{i_1}\otimes\cdots\otimes\phi_{i_N}.
\label{eq:expansion_MCSCF} 
\end{equation}
We follow here the notation of \cite{Lewin-04a}: $c_{i_1...i_N}$ reflects the symmetry of the wavefunction, that is $c_{i_{\sigma(1)}...i_{\sigma(N)}}=\epsilon(\sigma)c_{i_1...i_N}$ and $c_{i_1,...,i_N}=0$ as soon as two indices are equal.

We have the freedom to choose any orthonormal basis of the finite-dimensional space $V=\text{span}(\phi_j)=\text{Range}(P)$.
Indeed, if we replace the functions $\phi_j$ by $\phi_j'=\sum_{i=1}^rU_{ij}\phi_j$ for an $r\times r$ unitary matrix $U=(U_{ij})$, then \eqref{eq:expansion_MCSCF} is still valid, with adequately modified configuration coefficients $c'_{i_1,...,i_N}$ (see Formula (12) in \cite{Lewin-04a}). Taking advantage of this gauge freedom, we can diagonalize any matrix of the form $\pscal{\phi_i,M\phi_j}$ for a well-chosen one-body self-adjoint operator $M$. Here we choose $M=B^*B$, that is we work with an orthonormal system for which the $r\times r$ hermitian matrix $(\pscal{B\phi_i,B\phi_j})_{1\leq i,j\leq r}$ is diagonal. In particular we have $\pscal{B\phi_j,B\phi_j}=0$ if $i\neq j$, which dramatically simplifies the expression of $\Gamma_B$. Using Formula \eqref{eq:localization_N_body}, we find that
\begin{equation}
G_{kk}^B={N\choose k}\!\!\sum_{1\leq \ell_{k+1},\cdots, \ell_N\leq r}\!\!\!\left(1-\norm{B\phi_{\ell_{k+1}}}^2\right)\cdots\left(1-\norm{B\phi_{\ell_{N}}}^2\right)|\Psi_{\ell_{k+1}...\ell_N}\rangle\langle \Psi_{\ell_{r+1}...\ell_N}|
\end{equation}
where
\begin{align}
\Psi_{\ell_{k+1}...\ell_N}&=\sum_{i_{1},\cdots, i_k\in\{1,...,r\}}c_{i_1...i_k\ell_{k+1}...\ell_N}\;B\phi_{i_1}\otimes\cdots\otimes B\phi_{i_k}\nonumber\\
&=\sum_{i_{1},\cdots, i_k\in\{1,...,r\}\setminus\{\ell_{k+1},...,\ell_{N}\}}c_{i_1...i_k\ell_{k+1}...\ell_N}\;B\phi_{i_1}\otimes\cdots\otimes B\phi_{i_k}.\label{eq:c_i_vanish}
\end{align}
In \eqref{eq:c_i_vanish}, we have used that, for fermions, $c_{i_1,...,i_N}=0$ when two indices coincide. Clearly, $\Psi_{\ell_ {k+1}...\ell_N}$ has rank $\leq r-N+k$ and the result follows.
\end{proof}

\begin{example}[Localization of pure Hartree-Fock states]\label{ex:localization_HF}
Let $\Psi:=\phi_1\wedge\cdots\wedge\phi_N$ be a pure Hartree-Fock state and $\Gamma:=0\oplus\cdots\oplus|\Psi\rangle\langle\Psi|$ be the corresponding state in $\cFN_a$.
The localization of $\Gamma$ is $\Gamma_B=G_{00}^B\oplus\cdots\oplus G_{NN}^B$, with
\begin{multline}
G^B_{kk}=\sum_{I=\{i_1<\cdots< i_k\}\subset\{1,...,N\}}\\
\left(\prod_{\alpha\in\{1,...,N\}\setminus I}\left(1-\norm{B\phi_{\alpha}}^2\right)\right)\big|B\phi_{i_1}\wedge\cdots\wedge B\phi_{i_k}\big\rangle\big\langle B\phi_{i_1}\wedge\cdots\wedge B\phi_{i_k}\big|,
\end{multline}
assuming that the orbitals have been chosen such as to ensure $\pscal{B\phi_i,B\phi_j}=0$ when $i\neq j$.\dqed
\end{example}

From Lemma~\ref{lem:localization_MCSCF} we can deduce the general form of the geometric limit of fermionic $N$-body finite-rank states.

\begin{lemma}[Geometric limit of fermionic $N$-body finite-rank states]\label{lem:rank_MCSCF}
Let $\Gamma_n=0\oplus\cdots\oplus G^n\in\cFN$ be a sequence of fermionic $N$-body states, with $\rank(\Gamma_n)\leq r$ for all $n$. If $\Gamma_n\wto_g\Gamma=G_{00}\oplus\cdots\oplus G_{NN}$ geometrically, then each $G_{kk}$ belongs to the convex hull of $k$-body states of rank at most $r-N+k$.
\end{lemma}
\begin{proof}
Let $\{\phi_i\}$ be any fixed orthonormal basis of $\gH$ and $P_J:=\sum_{j=1}^J|\phi_j\rangle\langle\phi_j|$ be the projector onto the space spanned by the first $J$ elements of this basis. Since $P_J$ is compact for every fixed $J$, we have by Corollary~\ref{cor:compact_localization_operators}, $(\Gamma_n)_{P_J}\to (\Gamma)_{P_J}$ strongly. By Lemma~\ref{lem:localization_MCSCF}, we know that each $(\Gamma_n)_{P_J}$ can be written in the form $(\Gamma_n)_{P_J}=\oplus_{k=0}^NG^{J,n}_{kk}$, each $G_{kk}^{J,n}$ being a convex combinations of states of rank at most $r-N+k$. By strong convergence, we infer that $(\Gamma)_{P_J}$ has the same property. Now since $(\Gamma)_{P_J}\to\Gamma$ strongly as $J\to\ii$ by Lemma~\ref{lem:approx_localized}, we conclude that $\Gamma$ also satisfies the same property.
\end{proof}

\begin{example}[Geometric limit of pure Hartree-Fock states]\label{ex:geom_limit_HF}
Let $\Psi_n:=\phi_1^n\wedge\cdots\wedge\phi^n_N$ be a pure Hartree-Fock state and $\Gamma_n:=0\oplus\cdots\oplus|\Psi_n\rangle\langle\Psi_n|$ be the corresponding state in $\cFN_a$. We assume that $\phi^n_j\wto\phi_j$ weakly in $\gH$, for $j=1,...,N$. Up to applying an $n$-independent unitary transform $U$ to the $\phi_j^n$'s we may also suppose that $\pscal{\phi_i,\phi_j}=0$ when $i\neq j$. We then have that $\Gamma_n\wto_g G_{00}\oplus\cdots\oplus G_{NN}$ geometrically, with
\begin{equation}
G_{kk}=\sum_{I=\{i_1<\cdots< i_k\}\subset\{1,...,N\}}\left(\prod_{\alpha\in\{1,...,N\}\setminus I}\left(1-\norm{\phi_{\alpha}}^2\right)\right)\big|\phi_{i_1}\wedge\cdots\wedge\phi_{i_k}\big\rangle\big\langle\phi_{i_1}\wedge\cdots\wedge\phi_{i_k}\big|.
\end{equation}
We see that 
\begin{description}
\item[either] there is \textit{strong convergence}, $\phi_j^n\to\phi_j$ for all $j=1,...,N$, hence $G_{kk}=0$ for all $k=0,...,N-1$;
\item[or] \textit{all the particle are lost}, $\phi_j=0$ for all $j=1,...,N$, thus $G_{kk}=0$ for all $k=1,...,N$, that is $\Gamma=|\Omega\rangle\langle\Omega|$; 
\item[or] \textit{not all the particle are lost}, that is $0<\norm{\phi_j}<1$ for at least one $\phi_j$, and there exists $1\leq k\leq N -1$ such that $G_{kk}\neq0$.
\end{description}
Indeed if we assume
 (up to reordering) that $\phi_1,...,\phi_{N_1}\neq0$ but $\phi_{N_1+1}=\cdots=\phi_N=0$, we see that 
$\tr(G_{N_1\,N_1})\geq \prod_{j=1}^{N_1}\norm{\phi_j}^2>0$. The fact that we cannot have $G_{kk}=0$ for all $k=1,...,N-1$ while both $G_{00}$ and $G_{NN}$ are $\neq0$ will be very useful later in the proof of Theorem~\ref{thm:HF_translation_invariant}.\dqed
\end{example}

\subsection{HVZ-type results for finite-rank many-body systems}
\subsubsection{A general result}

Let us come back to the $N$-body Hamiltonian
$$H^V(N)=\sum_{j=1}^N\left(-\frac{\Delta_{x_j}}{2}+V(x_j)\right)+\sum_{1\leq k\leq \ell\leq N}W(x_k-x_\ell)$$
which we have already introduced in \eqref{def:Hamiltonian_W}. 
As usual we make the assumption that $W$ is even, and that $V$ and $W$ can both be written in the form $\sum_{i=1}^Kf_i$ with $f_i\in L^{p_i}(\R^d)$ where $\max(1,d/2)<p_i<\ii$ or $p_i=\ii$ but $f_i\to0$ at infinity.

For bosons or fermions we may introduce the approximated ground state energy obtained by restricting to finite-rank states:
\begin{equation}
E^V_r(N):=\inf_{\substack{\Psi\in H^1_{a/s}((\R^d)^N)\\ \rank(\Psi)\leq r\\ \norm{\Psi}=1}}\pscal{\Psi,\,H^V(N)\Psi}.
\label{def:E_V_N_R}
\end{equation}
We clearly have $E^V_r(N)\geq E^V(N)$ for all $r$, and $\lim_{r\to\ii}E^V_r(N)=E^V(N)$. 

Let us emphasize that, although the energy functional is the same as for the full linear model, we now have the additional constraint that $\rank(\Psi)\leq r$ which is itself highly nonlinear. Thus the so-obtained Euler-Lagrange equations are themselves nonlinear. If $r=1$, one gets for bosons the Hartree nonlinear equation. For fermions, one obtains the Hartree-Fock equations \cite{LieSim-77,Lions-87} for $r=N$ and the multiconfiguration equations \cite{Friesecke-03,Lewin-04a} for $r>N$.

We are interested here in existence results for ground states by means of geometric methods. The following theorem is a generalization to the nonlinear case of the HVZ Theorem~\ref{thm:HVZ}.

\begin{theorem}[Finite rank HVZ-type result, general case]\label{thm:HVZ_finite_rank_general}
If the following inequalities hold true
\begin{equation}
E_r^V(N)<E^V_r(N-k)+E^0_r(k),\qquad \forall k=1,...,N,
\label{eq:binding_finite_rank_general}
\end{equation}
then all the minimizing sequences $\{\Psi_n\}$ for the variational problem $E^V_r(N)$ are precompact, hence converge, up to a subsequence, to a ground state of rank $\leq r$.

If all the particles are fermions,  \eqref{eq:binding_finite_rank_general} can be replaced by
\begin{equation}
E_r^V(N)<E^V_{r-k}(N-k)+E^0_{r-N+k}(k),\qquad \forall k=1,...,N.
\label{eq:binding_finite_rank_fermions}
\end{equation}
\end{theorem}

\begin{proof}
Up to a subsequence we may assume that $\Psi_n\wto\Psi$ weakly in $H^1_{a/s}((\R^d)^N)$, and that the corresponding state $\Gamma_n:=0\oplus\cdots\oplus|\Psi_n\rangle\langle\Psi_n|\in\cS(\cFN)$ converges geometrically to $\Gamma=G_{00}\oplus\cdots\oplus G_{NN}$. If $\norm{\Psi}^2=\tr(G_{NN})=1$ then we have strong convergence $\Gamma_n\to\Gamma$ in $\gS_1(\cFN)$, hence $\Psi_n\to\Psi$ in $L^2$. Under our assumptions on $W$, this can then be used to prove that the two-body term converges strongly:
\begin{multline*}
\lim_{n\to\ii}\sum_{1\leq i<j\leq N}\int_{\R^d}dx_1\cdots\int_{\R^d}dx_N\, W(x_i-x_j)|\Psi_n(x_1,...,x_N)|^2\\
=\sum_{1\leq i<j\leq N}\int_{\R^d}dx_1\cdots\int_{\R^d}dx_N\, W(x_i-x_j)|\Psi(x_1,...,x_N)|^2. 
\end{multline*}
Since the interaction term is the only one which can fail from being weakly lower semi-continuous, we deduce that
$$E^V_r(N)=\lim_{n\to\ii}\cE^V(\Psi_n)\geq \cE^V(\Psi)\geq E^V_r(N),$$
hence that $\Psi$ is a ground state for $E^V(N)$. Finally, strong convergence in $H^1_{a/s}((\R^d)^N)$ is obtained by noting that $\lim_{n\to\ii}\cE^V(\Psi_n)= \cE^V(\Psi)$, hence that the kinetic energy must also converge.

Summarizing the previous paragraph, we only have to prove that $G_{kk}=0$ for all $k=0,...,N-1$. We follow the proof of Theorem~\ref{thm:HVZ}: we localize the system in and outside a ball of radius $R$, by means of a smooth partition of unity, $\chi_R^2+\eta_R^2=1$. In the lower bound corresponding to \eqref{eq:decomp_energy_HVZ}, we may use that each $G^n_{\chi_R,k}$ has rank $\leq r$ by Lemma~\ref{lem:loc_geom_finite_rank} (or rank $\leq r-N+k$ for fermions, by Lemma~\ref{lem:localization_MCSCF}). 
To be more precise, each $G^n_{\chi_R,k}$ can be diagonalized as follows 
$$G^n_{\chi_R,k}=\sum_{j}g_j^{R,k,n}|\Psi^{R,k,n}_j\rangle\langle\Psi^{R,k,n}_j|$$
where $g_j^{R,k,n}\geq0$ and $(P_n)^{\otimes k}\Psi^{R,k,n}_j=\Psi^{R,k,n}_j$ for an orthogonal projector $P_n$ of rank $\leq r$ (or $r+N-k$ for fermions). Saying differently each $G^n_{\chi_R,k}$ is a convex combination of pure states of rank $\leq r$. Hence we have an estimate of the form
$$\tr_{\gH^k}\left(H^V(k)G^n_{\chi_R,k}\right)\geq E^V_r(k)\;\tr_{\gH^k}\left(G^n_{\chi_R,k}\right),$$
with $E^V_r(k)$ replaced by $E^V_{r-N+k}(k)$ for fermions. A similar argument applies to the terms involving $G^n_{\eta_R,k}$.
Taking the limit $n\to\ii$ first and then removing the radius $R$ of the localization, following the proof of Theorem~\ref{thm:HVZ}, we arrive at the following estimate, similar to \eqref{eq:final_estimate_HVZ}:
$${E^V_r(N)\geq \sum_{k=0}^N\left(E^V_r(k)+E^0_r(N-k)\right)\tr_{\gH^{k}}(G_{kk})}$$
(with an obvious modification for fermions).
The term on the right is a convex combination of $E^V_r(N)$ (for $k=N$) and $E^V_r(k)+E^0_r(N-k)$ for $k=0,...,N-1$.
When \eqref{eq:binding_finite_rank_general} holds, this is only possible if $G_{kk}=0$ for all $k=0,...,N-1$.
\end{proof}

Unfortunately Theorem~\ref{thm:HVZ_finite_rank_general} only provides a sufficient condition for the compactness of minimizing sequences. In general we do not expect that \eqref{eq:binding_finite_rank_general} (or \eqref{eq:binding_finite_rank_fermions} for fermions) is also a necessary condition. The reason is that when two systems are placed far away in space, the rank of the whole system becomes the sum of the ranks of the two subsystems. This sum being $2r$ for \eqref{eq:binding_finite_rank_general} and $2r-N$ for \eqref{eq:binding_finite_rank_fermions}, the inequalities \eqref{eq:binding_finite_rank_general} and \eqref{eq:binding_finite_rank_fermions} are not expected to be correct in general when the strict inequality $<$ is replaced by a large inequality $\leq$. It is usually when large inequalities hold true that one can get necessary and sufficient conditions.

In the next section we will give two examples for fermions, due to Friesecke \cite{Friesecke-03}, for which one can reduce \eqref{eq:binding_finite_rank_fermions} to inequalities of the form 
\begin{equation}
 E^V_r(N)< E^V_{r-r'}(N-k)+E^0_{r'}(k),
\label{eq:desired_binding_rank}
\end{equation}
hence providing a necessary and sufficient condition of compactness of minimizing sequences. The case of geometric methods for finite-rank bosonic systems is still largely unexplored.

\subsubsection{Two corollaries for fermions}

We give two corollaries of Theorem~\ref{thm:HVZ_finite_rank_general} in the fermionic case. These two results are contained in a paper \cite{Friesecke-03} of Friesecke (see in particular Corollary 6.1 of \cite{Friesecke-03}), with a proof that is not very much different from our approach. Our formalism automatically takes care of the complicated geometrical methods for finite-rank states which was detailed in \cite{Friesecke-03} (in particular, the reader should compare Friesecke's Lemma 4.1 in \cite{Friesecke-03} with our Lemma~\ref{lem:localization_MCSCF}).

The first result deals with the Hartree-Fock case, corresponding to having rank $r=N$.

\begin{corollary}[Hartree-Fock HVZ-type]\label{cor:HF_V}
Assume that all the particles are fermions, and that $V$ and $W$ satisfy the same assumptions as before. Then the following assertions are equivalent:
\begin{enumerate}
 \item $E^V_N(N)<E^V_{N-k}(N-k)+E^0_k(k)$ for all $k=1,...,N$;

\medskip

\item all the minimizing sequences $\{\Psi_n\}$ for the Hartree-Fock ground state energy $E^V_N(N)$ are precompact in $H^1_a((\R^d)^N)$, hence converge, up to a subsequence, to a minimizer for $E^V_N(N)$.
\end{enumerate}
\end{corollary}
\begin{proof}
The implication \textit{(1)}$\Rightarrow$\textit{(2)} follows from Theorem~\ref{thm:HVZ_finite_rank_general} in the fermionic case, with $r=N$. To prove the converse inequality we first notice that it always holds $E^V_N(N)\leq E^V_{N-k}(N-k)+E^0_k(k)$ for all $k=1,...,N$. This is easily seen by taking a trial function of the form
\begin{equation}
\Psi_n=\Psi^1\wedge \Psi^2(\cdot-n\vec{v})
\label{eq:form_Psi_n_non_compact}
\end{equation}
where $\vec{v}\in\R^d\setminus\{0\}$, $\Psi^1=\phi_1\wedge\cdots\wedge\phi_{N-k}$ and $\Psi^2=\phi_{N-k+1}\wedge\cdots\wedge\phi_N$ are trial functions for, respectively, the problems $E^V_{N-k}(N-k)$ and $E^0_k(k)$. For simplicity one can take all the $\phi_j$'s of compact support. If there is equality $E^V_N(N)= E^V_{N-k}(N-k)+E^0_k(k)$ for some $k\in\{1,...,N\}$, then a minimizing sequence for $E^V_N(N)$ of the same form as \eqref{eq:form_Psi_n_non_compact} can be constructed and it is clearly not compact. This shows the converse implication \textit{(2)}$\Rightarrow$\textit{(1)}.
\end{proof}

There are now many different proofs for the existence of ground states in Hartree-Fock theory. For atoms and molecules, the first is due to Lieb and Simon \cite{LieSim-77}. An approach based on a second-order Palais-Smale information was proposed later by Lions \cite{Lions-87}. These two methods rely on a formulation of the problem in terms of the $N$ orbitals $\phi_1,...,\phi_N$ of the Hartree-Fock state as well as on the assumption that $W\geq0$. A different approach due to Lieb \cite{Lieb-81} (see also \cite{Bach-92,BacLieSol-94,BacLieLosSol-94}) uses generalized Hartree-Fock states and the fact that, when $W\geq0$, a generalized ground state is necessarily a pure state. In this formulation the minimization problem is expressed using as main variable the one-body density matrix $[\Gamma]^{(1)}$ which completely characterizes the Hartree-Fock state. When $W$ is not positive, it cannot be guaranteed that a generalized ground state is necessarily a pure state, and Lieb's variational principle of \cite{Lieb-81} cannot be employed.

Our approach here (due first to Friesecke \cite{Friesecke-03}) is completely different and it is based on geometric properties of $N$-body Hartree-Fock states. It leads to \emph{quantized inequalities} of the form of that of Corollary~\ref{cor:HF_V}, without any assumption on the sign of $W$.

Of course, the next step when studying a specific model is to prove that the binding inequality holds true. As explained by Friesecke in \cite{Friesecke-03}, this can be done by induction: using that there exist ground states for the problems with $k$ particles ($1\leq k<N$), one tries to prove by a convenient trial state that $E^V_N(N)<E^V_{N-k}(N-k)+E^0_k(k)$, showing the existence of a ground state for $E^V_N(N)$. For atoms and molecules, this argument can be carried over as soon as $N-1<Z$, where $Z$ is the total charge of the nuclei.

\medskip

Our second application of Theorem~\ref{thm:HVZ_finite_rank_general} in the fermionic case is the multiconfiguration case $N\leq r$ for \emph{repulsive} interactions. 

\begin{corollary}[Multiconfigurational HVZ-type in the repulsive case]
We assume that all the particles are fermions, that $V$ and $W$ satisfy the same assumptions as before and, additionally, that $W\geq0$. For every $r\geq N$, the following two assertions are equivalent:
\begin{enumerate}
\item $E^V_r(N)<E^V_{r-1}(N-1)$;

\medskip

\item all the minimizing sequences $\{\Psi_n\}$ for $E^V_r(N)$ are precompact in $H^1_a((\R^d)^N)$, hence converge, up to a subsequence, to a minimizer for $E^V_r(N)$.
\end{enumerate}
\end{corollary}

The reason why we restrict to $W\geq0$ is because it then holds
$$E^0_{r-N+k}(k)=E^0_{k}(k)=0.$$
Hence if we insert this in \eqref{eq:binding_finite_rank_fermions} we are left with an inequality of the form of \eqref{eq:desired_binding_rank}.
It is still an open question to understand the geometric behavior of multiconfiguration methods for non-repulsive interaction potentials (see, in particular, the comments on page 56 of \cite{Friesecke-03}).

\begin{proof}
The proof follows that of Corollary~\ref{cor:HF_V}, using that $E^0_{r-N+k}(k)=0$ since $W\geq0$, and that $\inf_{k=1,...,N}\{E^V_{r-k}(N-k)\}=E^V_{r-1}(N-1)$.
\end{proof}

Again for atoms and molecules, one can prove by induction the existence of a ground state as soon as $N<Z+1$, see \cite{Friesecke-03}.

\subsubsection{Translation-invariant Hartree-Fock theory}

In this subsection we study a translation-invariant Hartree-Fock model, that is we assume that $V=0$. It is known that (by translation-invariance) the $N$-body Hamiltonian $H^0(N)$ never has any ground state, but it can happen that there is one when restricting to Hartree-Fock states. Of course translation-invariance is not really broken: minimizers are not unique as they can be translated anywhere in space and it is the whole set of minimizers which is invariant under translations.

Because of the action of the group of translations it can only be hoped to prove compactness of all minimizing sequences \emph{up to translations}. 

\begin{theorem}[Translation-invariant Hartree-Fock]\label{thm:HF_translation_invariant}
We assume that $W$ satisfies the same assumptions as before (but $W$ need not be non-negative). Then for all $N\geq2$, the following assertions are equivalent:
\begin{enumerate}
\item $E^0_N(N)<E^0_{N-k}(N-k)+E^0_k(k)$ for all $k=1,...,N-1$;

\medskip

\item all the minimizing sequences $\{\Psi_n\}$ for $E^0_N(N)$ are precompact in $H^1_a((\R^d)^N)$ up to translations. Hence there exists $\{v_n\}\subset \R^d$ such that $\Psi_n(\cdot-v_n)$ converges, up to a subsequence, to a Hartree-Fock minimizer for $E^0_N(N)$.
\end{enumerate}
\end{theorem}

The notation $\Psi_n(\cdot-v_n)$ is interpreted in the sense of $(x_1,...,x_N)\mapsto \Psi_n(x_1-v_n,...,x_N-v_n)$. A result of the same kind was shown for the first time by Lenzmann and the author in \cite{LenLew-10}, for a model of neutron stars with a pseudo-relativistic kinetic energy and the gravitational Newton interaction. The pseudo-relativistic kinetic energy yields new difficulties concerning boundedness from below of the energy and localization errors (see Lemma A.1 in \cite{LenLew-10}). For nonrelativistic systems one easily arrives at the following result:

\begin{corollary}[Nonrelativistic Newtonian Hartree-Fock systems]
Assume that all the particles are fermions, that $d=3$ and $W(x-y)=-g/|x-y|$ with $g>0$. Then $E^0_N(N)$ has a Hartree-Fock ground state for all $N\geq2$ (hence infinitely many by translation-invariance).
\end{corollary}
\begin{proof}
The binding inequality $E^0_N(N)<E^0_{N-k}(N-k)+E^0_k(k)$ can be proved by induction using Newton's theorem, as explained in \cite{LenLew-10}.
\end{proof}

We are now ready to prove Theorem~\ref{thm:HF_translation_invariant}.

\begin{proof}
It was already shown in the proof of Corollary~\ref{cor:HF_V} that $E^0_N(N)\leq E^0_{N-k}(N-k)+E^0_k(k)$ for all $k=1,...,N-1$. Furthermore, if there is equality for some $k$, then one can construct a minimizing sequence which is not compact, even up to translations. Therefore we only have to prove that \textit{(1)}$\Rightarrow$\textit{(2)}.

To this end, we consider one minimizing sequence $\Psi_n=\phi_1^n\wedge\cdots\wedge\phi_N^n$ for $E^0_N(N)$ and we define the associated state in $\cFN$, $\Gamma_n=0\oplus\cdots\oplus|\Psi_n\rangle\langle\Psi_n|$. Since $\cE^0(\Gamma_n)$ is bounded, by \eqref{eq:bound_Hamiltonian} we have a uniform bound on the kinetic energy:
$$\tr_\gH\big((-\Delta)[\Gamma_n]^{(1)}\big)\leq C.$$
This itself implies a uniform bound on the $H^1(\R^d)$ norm of $\sqrt{\rho_{\Gamma_n}}$, by the Hoffmann-Ostenhof inequality \eqref{eq:Hoffmann-Ostenhof}.

Our goal is to prove convergence of $\Psi_n(\cdot-v_n)$ for an appropriate translation $v_n$. The first step is to determine this translation $v_n$ by detecting a piece of mass which retains its shape for $n$ large and, possibly, escapes to infinity. We therefore consider all the possible geometric limits, up to translations, of subsequences of $\{\Gamma_n\}$ and we define the largest possible average particle number that these limits can have:
\begin{equation}
m(\{\Gamma_n\}):=\sup\bigg\{\tr_\cF(\cN\Gamma)\ :\  \exists\{\vec{v}_{k}\}\subset\R^d,\ \tau_{\vec{v}_k}\Gamma_{n_k}\tau_{-\vec{v}_k}\wto_g\Gamma\bigg\}.
\label{eq:def_m_vanishing}
\end{equation}
Here $\tau_{\vec{v}}$ is the translation unitary operator defined by $(\tau_{\vec{v}}\Psi)(x_1,...,x_N)=\Psi(x_1-\vec{v},...,x_N-\vec{v})$ when $\Psi\in\gH^N$ and extended by linearity on the whole Fock space. By the strong convergence  $\rho_{\Gamma_n}\to\rho_\Gamma$ in $L^1_{\text{loc}}(\R^d)$ when $\Gamma_n\wto_g\Gamma$ (with bounded kinetic energy), we also have that
\begin{equation}
m(\{\Gamma_n\})=\sup\bigg\{\int_{\R^d}\rho\ :\  \exists\{\vec{v}_{k}\}\subset\R^d,\ \rho_{\Gamma_{n_k}}(\cdot-\vec{v}_k)^{1/2}\wto\rho^{1/2}\text{ weakly in $H^1(\R^d)$}\bigg\}.
\label{eq:def_m_vanishing_rho}
\end{equation}
The definition of $m(\{\Gamma_n\})$ is inspired of a result of Lieb \cite{Lieb-83} as well as of the concentration-compactness method of Lions \cite{Lions-84,Lions-84b}. 
The purpose of $m\big(\{\Gamma_n\}\big)$ is to detect the piece containing the largest average number of particles, which possibly escape to infinity (when $|{\vec{v}}_{k}|\to\ii$). 
Following Lions' terminology, a sequence $\{\Gamma_n\}$ is said to \emph{vanish} when $m(\{\Gamma_n\})=0$, which is equivalent to the property that
$$\forall\{\vec{v}_n\}\subset\R^d,\qquad \tau_{\vec{v}_n}\Gamma_n\tau_{-\vec{v}_n}\gto|\Omega\rangle\langle\Omega|$$
or that
$$\forall\{\vec{v}_n\}\subset\R^d,\qquad \rho_{\Gamma_n}(\cdot-\vec{v}_n)\to0 \text{ a.e.}$$
As we now explain, saying that $m\big(\{\Gamma_n\}\big)=0$ is actually quite a strong statement.

\begin{lemma}[Vanishing]\label{lem:vanishing}
Let $\{\Gamma_n\}$ be any sequence of states on $\cFN$, with a uniformly bounded kinetic energy.
The following assertions are equivalent:

\smallskip

\noindent$(i)$ $m\big(\{\Gamma_n\}\big)=0$;

\smallskip

\noindent$(ii)$ for all $R>0$, one has $\displaystyle\lim_{n\to\ii}\sup_{x\in\R^d}\int_{B(x,R)}\rho_{\Gamma_n}=0$;

\smallskip

\noindent$(iii)$ $\rho_{\Gamma_n}\to0$ strongly in $L^p(\R^d)$ for all $1<p<p^*$, where
$p^*=d/(d-2)$ if $d\geq3$, $p^*=\ii$ if $d=1,2$.
\end{lemma}

\begin{proof}
The fact that $(i)\Rightarrow(ii)$ follows from the strong local convergence of $\rho_{\Gamma_n}$. The implication $(ii)\Rightarrow(iii)$ was proved first by Lions in \cite{Lions-84b} (Lemma I.1). Finally, it is clear that if $\rho_{\Gamma_n}\to0$ strongly in one $L^p(\R^d)$, then $\rho_{\Gamma_n}(\cdot-x_n)\to0$ strongly in $L^p(\R^d)$ for every sequence $\{x_n\}\subset\R^d$, hence $(i)$ follows.
\end{proof}

We will now show using Lemma~\ref{lem:vanishing} that our Hartree-Fock minimizing sequence  cannot vanish. We have, using Wick's Theorem for generalized Hartree-Fock states \cite{LieSim-77,BacLieSol-94},
\begin{align*}
&\left|\pscal{\Psi_n,\left(\sum_{1\leq i<j\leq N}W(x_i-x_j)\right)\Psi_n}\right|\\
&\qquad\qquad\qquad\qquad\leq \pscal{\Psi_n,\left(\sum_{1\leq i<j\leq N}|W|(x_i-x_j)\right)\Psi_n}\\
&\qquad\qquad\qquad\qquad=\frac12\int_{\R^d}\int_{\R^d}|W(x-y)|\left(\rho_{\Gamma_n}(x)\rho_{\Gamma_n}(y)-|[\Gamma_n]^{(1)}(x,y)|^2\right)\,dx\,dy\\
&\qquad\qquad\qquad\qquad\leq\frac12\int_{\R^d}\rho_{\Gamma_n}\left(\rho_{\Gamma_n}\ast|W|\right).
\end{align*}
When $m(\{\Gamma_n\})=0$, we have that $\rho_{\Gamma_n}\to0$ in $L^p(\R^d)$ for all $1<p<p^*$ by Lemma~\ref{lem:vanishing}. Under our assumptions on $W$, this implies that the interaction term converges to 0. The kinetic energy being non-negative, this shows that in the case of vanishing
$$E^0_N(N)=\lim_{n\to\ii}\cE^0(\Gamma_n)\geq0,$$
which contradicts the assumption that \textit{(1)} holds true (it is clear that \textit{(1)} implies that $E^0_k(k)\leq k\,E^0_1(1)=0$ hence, since the inequality is strict in \textit{(1)}, that $E^0_N(N)<0$).

We have shown that $m\big(\{\Gamma_n\}\big)>0$. This proves that there exists a sequence $\{\vec{v}_k\}\subset\R^d$ and a subsequence $\Gamma_{n_k}$ such that $\Gamma_k':=\tau_{\vec{v}_k}\Gamma_{n_k}\tau_{-\vec{v}_k}\wto_g\Gamma$ with $\Gamma\neq|\Omega\rangle\langle\Omega|$. Since the problem $E^0_N(N)$ is invariant under translations, the new sequence $\Gamma'_k$ is also a minimizing sequence for $E^0_N(N)$. To simplify our exposition, we do not change our original notation and we assume that $\Gamma_n\wto_g\Gamma$ with $\Gamma=G_{00}\oplus\cdots\oplus G_{NN}$. The assumption that $\Gamma\neq|\Omega\rangle\langle\Omega|$ means that $0\leq G_{00}<1$. As usual strong convergence of $\{\Psi_n\}$ in $L^2$ implies strong convergence in $H^1$ and it suffices to prove that $G_{kk}=0$ for all $k=0,...,N-1$.

We can now follow the proof of Theorem~\ref{thm:HVZ_finite_rank_general} which uses a localization in a ball of radius $R$ as well as strong convergence in this ball, before passing to the limit $R\to\ii$. This yields an inequality of the form
$${E^0_N(N)\geq \sum_{k=0}^N\big(E^0_{N-k}(N-k)+E^0_k(k)\big)\,\tr_{\gH^{N-k}}(G_{kk}).}$$
Note that in comparison with Theorem~\ref{thm:HVZ_finite_rank_general}, the terms corresponding to $k=0$ and $k=N$ are equal.
When the binding inequality holds, this is only possible when $G_{kk}=0$ for all $k=1,...,N-1$. Hence we have
$$\Gamma=G_{00}\oplus0\oplus\cdots\oplus0\oplus G_{NN}.$$
We also know that $G_{00}\neq1$, hence $G_{NN}\neq0$. We have already explained in Example~\ref{ex:geom_limit_HF} that the only geometric limit of a sequence of pure Hartree-Fock states of this form must have $G_{00}=0$. This ends the proof of Theorem~\ref{thm:HF_translation_invariant}.\hfill \qed
\end{proof}

\section{Many-body systems with effective nonlinear interactions}\label{sec:nonlinear}

In this section we consider a system of $N$ quantum particles whose many-body energy is not linear with respect to the state $|\Psi\rangle\langle\Psi|$ of the system, but also contains a nonlinear term $F$:
$$\cE(\Psi)=\pscal{\Psi,H(N)\Psi}+F\big(|\Psi\rangle\langle\Psi|\big).$$
The purpose of the last term is often to effectively describe complicated interactions between our $N$ particles, through a second quantum system which has been eliminated from the model. Even when the model is translation-invariant, the $N$ particles can form bound systems thanks to the nonlinear term $F$.

Situations of this kind are ubiquitous in quantum physics. In Section \ref{sec:polaron}, we study the example of the $N$-polaron, which is a system of $N$ electrons in a polar crystal. In the so-called Pekar-Tomasevich model, the crystal is eliminated and replaced by an effective nonlinear Coulomb-like force between the electrons.

In nuclear physics, strong forces between nucleons are also often described by effective nonlinear terms. The most celebrated ones are the Skyrme \cite{VauBri-72} and the Gogny \cite{DecGog-80} forces. Although these methods have been mainly used in the context of mean-field theory, their extension to correlated models was recently considered in \cite{PilBerCau-08}.

In this section we illustrate our geometric techniques by studying the simple case of a concave nonlinear term $F$ depending only on the density $\rho_\Psi$ of the system. We state a general theorem in Section \ref{sec:nonlinear_general} and apply it to the multi-polaron in Section \ref{sec:polaron}.

\subsection{A general result}\label{sec:nonlinear_general}
Let us consider a system of $N$ spinless particles (bosons or fermions) in $\R^d$, interacting via a potential $W$ and a nonlinear effective term $F$. For simplicity we assume that $F$ only depends on the density of charge $\rho_\Psi$ of the many-body state $\Psi$:
\begin{equation}
\boxed{\cE(\Psi):=\pscal{\Psi,\left(\sum_{j=1}^N\frac{-\Delta_{x_j}}{2}+\sum_{1\leq k<\ell\leq N}W(x_k-x_\ell)\right)\Psi}+F(\rho_\Psi).}
\label{eq:def_nonlinear_many_body}
\end{equation}
We also introduce the corresponding ground state energy, for bosons or fermions,
\begin{equation}
E(N)=\inf_{\substack{\Psi\in H^1_{a/s}((\R^d)^N)\\ \norm{\Psi}=1}}\cE(\Psi).
\label{eq:def_nonlinear_GS}
\end{equation}

As before we make the assumption that $W$ can be written in the form $\sum_{i=1}^KW_i$ with $W_i\in L^{p_i}(\R^d)$ where $\max(1,d/2)<p_i<\ii$ or $p_i=\ii$ but $W_i\to0$ at infinity. As for the functional $F$, we assume that it satisfies the following assumptions:
\begin{description}
\item[\textbf{(A1)}] \textsl{(Subcriticality)} $F$ is a locally uniformly continuous functional on $L^{p_1}(\R^d)\cap L^{p_2}(\R^d)$, for some $1<p_1\leq p_2<p^*$, where $p^*=d/(d-2)$ when $d>2$ and $p^*=\ii$ when $d=1,2$, and such that $F(0)=0$. Furthermore, there exists $0<\epsilon<1$ and $C>0$ such that
\begin{equation}
\forall \phi\in H^1(\R^d), \quad \int_{\R^d}|\phi|^2\leq N\ \Longrightarrow\ F(|\phi|^2)\geq -\frac\epsilon2\int_{\R^d}|\nabla\phi|^2-C; 
\label{eq:lower_bound_F}
\end{equation}
\end{description}

\begin{description}
\item[\textbf{(A2)}] \textsl{(Translation invariance)} $F(\rho(\vec{v}+\cdot))=F(\rho)$ for all $\rho\in L^{p_1}(\R^d)\cap L^{p_2}(\R^d)$ and all $\vec{v}\in\R^d$;
\end{description}

\begin{description}
\item[\textbf{(A3)}] \textsl{(Decoupling at infinity)} If $\{\rho_n^1\}$ and $\{\rho_n^2\}$ are two bounded sequences of $L^1(\R^d)\cap L^{p_2}(\R^d)$ such that $\text{d}(\text{supp}(\rho_n^1)\,,\,\text{supp}(\rho_n^1))\to\ii$, then it holds 
$$F(\rho_n^1+\rho_n^2)-F(\rho_n^1)-F(\rho_n^2)\to0\text{ as $n\to\ii$};$$
\end{description}

\begin{description}
\item[\textbf{(A4)}] \textsl{(Concavity)} $F$ is concave on the cone $\{\rho\in L^{p_1}(\R^d)\cap L^{p_2}(\R^d)\ :\ \rho\geq0\}$;
\end{description}

\begin{description}
\item[\textbf{(A5)}] \textsl{(Strict concavity at the origin)} For all $\rho\in L^{p_1}(\R^d)\cap L^{p_2}(\R^d)$ with $\rho\geq0$ and $\rho\neq0$, one has $F(t\rho)>t\,F(\rho)$ for all $0<t<1$.
\end{description}

\begin{example}
Consider the following functional:
$$F(\rho)=-\alpha\int_{\R^d}\rho^\beta+\rho\left(\rho\ast h\right).$$
It can be verified that $F$ satisfies all the previous assumptions when $\alpha>0$, $1<\beta<1+2/d$, and when the function $h$ is of positive type ($\widehat{h}\geq0$) and can be written in the form $h=\sum_{i=1}^kh_i$ with $h_i\in L^{q_i}(\R^d)$ for some $\max\big(1,(d+1)/2\big)<q_i<\ii$. When $d=3$, this covers Coulomb interactions $h(x)=1/|x|$, as well as Dirac's term corresponding to $\beta=4/3$.\dqed
\end{example}

In the proof, the concavity of the functional $F$ is crucially used to extend the energy $\cE$ to mixed states in the truncated Fock space $\cFN$, making possible the use of geometric methods. Concavity might seem a very strong assumption but it is indeed very natural from a physical point of view. As we have explained the term $F(\rho_\Psi)$ usually empirically describes the interaction of our $N$ particles with a second (infinite) system (for instance phonons of a crystal for the multi-polaron studied in Section \ref{sec:polaron}). In most physical models the real coupling between the two systems is \emph{linear} with respect to the state of the $N$ particles (for instance linear with respect to $\rho_\Psi$). Eliminating the degrees of freedom of the second system by simple perturbation theory or minimization over product states always leads to concave functionals $F$.

The assumption \textbf{(A1)} that $F$ is subcritical will be used in the proof to discard the possibility that minimizing sequences vanish. 
The other assumptions on $F$ are of a more technical nature, and they can certainly be relaxed a bit. It is possible to treat non translation-invariant functionals but in this case the main result below is not stated the same. It is also easy to generalize the main theorem below to the case of a functional $F$ which is not a simple function of the density (for instance when $F$ is a function of the one-body density matrix), with appropriate assumptions.

\medskip

It is a simple exercise to verify that, under the previous assumptions, the energy functional $\cE$ is well-defined and continuous on $H^1_{a/s}((\R^d)^N)$. Moreover, using \eqref{eq:lower_bound_F} in \textbf{(A1)} and the Hoffmann-Ostenhof inequality \eqref{eq:Hoffmann-Ostenhof}, we have
\begin{align}
\cE(\Psi)&\geq \pscal{\Psi,\left((1-\epsilon)\sum_{j=1}^N\frac{-\Delta_{x_j}}{2}+\sum_{1\leq k<\ell\leq N}W(x_k-x_\ell)\right)\Psi}-C\nonumber\\
&\geq \frac{1-\epsilon}{2}\pscal{\Psi,\left(\sum_{j=1}^N\frac{-\Delta_{x_j}}{2}\right)\Psi}-C'.\label{eq:coercive_nonlinear}
\end{align}
In the second line we have used the assumptions on $W$, similarly as in \eqref{eq:bound_Hamiltonian}. This shows that $\cE$ is bounded from below, hence that $E(N)$ is finite.
In the following we denote by
$$H(N):=\sum_{j=1}^N\frac{-\Delta_{x_j}}{2}+\sum_{1\leq k<\ell\leq N}W(x_k-x_\ell)$$
the translation-invariant many-body Hamiltonian. The main theorem is the following:

\begin{theorem}[Nonlinear HVZ for many-body systems]\label{thm:nonlinear}
Under the previous assumptions, the following assertions are equivalent:
\begin{enumerate}
\item One has
\begin{equation}
E(N)<E(N-k)+E(k)\quad\text{for all $k=1,...,N-1$,} 
\label{eq:HVZ_nonlinear}
\end{equation}
and 
\begin{equation}
E(N)<\inf\sigma \left(H(N)\right);
\label{eq:cond_no_vanishing_nonlinear}
\end{equation}

\medskip

\item All the minimizing sequences $\{\Psi_n\}$ for $E(N)$ are precompact in $H^1_{a/s}((\R^d)^N)$ up to translations. Hence there exists $\{\vec{v}_n\}\subset \R^d$ such that $\Psi_n(\cdot-\vec{v}_n)$ converges, up to a subsequence, to a minimizer for $E(N)$.
\end{enumerate}
\end{theorem}

As we will explain in the proof, the role of the additional condition \eqref{eq:cond_no_vanishing_nonlinear} is to avoid vanishing.

\begin{proof}
We split the proof in several steps. We start by proving that the inequalities \eqref{eq:HVZ_nonlinear} and \eqref{eq:cond_no_vanishing_nonlinear} always hold true when the strict inequality $<$ is replaced by a large inequality $\leq$, and that if there is equality, then there exists a minimizing sequence which is non-compact, for any translations. This shows that \textit{(2)} implies \textit{(1)}.

\medskip

\noindent\textbf{Step 1: Large binding inequalities.} \textit{The inequalities in \eqref{eq:HVZ_nonlinear} always hold true when the strict inequality $<$ is replaced by $\leq$. If there is equality for some $1\leq k\leq N-1$, then there exists a minimizing sequence $\{\Psi_n\}$ for $E(N)$ which is not compact, even up to translations.}

\begin{proof}
The proof proceeds as usual by constructing a trial sequence $\Psi_n=\Psi_n^1\circ\Psi^2_n(\cdot-R_n\vec{v})$ (with $\circ=\wedge$ for fermions and $\circ=\vee$ for bosons), where $\Psi_n^1$ and $\Psi_n^2$ are minimizing sequences of compact support for $E(N-k)$ and $E(k)$ and $R_n$ is large enough. The energy is decoupled by \textbf{(A3)}. We omit the details.
\end{proof}

\smallskip

\noindent\textbf{Step 2: Large inequality \eqref{eq:cond_no_vanishing_nonlinear}.} \textit{The inequality \eqref{eq:cond_no_vanishing_nonlinear} always holds true when the strict inequality $<$ is replaced by $\leq$. If there is equality, then there exists a minimizing sequence $\{\Psi_n\}$ for $E(N)$ which is not compact, even up to translations.}

\begin{proof}
Removing the center of mass by performing the change of variables $x_0'=\sum_{j=1}^Nx_j/N$, $x'_1=x_2-x_1$, ...,  $x'_{N-1}=x_{N}-x_1$, we see that the original Hamiltonian $H(N)$ can be rewritten as
\begin{align*}
H(N)&=\frac{|p'_0|^2}{2N}+\left(\sum_{j=1}^{N-1}\frac{|p'_j|^2}2+\frac12\left|\sum_{j=1}^{N-1}p'_j\right|^2+\sum_{j=1}^{N-1}W(x'_j)+\sum_{1\leq k<\ell\leq N-1}W(x'_k-x'_\ell)\right)\\
&:=\frac{|p'_0|^2}{2N}+H'(N-1).
\end{align*}
This shows that the bottom of the spectrum of $H(N)$ is also the bottom of the spectrum of $H'(N-1)$. To account for the original statistics of our particles, the latter Hamiltonian $H'(N-1)$ is restricted to $(N-1)$--body functions $\Phi$ that are symmetric (bosons) or antisymmetric (fermions), and additionally satisfy the following relation
$$\Phi(-x'_1,x'_2-x'_1,\cdots,x'_{N-1}-x'_1)=\tau\;\Phi(x'_1,x'_2,\cdots,x'_{N-1})$$
with $\tau=1$ for bosons and $\tau=-1$ for fermions. Let $\{\Phi_n\}$ be a Weyl sequence for the bottom of the spectrum of the Hamiltonian $H'(N-1)$ (under the appropriate symmetry constraints) and let $\phi_n:=n^{-d/2}\phi(\cdot/n)$ for a fixed normalized function $\phi\in H^2(\R^d)\cap L^\ii(\R^d)$. We take as test function the product state
$$\Psi_n(x_1,\cdots ,x_N)=\phi_n\left(\frac{\sum_{j=1}^Nx_j}{N}\right)\;\Phi_n(x_2-x_1,\cdots,x_{N}-x_1)$$
whose density is
\begin{equation}
\rho_{\Psi_n}(x)=N\int_{\R^d}dx_2\cdots\int_{\R^d}dx_{N} \left|\phi_n\left(\frac{\sum_{j=2}^Nx_j}{N}+\frac{x}{N}\right)\right|^2\left|\Phi_n(x_2-x,\cdots,x_{N}-x)\right|^2.
\end{equation}
This proves that
\begin{equation*}
\norm{\rho_{\Psi_n}}_{L^\ii(\R^d)}\leq \frac{N\norm{\phi}_{L^\ii(\R^d)}^2}{n^d}\to_{n\to\ii}0,
\end{equation*}
hence that $\rho_{\Psi_n}\to0$ in $L^p(\R^d)$ for all $1<p\leq \ii$. Under our assumption \textbf{(A1)} on the nonlinearity $F$, this implies that $F(\rho_{\Psi_n})\to0$. On the other hand we have by construction 
$$\lim_{n\to\ii}\pscal{\Psi_n,H(N)\Psi_n}=\inf\sigma(H(N))$$
and it follows that $E(N)\leq \inf\sigma(H(N))$. 
If there is equality, the previous sequence $\{\Psi_n\}$ furnishes a vanishing minimizing sequence. It is not compact, even up to translations. This ends the proof of Step 2.
\end{proof}

\smallskip

The previous steps show that \textit{(2)} implies \textit{(1)}. We now turn to the proof of the converse implication. We consider a minimizing sequence $\{\Psi_n\}$ and note that it is necessarily bounded in $H^1_{a/s}((\R^d)^N)$, by \eqref{eq:coercive_nonlinear}. As usual we denote by $\Gamma_n=0\oplus\cdots\oplus|\Psi_n\rangle\langle\Psi_n|$ the associated mixed state in the truncated Fock space. We define like in the proof of Theorem \ref{thm:HF_translation_invariant} the number
\begin{equation}
m(\{\Gamma_n\}):=\sup\bigg\{\tr_\cF(\cN\Gamma)\ :\  \exists\{\vec{v}_{k}\}\subset\R^d,\ \tau_{\vec{v}_k}\Gamma_{n_k}\tau_{-\vec{v}_k}\wto_g\Gamma\bigg\}.
\label{eq:def_m_vanishing_nonlinear}
\end{equation}
We start by proving that vanishing does not hold, that is $m(\{\Gamma_n\})>0$.

\bigskip

\noindent\textbf{Step 3: Absence of vanishing.} \textit{One has $m(\{\Gamma_n\})>0$.}

\begin{proof}
As we have already seen in Lemma \ref{lem:vanishing}, $m(\{\Gamma_n\})=0$ is equivalent to having $\rho_{\Psi_n}\to0$ strongly in $L^p((\R^d)^N)$, for all $1<p<p^*$. By Assumption \textbf{(A1)}, the function $F$ is uniformly continuous on $L^{p_1}(\R^d)\cap L^{p_2}(\R^d)$ for some $1<p_1\leq p_2<p^*$. Hence $m(\{\Gamma_n\})=0$ implies that $F(\rho_{\Psi_n})\to0$ and therefore that 
$$E(N)=\lim_{n\to\ii}\cE(\Psi)=\lim_{n\to\ii}\pscal{\Psi,H(N)\Psi}\geq \inf\sigma(H(N)).$$
This contradicts \textit{(1)}, hence shows that it must hold  $m(\{\Gamma_n\})>0$.
\end{proof}

\smallskip

Up to a translation (we use that $\cE$ is translation-invariant) and extraction of a subsequence, we may therefore assume that $\Gamma_n\wto_g\Gamma$ geometrically, with $\tr(\cN\Gamma)>0$, that is $\Gamma=G_{00}\oplus\cdots\oplus G_{NN}$ with $0\leq G_{00}<1$. In order to show that $\{\Gamma_n\}$ is compact, we have to prove that $\tr(G_{NN})=1$. This only shows that $\Psi_n\to\Psi$ strongly in $L^2_{a/s}((\R^d)^N)$ but strong convergence in $H^1_{a/s}((\R^d)^N)$  follows by usual arguments.

\bigskip

\noindent\textbf{Step 4: Decoupling via localization.} In this step we split $\Gamma_n$ into a part which converges to $\Gamma$ strongly and a part which escapes to infinity. Contrary to the previous sections, we use a radius of localization which depends on $n$, following Lions \cite{Lions-84,Lions-84b}. The following is a well-known result:
\begin{lemma}[Dichotomy]\label{lem:dichotomy}
Up to extraction of a subsequence, it holds
$$\lim_{n\to\ii}\int_{|x|\leq R_n} \rho_{\Psi_n}(x)\,dx=\int_{\R^d} \rho_{\Gamma}(x)\,dx,$$
\begin{multline*}
\lim_{n\to\ii}\int_{R_n\leq |x|\leq 6R_n}\left(\rho_{\Psi_n}(x)+|\nabla\sqrt{\rho_{\Psi_n}(x)}|^2\right)\,dx\\=\lim_{n\to\ii}\int_{R_n\leq |x_1|\leq 6R_n}dx_1\,\int_{\R^d}dx_2\cdots\int_{\R^d}dx_N\;|\nabla_{x_1}\Psi_n(x_1,...,x_N)|^2=0 
\end{multline*}
for a sequence $R_n\to\ii$.
\end{lemma}
The proof of this lemma uses concentration functions in the spirit of Lions \cite{Lions-84,Lions-84b} as well as the strong local compactness of $\rho_{\Psi_n}$. See for instance Lemma 3.1 in \cite{Friesecke-03} for a similar result. Let $\chi$ be a smooth radial localization function with $0\leq\chi\leq1$, $\chi(x)=1$ if $|x|\leq 1$ and $\chi(x)=0$ if $|x|\geq2$, and let $\eta:=\sqrt{1-\chi^2}$. Let us consider the smooth localization functions $\chi_n:=\chi(\cdot/R_n)$ and $\eta_n=\eta(\cdot/R_n)$, in and outside the ball of radius $R_n$. By Lemma \ref{lem:continuity}, we have $(\Gamma_n)_{\chi_n}\wto_g \Gamma$ geometrically. However by Lemma \ref{lem:dichotomy} it holds 
$$\lim_{n\to\ii}\tr\;[(\Gamma_n)_{\chi_n}]^{(1)}=\lim_{n\to\ii}\int_{\R^d}(\chi_n)^2\rho_{\Gamma_n}=\int_{\R^d}\rho_{\Gamma}=\tr\;[\Gamma]^{(1)}.$$
This shows that $[(\Gamma_n)_{\chi_n}]^{(1)}\to [\Gamma]^{(1)}$ strongly in the trace class, hence by Lemma \ref{lem:N_wlsc} that
$$(\Gamma_n)_{\chi_n}\to \Gamma\ \text{ strongly in $\cS(\cFN)$ as $n\to\ii$}.$$

We can now show that the energy decouples. For the linear part we have by the IMS formula (like in the proof of Theorem \ref{thm:HVZ})
\begin{multline}
\pscal{\Psi_n,H(N)\Psi_n}\geq \tr_{\cFN}\left(\bH\,(\Gamma_n)_{\chi_n}\right)+ \tr_{\cFN}\left(\bH\,(\Gamma_n)_{\eta_n}\right)-\frac{CN}{R_n^2}\\
+N(N-1)\int_{\R^d}dx_1\cdots\int_{\R^d}dx_N\;W(x_1-x_2)\chi_n(x_1)^2\eta_n(x_2)^2|\Psi_n(x_1,...,x_N)|^2,\label{eq:decomp_energy_nonlinear}
\end{multline}
where $\bH=0\oplus \bigoplus_{n=1}^NH(n)$ is the second quantization of $H(N)$ in $\cFN$. Performing a decomposition similar to \eqref{eq:loc_inter_3} and using Lemma \ref{lem:dichotomy}, one sees that the last term of \eqref{eq:decomp_energy_nonlinear} goes to zero as $n\to\ii$. For the nonlinear term, we write
$$\rho_{\Psi_n}=|\chi_n|^2\rho_{\Psi_n}+|\eta_n|^2\rho_{\Psi_n}=|\chi_n|^2\rho_{\Psi_n}+|\eta_n|^2|\chi_{3R_n}|^2\rho_{\Psi_n}+|\eta_{3R_n}|^2\rho_{\Psi_n}.$$
By Lemma \ref{lem:dichotomy} we have that $|\eta_n|^2|\chi_{3R_n}|^2\rho_{\Psi_n}\to0$ in $L^1(\R^d)\cap L^{p^*}(\R^d)$, hence in $L^{p_1}(\R^d)\cap L^{p_2}(\R^d)$. Using that $F$ is locally uniformly continuous on $L^{p_1}(\R^d)\cap L^{p_2}(\R^d)$, we deduce since $\rho_{\Psi_n}$ is bounded in $L^{p_1}(\R^d)\cap L^{p_2}(\R^d)$, that
$$F(\rho_{\Psi_n})=F\left(|\chi_n|^2\rho_{\Psi_n}+|\eta_{3R_n}|^2\rho_{\Psi_n}\right)+o(1).$$
By Assumption \textbf{(A3)} we have
$$F\left(|\chi_n|^2\rho_{\Psi_n}+|\eta_{3R_n}|^2\rho_{\Psi_n}\right)=F\left(|\chi_n|^2\rho_{\Psi_n}\right)+F\left(|\eta_{3R_n}|^2\rho_{\Psi_n}\right)+o(1).$$
Using again that $|\eta_n|^2|\chi_{3R_n}|^2\rho_{\Psi_n}\to0$ we finally deduce that 
$$F(\rho_{\Psi_n})=F\left(|\chi_n|^2\rho_{\Psi_n}\right)+F\left(|\eta_n|^2\rho_{\Psi_n}\right)+o(1).$$
Hence we arrive at the following estimate
\begin{equation}
\pscal{\Psi_n,H(N)\Psi_n}\geq \tr_{\cFN}\!\left(\bH\,(\Gamma_n)_{\chi_n}\right)+F\left(\rho_{(\Gamma_n)_{\chi_n}}\right)+ \tr_{\cFN}\!\left(\bH\,(\Gamma_n)_{\eta_n}\right)+F\left(\rho_{(\Gamma_n)_{\eta_n}}\right)+o(1).
\label{eq:lower_bound_nonlinear}
\end{equation}

Let us write the localized states on $\cFN$ as
$$(\Gamma_n)_{\chi_n}=G_0^{\chi,n}\oplus\cdots\oplus G_N^{\chi,n},\qquad (\Gamma_n)_{\eta_n}=G_0^{\eta,n}\oplus\cdots\oplus G_N^{\eta,n}.$$
By the concavity of $F$, we have 
\begin{equation}
F\left(\rho_{(\Gamma_n)_{\eta_n}}\right)\geq \sum_{j=0}^N\tr(G_j^{\eta,n})\; F\left(\rho_{\tilde{G}^{\eta,n}_j}\right),
\end{equation}
with $\tilde{G}^{\eta,n}_j:=G^{\eta,n}_j/\tr(G^{\eta,n}_j)$ (and an obvious convention when $G_j^{\eta,n}=0$).
Using the fundamental relation $\tr(G_j^{\chi,n})=\tr(G_{N-j}^{\eta,n})$, we arrive at the lower bound
\begin{equation}
\tr_{\cFN}\left(\bH\,(\Gamma_n)_{\eta_n}\right)+F\left(\rho_{(\Gamma_n)_{\eta_n}}\right)
\geq \sum_{j=0}^N\tr(G_j^{\chi,n})\;\cE(\tilde{G}^{\eta,n}_{N-j})
\geq \sum_{j=0}^N\tr(G_j^{\chi,n})\;E(N-j).
\label{eq:lower_bound_eta}
\end{equation}
In the previous bounds, the energy $\cE$ is extended to mixed states of $\gH^N$ in an obvious fashion. Furthermore, for any mixed state $G\in\cS(\gH^N)$, we have, writing $G=\sum_{j}g_j|\Psi_j\rangle\langle\Psi_j|$ with $\sum_jg_j=1$,
\begin{equation*}
\cE(G)=\sum_{j}g_j\pscal{\Psi_j,H(N)\Psi_j}+F\left(\sum_{j}g_j\rho_{\Psi_j}\right)\geq \sum_j\,g_j\,\cE(\Psi_j)\geq E(N),
\end{equation*}
by the concavity of $F$. Therefore minimizing over mixed states is the same as minimizing over pure states, a property that we have used in \eqref{eq:lower_bound_eta}.

Coming back to the term involving $\chi_n$ in \eqref{eq:lower_bound_nonlinear}, we claim that it holds
\begin{equation*}
\liminf_{n\to\ii}\bigg(\tr_{\cFN}\left(\bH\,(\Gamma_n)_{\chi_n}\right)+F\left(\rho_{(\Gamma_n)_{\chi_n}}\right)\bigg)\geq \tr_{\cFN}\left(\bH\,\Gamma\right)+F\left(\rho_{\Gamma}\right).
\end{equation*}
Indeed the interaction term and $F(\rho_{\Gamma_n})$ converge as $n\to\ii$, by the strong convergence of $(\Gamma_n)_{\chi_n}$ towards $\Gamma$ in $\gS_1(\cFN)$. The kinetic energy is lower semi-continuous, by Lemma \ref{lem:wlsc}.

Summarizing, we have obtained the following lower bound
\begin{equation}
E(N)\geq \tr_{\cFN}\left(\bH\,\Gamma\right)+F\left(\rho_{\Gamma}\right)+\sum_{j=0}^N\tr(G_{jj})\;E(N-j).
\label{eq:final_lower_bound_nonlinear} 
\end{equation}
Using the concavity of $F$ as for $(\Gamma_n)_{\eta_n}$, we have
$$\tr_{\cFN}\left(\bH\,\Gamma\right)+F\left(\rho_{\Gamma}\right)\geq \sum_{j=0}^N\tr(G_{jj})\;E(j),$$
hence it follows that
$$E(N)\geq \sum_{j=0}^N\tr(G_{jj})\;\big(E(j)+E(N-j)\big).$$
When the binding condition \eqref{eq:HVZ_nonlinear} holds true, this is only possible when $G_{11}=\cdots=G_{N-1\,N-1}=0$. 

\bigskip

\noindent\textbf{Step 5: Conclusion.} It rests to prove that $G_{00}=0$. Let $\Psi$ be the weak limit in $\gH^N$ of the original minimizing sequence $\{\Psi_n\}$ and notice that $G_{NN}=|\Psi\rangle\langle\Psi|$. Since $G_{NN}\neq0$, it holds $\Psi\neq0$.  Inserting all this in \eqref{eq:final_lower_bound_nonlinear} (recall $\rho_{G_{00}}=0$), we obtain the estimate
\begin{equation}
\big(1-\tr(G_{00})\big)E(N)=\norm{\Psi}^2E(N)\geq \pscal{\Psi,H(N)\Psi}+F(\rho_\Psi).
\label{eq:very_final_lower_bound_nonlinear} 
\end{equation}
If $\norm{\Psi}<1$, then we use \textbf{(A5)} and get
$$F(\rho_\Psi)>\norm{\Psi}^2F(\rho_{\Psi/\norm{\Psi}}),$$
that is
$$\pscal{\Psi,H(N)\Psi}+F(\rho_\Psi)>\norm{\Psi}^2\cE\left(\frac{\Psi}{\norm{\Psi}}\right)\geq \norm{\Psi}^2E(N).$$
This contradicts \eqref{eq:very_final_lower_bound_nonlinear}, hence implies that it must hold $\norm{\Psi}=1$ and $G_{00}=0$. This ends the proof of Theorem \ref{thm:nonlinear}.\hfill\qed
\end{proof}

\bigskip

Theorem \ref{thm:nonlinear} can be generalized to finite-rank fermionic systems (Hartree-Fock case or multiconfiguration theory when $W\geq0$), following the arguments of Section \ref{sec:HF_MCSCF}. For instance, in the Hartree-Fock case one can easily prove the following

\begin{theorem}[Nonlinear HVZ for many-body systems in the Hartree-Fock approximation]\label{thm:nonlinear_HF}
Let $E_N(N)$ be the (fermionic) ground state energy in the Hartree-Fock approximation, defined by
\begin{equation}
E_N(N):=\inf_{\substack{\Psi\in H^1_{a}((\R^d)^N)\\ \rank(\Psi)=N\\ \norm{\Psi}=1}}\cE(\Psi).
\label{eq:def_HF_nonlinear_case}
\end{equation}
Under the previous assumptions, the following assertions are equivalent:
\begin{enumerate}
\item One has
\begin{equation}
E_N(N)<E_{N-k}(N-k)+E_k(k)\quad\text{for all $k=1,...,N-1$}; 
\label{eq:HVZ_nonlinear_HF}
\end{equation}

\medskip

\item All the Hartree-Fock minimizing sequences $\{\Psi_n\}$ for $E_N(N)$ are precompact in $H^1_{a}((\R^d)^N)$ up to translations. Hence there exists $\{\vec{v}_n\}\subset \R^d$ such that $\Psi_n(\cdot-\vec{v}_n)$ converges, up to a subsequence, to a minimizer for $E_N(N)$.
\end{enumerate}
\end{theorem}

Note the absence of a condition of the form \eqref{eq:cond_no_vanishing_nonlinear}: as we have seen in the proof of Theorem \ref{thm:HF_translation_invariant}, in the case of vanishing of a Hartree-Fock state, the interaction energy always tends to zero. The condition \eqref{eq:HVZ_nonlinear_HF} is sufficient to avoid this.

\subsection{Application: the multi-polaron}\label{sec:polaron}
In this section we study a system of $N$ electrons in a polar (ionic) crystal, called \emph{$N$-polaron}. Thanks to the underlying deformations of the crystal, the $N$ electrons can overcome their Coulomb repulsion and form a bound system. Recently there has been a renewed interest in the multi-polaron problem, triggered by the possibility of bipolaronic superconductivity in high-temperature superconductors \cite{Emin-89}.

Under the assumption that the polaron extends over a region much bigger than the typical spacing between the ions of the crystal, one can use a continuous model based on phonons. A model of this form was proposed by H. Fröhlich in \cite{Frohlich-54}. It assumes a linear coupling between the electrons and the longitudinal optical phonons, together with a constant dispersion relation for the phonons. The corresponding Hamiltonian takes the form
\begin{equation}
\sum_{j=1}^N\left(\frac{-\Delta_{x_j}}{2}-\sqrt{\alpha}\phi(x_j)\right)+\sum_{1\leq k<\ell\leq N}\frac{U}{|x_k-x_\ell|}+\int_{\R^3}\,dk\, a^\dagger(k)\,a(k), 
\label{eq:Hamil_polaron}
\end{equation}
where
$$\phi(x)=\frac{1}{2\pi}\int_{\R^3}\frac{dk}{|k|}\left(e^{ik\cdot x}\, a^\dagger(k)+e^{-ik\cdot x}\, a(k)\right).$$
The Hamiltonian acts on the Hilbert space $L^2_a((\R^3)^N)\otimes\cF_s$, with $a^\dagger(k)$ and $a(k)$ being the creation and annihilation operators (in the Fourier representation) for the phonons on the bosonic Fock space $\cF_s$. Because of its relation to the dielectric constants of the polar crystal \cite{Frohlich-37,VerPeeDev-91}, the parameter $\alpha$ must satisfy the constraint $\alpha<U$ in the physical regime. For simplicity we have discarded the spin of the electrons.

In the regime of strong coupling, the model reduces to the so-called \emph{Pekar-Tomasevich} (PK) theory \cite{Pekar-63,PekTom-51,MiySpo-07} in which the interaction with the crystal is modelled by a classical Coulomb self-interaction. The energy is now given by
\begin{equation}
\cE_{\alpha,U}(\Psi)=\pscal{\Psi,\left(\sum_{j=1}^N\frac{-\Delta_{x_j}}{2}+\sum_{1\leq k<\ell\leq N}\frac{U}{|x_k-x_\ell|}\right)\Psi}-\frac\alpha2\int_{\R^3}\int_{\R^3}\frac{\rho_\Psi(x)\,\rho_\Psi(y)}{|x-y|}dx\,dy,
\label{eq:Pekar-Tom_energy} 
\end{equation}
for $\Psi\in L^2_a((\R^3)^N)$.
The corresponding ground state energy is as usual defined as
\begin{equation}
E_{\alpha,U}(N)=\inf_{\substack{\Psi\in H^1_{a}((\R^3)^N)\\ \norm{\Psi}=1}}\cE_{\alpha,U}(\Psi).
\end{equation}
We have emphasized the dependence in the parameters $\alpha$ and $U$.
A simple scaling argument shows that $E_{\alpha,U}=U^2 E_{\alpha/U,1}$, hence we may work in a system of units such that $U=1$. In this case, for simplicity we use the notation $\cE_{\alpha}:=\cE_{\alpha,1}$ and $E_{\alpha}(N):=E_{\alpha,1}(N)$.

Another way to derive the Pekar-Tomasevich energy is to restrict to (uncorrelated) products states of the form $\Psi\otimes\Phi\in L^2_a((\R^3)^N)\otimes\cF_s$ and to minimize with respect to the state $\Phi$ of the phonons \cite{GriMol-10}.

Both the original model of Fröhlich and the Pekar-Tomasevich theory have stimulated many works. On the mathematical side, the validity of PK theory in the large coupling regime was shown for $N=1$ by Donsker and Varadhan in \cite{DonVar-83}, and with a different approach by Lieb and Thomas in \cite{LieTho-97}. The case $N=2$ was treated by Miyao and Spohn in \cite{MiySpo-07}. The stability or instability of large polaron systems was studied by Griesemer and M{\o}ller \cite{GriMol-10}, then by Frank,  Lieb, Seiringer and Thomas \cite{FraLieSeiTho-10a,FraLieSeiTho-10b}. In this latter work, the absence of binding of $N$-polaron for small $\alpha$ is also proven. 

Using geometric techniques, we are able to study the existence of multi-polaron systems:

\begin{theorem}[Binding of Pekar-Tomasevich multi-polarons]\label{thm:polaron}
Assume $U=1$. For every $N\geq2$, there exists a constant $\tau_c(N)<1$ such that the following hold for all $\alpha>\tau_c(N)$:
\begin{enumerate}
\item $E_\alpha(N)<E_\alpha(N-k)+E_\alpha(k)$ for all $k=1,...,N-1$;

\smallskip

\item All the minimizing sequences $\{\Psi_n\}$ for $E_\alpha(N)$ are precompact in $H^1_{a/s}((\R^d)^N)$ up to translations. Hence there exists $\{\vec{v}_n\}\subset \R^d$ such that $\Psi_n(\cdot-\vec{v}_n)$ converges, up to a subsequence, to a minimizer $\Psi$ for $E_\alpha(N)$;

\smallskip

\item Any such minimizer satisfies the following nonlinear eigenvalue equation:
\begin{equation}
\left(\sum_{j=1}^N\left(\frac{-\Delta}{2}-\alpha\rho_\Psi\ast|\cdot|^{-1}\right)_{x_j}+\sum_{1\leq k<\ell\leq N}\frac{1}{|x_k-x_\ell|}\right)\Psi=\mu\Psi
\label{eq:SCF_polaron}
\end{equation}
where $\mu$ is the first eigenvalue of the many-body Schrödinger operator in the parenthesis.
\end{enumerate}
\end{theorem}

\smallskip

Our result covers the physical range $\alpha\in(\tau_c(N),1)$ but we do not provide any bound on the critical $\tau_c(N)$. It was proved in \cite{FraLieSeiTho-10a} that binding does \emph{not} occur when $\alpha$ is small enough, hence one must have $\tau_c(N)>0$. We expect that $\tau_c(N)\to1$ when $N\to\ii$ but we do not have a proof of this.

For $N=1$, the Pekar-Tomasevich energy is defined as 
\begin{equation}
\cE_\alpha(\phi)=\frac12\int_{\R^3}|\nabla\phi|^2-\frac\alpha2\int_{\R^3}\int_{\R^3}\frac{|\phi(x)|^2\,|\phi(y)|^2}{|x-y|}dx\,dy
\label{eq:Pekar-Tom_un} 
\end{equation}
and it is sometimes also called the Choquard functional. The existence and uniqueness of a ground state up to translations for all $\alpha>0$ was proved by Lieb in \cite{Lieb-77}. Nothing seems to be known on the uniqueness of ground states up to translations for $N\geq2$.

For the bipolaron ($N=2$), the binding energy 
$$2E_\alpha(1)-E_\alpha(2)=2E_1(1)\,\alpha^2-E_\alpha(2)$$
is a convex and non-decreasing function of $\alpha$. We deduce from Theorem \ref{thm:polaron} that there exists $\tau_c(2)<1$ such that binding does not hold for all $0\leq\alpha\leq\tau_c(2)$, whereas binding holds true and minimizers exist for all $\alpha>\tau_c(2)$. A result of the same form was already announced in \cite{MiySpo-07}. Numerical computations \cite{VerSmoPeeDev-92,SmoFom-94} suggest that, for the bipolaron, $\tau_c(2)\simeq 0.87$.

Since the Pekar-Tomasevich model is exact in the limit of strong coupling, $\alpha/U<1$ and $\alpha\gg1$, our result implies the existence of binding for Fröhlich's $N$-polaron described by the Hamiltonian \eqref{eq:Hamil_polaron}, when $\tau_c(N)< \alpha/U<1$ and $\alpha$ is large enough. For small $\alpha$, numerical computations indeed suggest that Fröhlich's polaron does not bind for any $U>\alpha$. In \cite{VerSmoPeeDev-92} (Fig. 4) the critical value above which Fröhlich's bipolaron formation is possible was found to be $\alpha\simeq13.15$.

\begin{remark}[Extensions]
For anisotropic materials, one can take $F$ of the form
$$F(\rho)=-\frac{4\pi}2 \int_{\R^3}\frac{|\widehat{\rho}(k)|^2}{k^TMk}\,dk$$
where $M$ is a $3\times3$ symmetric matrix satisfying $M\geq1$. Existence of ground states follows from our method when $M$ is sufficiently close to the identity matrix.

Our results hold the same in 2D, assuming the particles interact with the 3D Coulomb potential, a model which is often considered in the physical literature (see, e.g. \cite{VerSmoPeeDev-92,VerPeeDev-91}).
\dqed
\end{remark}

Thanks to Theorem \ref{thm:nonlinear}, the proof of Theorem \ref{thm:polaron} is essentially reduced to showing the binding condition. This is done by building suitable trial states. The easy case is $\alpha>1$, when two multi-polaron always have a Coulomb attraction at large distances. The case $\alpha=1$ is more subtle, and we prove that there is always a Van Der Waals attraction at large distances, following Lieb and Thirring \cite{LieThi-86}. The existence of $\tau_c(N)$ is then obtained by continuity of $\alpha\mapsto E_\alpha(N)$, using that there are only finitely many binding conditions to verify.

\begin{proof}
The energy $\cE_\alpha$ is of the general form which we have considered in Section \ref{sec:nonlinear_general}. The nonlinear functional
$$F(\rho)=-\frac\alpha2\int_{\R^3}\int_{\R^3}\frac{\rho(x)\,\rho(y)}{|x-y|}dx\,dy$$
is clearly strictly concave, and it satisfies our assumptions \textbf{(A1)}--\textbf{(A5)} with $p_1=p_2=6/5$, by the Hardy-Littlewood-Sobolev inequality \cite{LieLos-01}. Furthermore, the condition \eqref{eq:cond_no_vanishing_nonlinear} reduces to $E_\alpha(N)<0$ since the interaction potential $W(x)=1/|x|$ is non-negative. This condition is implied by the binding condition, hence it is only necessary to verify that $E_\alpha(N)<E_\alpha(N-k)+E_\alpha(k)$ for $k=1,...,N-1$.
Since the function $\alpha\mapsto E_\alpha(N)$ is clearly continuous, it is sufficient to show that 
\begin{equation}
E_\alpha(N)<E_\alpha(N-k)+E_\alpha(k)\ \text{ for all integers $1\leq k\leq N-1$ and all $\alpha\geq1$.}
\end{equation}
As usual, we prove these binding inequalities by induction, assuming that $E_\alpha(k)$ has a minimizer for all $k=1,...,N-1$. 
For $N=1$, we already know that ground states of $E_\alpha(1)$ exist for all $\alpha>0$.
The following will be very useful.

\begin{lemma}[Properties of multi-polaron ground states]
Assume that $\Psi$ is a ground state for $E_\alpha(N)$ with $\alpha>0$. Then $\Psi$ solves the self-consistent equation \eqref{eq:SCF_polaron} where $\mu$ is the first eigenvalue of the many-body operator
$$H_{\Psi}^\alpha(N):=\sum_{j=1}^N\left(\frac{-\Delta}{2}-\alpha\rho_\Psi\ast|\cdot|^{-1}\right)_{x_j}+\sum_{1\leq k<\ell\leq N}\frac{1}{|x_k-x_\ell|}.$$
If $\alpha>1-1/N$, then $\mu<\inf\sigma_{\rm ess}(H^\alpha_\Psi(N))$ and both $\Psi$ and $\nabla\Psi$ decay exponentially at infinity.
\end{lemma}

\begin{proof}
We have already explained in the proof of Theorem \ref{thm:nonlinear} that, by the concavity of $F$, $E_\alpha(N)$ is also the lowest energy over all mixed states of $L^2_a((\R^3)^N)$. In particular it holds
$$\cE_\alpha\bigg((1-t)|\Psi\rangle\langle\Psi|+t|\Psi'\rangle\langle\Psi'|\bigg)\geq E_\alpha(N)$$
for all $\Psi'\in H^1_a((\R^3)^N)$ and all $0\leq t\leq1$. The first order in $t$ provides the bound
$\pscal{\Psi',H^\alpha_\Psi(N)\Psi'}\geq \pscal{\Psi,H^\alpha_\Psi(N)\Psi}$, showing that $\mu$ is the first eigenvalue of $H^\alpha_\Psi(N)$.
The Hamiltonian $H^\alpha_\Psi(N)$ is a usual Coulomb Hamiltonian of $N$ electrons with an external Coulomb field of total charge $Z=\alpha \int_{\R^3}\rho_\Psi=\alpha N$. It was shown by Zhislin and Sigalov \cite{Zhislin-60,ZhiSig-65} that $\mu$ is an isolated eigenvalue as soon as $N<Z+1=\alpha N+1$. The exponential decay follows from well-known results reviewed for instance in Section XIII.11 of \cite{ReeSim4}.
\end{proof}

Let us now assume that $E_\alpha(N-k)$ and $E_\alpha(k)$ have respective ground states $\Psi_1$ and $\Psi_2$, and that $\alpha\geq1$. We want to prove that $E_\alpha(N)<E_\alpha(N-k)+E_\alpha(k)$. Using their exponential decay, we can replace  $\Psi_1$ and $\Psi_2$ by functions with support in a ball of radius $R$, making an error in the energy of the form $e^{-aR}$. For the sake of simplicity we do not change our notation and assume that
$$\cE_\alpha(\Psi_1)\leq E_\alpha(N-k)+Ce^{-aR},\qquad \cE_\alpha(\Psi_2)\leq E_\alpha(k)+Ce^{-aR}.$$

When $\alpha>1$, we can take advantage of a Coulomb attraction at infinity and choose as trial function
$$\Psi_R^{U,V}:=\Psi_1^U\wedge\Psi_2^V(\cdot-3R\vec{v})$$
for rotations $U,V\in SO(3)$ and with rotated ground states $\Psi_j^U=\Psi_j(U^{-1}\cdot)$. Averaging over the rotations $U,V\in SO(3)$ and using Newton's theorem yields a bound
$$\int_{SO(3)}dU\int_{SO(3)}dV\;\cE_\alpha(\Psi_R^{U,V})\leq E_\alpha(N-k)+E_\alpha(k)-\frac{(N-k)k(1-\alpha)}{3R}+Ce^{-aR}.$$
This shows the binding inequality when $\alpha>1$.

When $\alpha=1$ there is \emph{a priori} no simple binding in $1/R$. Fortunately, there always exists a Van Der Waals force between two multi-polarons. Following a method of Lieb and Thirring \cite{LieThi-86}, we take as trial state
$$\Psi_R^{U,V}:=\Psi_1^U\wedge\Psi_2^V(\cdot-3R\vec{v})+\lambda\left\{\left(\mathbf{m}\cdot\sum_{j=1}^{N-k}\nabla_j\right)\Psi^U_1\right\}\wedge\left\{\left(\mathbf{n}\cdot\sum_{j=1}^{k}\nabla_j\right)\Psi^V_2(\cdot-3R\vec{v})\right\}.$$
Writing with an obvious convention $\Psi_R^{U,V}=\Phi_R^{U,V}+\lambda\tilde{\Phi}_R^{U,V}$, we have 
$$\int_{\R^3}dx_2\cdots\int_{\R^3}dx_N\;\overline{{\Phi}_R^{U,V}}\tilde{\Phi}_R^{U,V}=0$$
which is seen by using that $\Psi_1^U$ and $\Psi_2^V(\cdot-3R\vec{v})$ have disjoint supports, as well as the fact that $\Psi_1^U$ is orthogonal to $\left(\mathbf{m}\cdot\sum_{j=1}^{N-k}\nabla_j\right)\Psi^U_1$ and a similar property for $\Psi_2^V$.
As was already mentioned in \cite{LieThi-86}, this yields
$\|\Psi_R^{U,V}\|^2=1+O(\lambda^2)$, but this also gives 
\begin{equation}
\rho_{\Psi_R^{U,V}/\norm{\Psi_R^{U,V}}}=\rho_{\Psi_1^U}+\rho_{\Psi_2^V}(\cdot-3R\vec{v})+O(\lambda^2).
\label{eq:density_Van_Der_Waals} 
\end{equation}
Therefore we can mimic the argument of \cite{LieThi-86} and obtain an upper bound of the form
$$\int_{SO(3)}dU\int_{SO(3)}dV\;\cE_\alpha\left(\Psi_R^{U,V}/\norm{\Psi_R^{U,V}}\right)\leq E_1(N-k)+E_1(k)+a\frac{\lambda}{R^3}+b\lambda^2+Ce^{-aR}.$$
The linear term in $\lambda$ comes from the cross-term between the two functions appearing in the definition of $\Psi_R^{U,V}$, in the electron-electron interaction term. This term is exactly the same as the one calculated in \cite{LieThi-86}. The nonlinear term involving the density only provides a $O(\lambda^2)$ by \eqref{eq:density_Van_Der_Waals}. Taking $\lambda=-a/2bR^3$ yields the desired attractive Van Der Waals interaction potential $-C/R^6$, hence the binding of two polaron systems when $\alpha=1$. This ends the proof of Theorem \ref{thm:polaron}.\hfill\qed
\end{proof}


\end{document}